\documentclass[aps,prx, superscriptaddress, twocolumn,amsmath,amssymb,longbibliography]{revtex4-1}
\usepackage{amsmath,amssymb,amsfonts,bm}
\usepackage{graphicx}
\usepackage{xcolor}
\usepackage{bbold}
\usepackage{graphicx}
\usepackage{dcolumn}
\usepackage{epstopdf}
\usepackage{epsfig}
\usepackage{url}
\usepackage{hyperref}
\usepackage{verbatim}
\usepackage[export]{adjustbox}
\usepackage{scalerel}
\usepackage{braket}
\usepackage[normalem]{ulem}
\usepackage{amsthm}

\usepackage{simpler-wick}
\usepackage{IEEEtrantools}

\usepackage{enumitem}

\usepackage{tensor}

\usepackage{adjustbox}
\usepackage{float}

\newtheorem{lemma}{Lemma}[section]
\newtheorem{lemmaalt}{Lemma}
\newtheorem{definitionalt}{Definition}
\newtheorem{corollary}{Corollary}[lemma]

\newcommand{\tmax}{t_\mathrm{max}}
\newcommand{\NN}{N}
\newcommand{\Nhil}{\mathcal{N}}
\newcommand{\Abar}{\overline{A}}

\newcommand{\MM}{\mathcal{M}}

\newcommand{\half}{\mathrm{half}}

\newcommand{\ucoe}{u_{\mathrm{COE}}}

\newcommand{\uu}{\mathcal{U}}

\newcommand{\Tr}{\operatorname{Tr}}

\newcommand{\tr}{\operatorname{tr}}

\newcommand{\COE}{\mathrm{COE}}
\newcommand{\be}{\begin{equation}}
\newcommand{\ee}{\end{equation}}
\newcommand{\ba}{\begin{aligned}}
\newcommand{\ea}{\end{aligned}}

\newcommand{\thei}{t_{\mathrm{Hei}} }
\newcommand{\Lth}{L_{\mathrm{Th}} }
\newcommand{\tth}{t_{\mathrm{Th}} }

\newcommand{\TT}{\mathcal{T}}

\newcommand{\gqmbcs}{gMBQCs}

\usepackage{booktabs}

\newcommand{\wg}{\mathrm{Wg}}

\newcommand{\iden}{\mathbb{1}}

\newcommand{{\fermswap}}{FSG}

\newcommand{\Mod}[1]{\ (\mathrm{mod}\ #1)}

\graphicspath{{./Figures/}{./Figures_numerics/}}

\begin{document}

\title{Many-body quantum chaos and time reversal symmetry}

\author{Weijun Wu}
    \affiliation{Department of Chemistry, Princeton University, Princeton, New Jersey, 08544, USA}

\author{Saumya Shivam}
    \affiliation{Department of Physics, Princeton University, Princeton, New Jersey, 08544, USA}

\author{Amos Chan}
    \affiliation{Department of Physics, University of Warwick, Coventry, CV4 7AL, United Kingdom}
    \affiliation{Physics Department, Lancaster University, Lancaster, LA1 4YB, United Kingdom}
\date{\today}

%

\begin{abstract}
We investigate universal signatures of quantum chaos in the presence of time reversal symmetry (TRS) in generic many-body quantum chaotic systems ({\gqmbcs}). 
We study three classes of minimal models of {\gqmbcs} with TRS, realized through random quantum circuits with (i) local TRS, (ii) global TRS, and (iii) TRS combined with discrete time translation symmetry. 
In large local Hilbert space dimension $q$, we derive the emergence of random matrix theory (RMT) universality in the spectral form factor (SFF) at times larger than Thouless time $t_{\mathrm{Th}}$, which diverges with system sizes in {\gqmbcs}.
At times before $t_{\mathrm{Th}}$, we identify universal behaviour beyond RMT by deriving explicit scaling functions of SFF in the thermodynamic limit.
In particular, in the simplest non-trivial setting -- preserving global TRS while breaking time translation symmetry -- we show that the SFF is mapped to the partition function of an emergent classical ferromagnetic Ising model, where the Ising spins correspond to the time-parallel and time-reversed pairings of Feynman paths, and external magnetic fields are induced by TRS-breaking mechanisms.  
Without relying on the large-$q$ limit, we develop a second independent derivation of the Ising scaling behaviour of SFF using space-time duality and parity symmetric non-Hermitian Ginibre ensembles.
Moreover, we show that many-body effects originating from time-reversed pairings of Feynman paths manifest in the two-point autocorrelation function (2PAF), the out-of-time-ordered correlator (OTOC), and the partial spectral form factor -- quantities sensitive to both eigenvalue and eigenstate correlations.
Unlike the case of SFF, we find that 2PAF generically and asymetrically favours time-reversed over time-parallel pairings in the presence of TRS, resulting in characteristic negative and suppressed values of 2PAF for antisymmetric operators, compared to their symmetric counterparts.
By summing 2PAF over a complete operator basis, we obtain a third complementary approach to derive the universal Ising scaling behaviour of SFF in {\gqmbcs} with global TRS.
Additionally, we establish that the fluctuations of 2PAF are governed by an emergent three-state Potts model, leading to an exponential scaling with the operator support size, at a rate set by the three-state Potts model, a distinctive signature of many-body quantum chaos with TRS.
Lastly, we show that the leading order behaviour of the OTOC between two operators is insensitive to the presence of TRS when the operators are spatially separated, but becomes sensitive to TRS when their supports overlap at the initial time.
We demonstrate the general applicability of the results by numerically simulating three quantum circuit models and a time-independent Hamiltonian model.
\end{abstract}

\maketitle

\tableofcontents

\section{Introduction}

Understanding the chaotic dynamics of generic many-body quantum chaotic systems ({\gqmbcs}) is a notoriously difficult but fundamental problem with practical implications in exploring new computational paradigms~\cite{feynman1982simulating, Boixo_2018, arute2019quantum, Preskill2018quantumcomputingin}.
A vast class of interacting many-body systems is believed to exhibit quantum chaotic behaviour, as encapsulated by the eigenstate thermalization hypothesis (ETH)~\cite{deutsch1991quantum, Srednicki, Rigol2008}, which provides a framework for understanding how closed quantum systems reach thermal equilibrium.
A key diagnostic of quantum chaos is the emergence of random matrix theory (RMT) behaviour, as stated in the \textit{quantum chaos conjecture}  \cite{bohigas1984characterization}: 
A quantum system is considered chaotic if its spectral statistics, at sufficiently small energy scales, resemble those predicted by RMT.
This connection is crucial in theoretical physics, as RMT provides a framework that abstracts away microscopic details, capturing the universal properties of chaotic systems based solely on their symmetries \cite{Mehta, Haake}.
However, RMT fails to capture the structure of local interactions of many-body quantum systems, which gives rise to complex correlation in the Fock space. 
To address this limitation, random quantum circuits have been introduced as powerful toy models in quantum information and many-body physics, providing a versatile framework to explore universal features of strongly coupled quantum dynamics~\cite{Harrow_2009, nahum2017}. Recent explorations involving the use of random quantum circuits to examine operator growth \cite{nahum2018, vonKeyserlingk2017a, keyserlingk2018, Huse2017} and entanglement dynamics amid chaotic evolution and measurements~\cite{nahum2017, Skinner_2019, Li_2018, chan2019, altman2019, Jian_2020, Zabalo_2020}, spectral statistics~\cite{Kos_2018, chan2018solution, chan2018spectral, bertini2018exact}, and eigenstate correlations~\cite{cdc3, Garratt2021prx}.

\begin{figure}[ht]
    \centering
 \includegraphics[width=0.4 \textwidth]{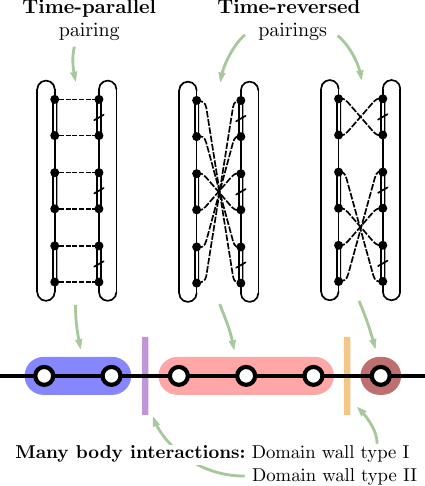}
    \caption{\textbf{Many-body interactions between time-parallel, time-reversed pairings of Feynman paths } give rise to universal scaling behaviours of SFF in the presence of TRS before the many-body Thouless time. 
    There are two types of interactions: Domain wall type I  arising from interactions among two time-parallel pairings or time-reversed pairings, and Domain wall type II  arising from interactions between a time-parallel pairing and a time-reversed pairing.
    }
    \label{fig:dw}
\end{figure}

A quantum system described by Hamiltonian $H(t)$ at time $t$ is \textit{time reversal symmetric} (TRS) if $H(t) = \TT H(-t) \TT^{-1}$, where $\TT$ is an antiunitary operator. 
Equivalently, a wavefunction $\psi(t)$ obeying the Schr\"{o}dinger equation of $H$ with TRS, i.e. $i \hbar \partial_t{\psi}(t) = H \psi (t)$,  has time-reversed counterpart given by $\TT \psi(-t)$. 
It can be shown that the antiunitary operator satisfies $\TT^2=\pm 1$, where the plus and minus signs correspond to the cases of systems with TRS with integer and half-integer spins respectively. Here we focus on dynamical and spectral properties of TRS quantum many-body systems with $\TT^2 = 1$. The case for $\TT^2= -1$ will be addressed in future work.
TRS plays a fundamental role in determining dynamical and spectral properties in quantum systems. It underlies key phenomena such as Kramers degeneracy in electronic systems~\cite{kramers1930theorie}, protected edge modes in topological insulators~\cite{kane2005, kane2005b, Bernevig_2006}, weak localization in disordered conductors~\cite{altshuler1980magnetoresistance, lee1985disordered, AltshulerShklovskii},  distinct spectral statistics in quantum chaotic systems and random matrices~\cite{dyson1962statistical, Zirnbauer_1996}, the CPT theorem~\cite{schwinger1951cpt, Luders1954, pauli1955cpt}, the problem of the arrow of time~\cite{boltzmann1872h, eddington1928nature, zeh2007direction}. 
Further, many of the most widely studied paradigmatic spin chain models possess TRS in their simplest formulations. Notable examples include the Heisenberg and XXZ chain~\cite{takahashi1999thermodynamics}, the transverse field Ising model~\cite{sachdev2011quantum}, and the AKLT model~\cite{affleck1987rigorous}, all of which have played important roles in understanding quantum magnetism, critical phenomena, and topological phases.
In this paper, we consider three classes of {\gqmbcs} with TRS: (i) local TRS, where each local interaction within the system is individually TRS; (ii) global TRS, where the entire unitary time evolution of the system is TRS; and (iii) TRS with discrete time translation symmetry (Floquet TRS). Note that a system can exhibit local or global TRS independently. However, a Floquet system with TRS  inherently satisfies both local and global TRS due to the periodic structure of its time evolution.
The objective of this paper is to understand the universal signature of quantum chaos arising from the presence of TRS and many-body interactions.

\begin{figure}[ht]
    \centering
    \includegraphics[width=0.45\textwidth]{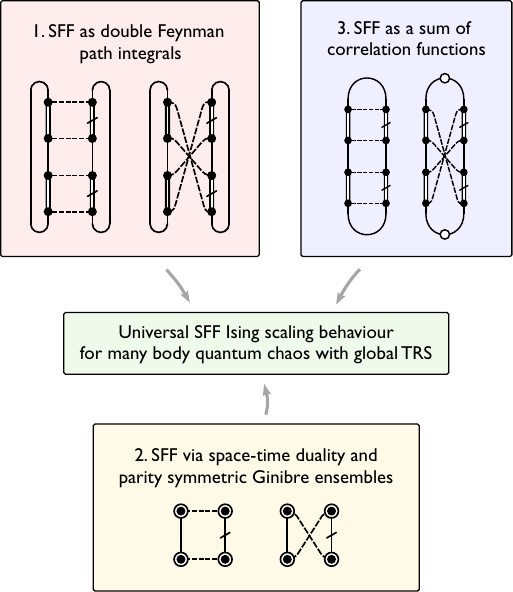}
     \caption{\textbf{Three routes to emergent Ising scaling behaviour of SFF in {\gqmbcs} with global TRS} by evaluating (1) SFF as double Feynman path integrals;  (2) SFF via space-time duality and parity symmetric Ginibre ensembles; and (3) SFF as a sum of correlation functions. See main text and Table~\ref{Tab:sff_3approaches}
for the corresponding effective Ising degrees of freedom with each approach.}
    \label{fig:three_routes}
\end{figure}

The quantum chaos conjecture can be concretely probed by the \textit{spectral form factor} -- the Fourier transform of two-level spectral correlation function \cite{Mehta, Haake}. 
Given quantum system with a time evolution operator $U(t)$, the SFF is defined as $K(t)=|\Tr[U(t)]|^2$, representing a double sum of amplitudes over Feynman paths. At the core of quantum chaos lies the diagonal approximation~\cite{Hannay1984PeriodicOA, berry1985semiclassical}, originally introduced in periodic orbit theory, which asserts that under an ensemble average, the dominant contributions to the SFF arises from pairings between a Feynman path and its complex conjugate counterpart, as phase cancellations suppress off-diagonal terms.  Additionally, time translational symmetry allows for $t$ pairing choices, leading to the characteristic linear ramp behaviour of RMT, a universal hallmark of quantum chaos~\cite{Mehta, Haake}. In the presence of TRS, each Feynman path has a corresponding time-reversed path. This introduces an additional class of pairings: Besides \textit{time-parallel pairings} between two paths evolving in the same direction, there exist \textit{time-reversed pairings} between two paths propagating in opposite directions. As a result, the leading order SFF pairings double from $t$ to $2t$ in the presence of TRS. 
The physics of time-reversed pairings has been extensively studied across various fields, e.g. Cooperons in disordered systems~\cite{langerneal_1966, lee1985disordered, AltshulerShklovskii}, and semiclassical periodic orbit theory in single- or few-particle physics~\cite{Hannay1984PeriodicOA, berry1985semiclassical, Sieber_2001, muller_2004, muller_2005}.

What are the signatures of quantum chaos in the presence of \textit{many-body} interactions?
Due to local many-body interaction, the connected part of SFF exhibits  a characteristic ``bump'' --  a generic deviation from RMT -- at times earlier than the \textit{many-body Thouless time} $\tth$~\cite{chan2018spectral, Prosen}. Note that the disconnected part of the SFF can introduce an early-time non-universal ``dip'', which should not be confused with the ``bump'' in the connected SFF.
In generic many-body systems, both the bump size and $\tth$ have been observed -- and in some cases analytically derived -- to grow with system size, ultimately diverging in the thermodynamic limit~\cite{chan2018spectral, friedman2019, Gharibyan_2018, moudgalya2020spectral, chan2020lyap, chaoschallengeMBL, chan2021trans, Garratt2021prx, yoshimura2023operator, Dag2023}. 
This divergence suggests the existence of an extended time and energy window in which one can explore novel signatures of many-body quantum chaos beyond the conventional RMT behaviour.
Such deviation from RMT can be understood through the pairings of Feynman paths, which are now many-body trajectories in the Fock space. 
In {\gqmbcs}, these pairings can occur locally, so that different subregions of the system can exhibit distinct pairings~\cite{chan2018spectral, chan2020lyap, garratt2020MBL}. The many-body interactions between these local pairings then give rise to the characteristic bump in the SFF of {\gqmbcs}.
Before $\tth$, the system can be heuristically viewed as a collection of expanding patches of random matrices, each adopting an identical pairing of Feynman paths. As time progresses, these patches grow and merge, eventually encompassing the entire system. For times beyond $\tth$, the system behaves as a single large random matrix, with all regions adopting the same pairing, thereby restoring the diagonal approximation.

In this paper, we show that many-body interactions between these time-reversed and time-parallel Feynman-path pairings give rise to novel, universal signatures of generic many-body quantum chaos with TRS, not only in the SFF, but also in dynamical probes such as two-point correlation functions and out-of-time-ordered correlators, which involve eigenstate correlations as well as spectral correlation.
Sec.~\ref{sec:res_overview} gives an overview and summary of our main results. Sec.~\ref{sec:models} defines three quantum circuit models, a time-independent Hamiltonian model, and the corresponding RMT models. 
In Sec.~\ref{sec:diag}, we develop the diagrammatic approach for random quantum circuits with TRS. We provide results on the SFF, partial SFF, two-point correlation functions, and out-of-time-ordered correlators in Sec.~\ref{sec:sff}, ~\ref{sec:psff}, ~\ref{sec:corr}, and ~\ref{sec:otoc} respectively. 
Finally, Sec.~\ref{sec:outlook} offers a discussion and an outlook for future work.
\newpage
\section{Results overview}\label{sec:res_overview} 
Here we provide an overview of our key results on universal signatures of generic many-body quantum chaos with TRS, as probed by the SFF, partial SFF, two-point correlation functions, and out-of-time-ordered correlators.  
To isolate the role of TRS in the SFF, we begin by considering the simplest non-trivial setting where all symmetries, including time translational symmetry, are eliminated except the global TRS. In this minimal setting, there are only two possible local pairings of Feynman paths: a single time-parallel pairing and a single time-reversed pairing [Fig.~\ref{fig:sff_diag}]. 
In the limit of large local Hilbert space dimension $q$, the SFF of {\gqmbcs} with global TRS can be exactly evaluated by mapping it to the partition function of the Ising model, where the time-parallel and time-reversed pairings of Feynman paths serve as effective Ising degrees of freedom at each physical site, and where a TRS-breaking mechanism can be introduced as an effective external magnetic field.
To establish the universality of our findings, we derive a scaling function of SFF [Eq.~\eqref{eq:ising_scaling}] in the \textit{Thouless scaling limit} where both $L$ and $t$ are large but $x \equiv L / \Lth(t)$ remains fixed \cite{chan2021trans}. Here, the Thouless length $\Lth(t)$ is the inverse function of the Thouless time $\tth(L)$.
We validate this  scaling function through numerical simulations of two distinct many-body quantum circuit models at finite $q$, finding excellent agreement with our theoretical predictions obtained at large $q$ [Fig.~\ref{fig:sff_gtrs_num}]. The Ising scaling behaviour of SFF in {\gqmbcs} with global TRS will be derived with two complementary approaches, as described below.
 Restoring the discrete time translational symmetry in addition to TRS, we show that the SFF in {\gqmbcs} can be mapped onto the partition function of a generalized Potts model where each physical site possesses $2t$ degrees of freedom -- $t$  time-parallel Feynman-path pairings, and $t$ time-reversed pairings [Fig.~\ref{fig:sff_diag} and Fig.~\ref{fig:dw}]. 
This mapping gives rise to two distinct types of many-body interactions: (i) Domain wall type I, describing interactions between either two time-parallel pairings, or two time-reversed pairings; and (ii) Domain wall type II, describing interactions between time-parallel and time-reversed pairings [Fig.~\ref{fig:dw}]. 
From these interactions, we derive an exact expression for SFF in the large-$q$ limit, and obtain the corresponding scaling function at the Thouless scaling limit [Eqs.~\eqref{eq:ftrs_obc_scaling} and ~\eqref{eq:ftrs_pbc_scaling}]. To validate these results at finite $q$, we numerically simulate two distinct quantum circuit models, finding good agreement with our theoretical predictions [Fig.~\ref{fig:sff_ftrs_num}].
Together,  we show that in the large-$q$ limit, depending on the types of TRS symmetries, the SFF of {\gqmbcs} is schematically given by
\be\label{eq:sff_overview}
\ba
\lim_{q\to\infty} \overline{K(t,L)}  
\propto 
\begin{cases}
     Z^{t\text{-Potts}}   \qquad & \text{No symmetries,}
    \\
      1 \qquad & \text{Local TRS,}
     \\
     Z^{\text{Ising}} \qquad & \text{Global TRS,}
    \\
     Z^{2t\text{-gPotts}} \qquad & \text{Floquet TRS,}
\end{cases} 
\ea
\ee
where the partition function of Ising, Potts, and generalized Potts models (gPotts) are defined in Eq.~\eqref{eq:ising_model}, and Eqs.~\eqref{eq:floq_trs_mapping} to \eqref{eq:dw_typ2}.  
These mappings will be made precise below and exactly evaluated in the one-dimensional settings. The features of SFF in the higher-dimensional case are studied in an upcoming work~\cite{upcoming}.

 In quantum many-body physics, space-time duality is an approach that provides insights into a quantum many-body system by analyzing the spatial propagation of a circuit instead of its conventional unitary time evolution~\cite{guhr2016kim, chan2018spectral, bertini2018exact, chan2020lyap, lerose2020influence, Ippoliti_2022}. 
Recent studies have demonstrated that the dual transfer matrix of spatially-extended, generic many-body quantum chaotic systems belongs to the universality class of the non-Hermitian Ginibre ensemble~\cite{Shivam_2023},  which consists of random matrices with independent complex Gaussian matrix elements.
Using Ginibre ensembles, we independently derive the Ising scaling function of the SFF in the presence of global TRS [Eq.~\eqref{eq:ginibre_scaling_form}], without relying on the infinite-$q$ limit or specific quantum circuit models, thereby substantiating the claims on the universality of our findings.

We derive exact results for the partial spectral form factor (PSFF) \cite{Gong_2020, Garratt2021prx, ZollerSFF2021}, which generalizes the SFF by tracing out only a subregion $A$ of the quantum many-body system. Notably, beyond capturing eigenvalue correlations, the PSFF also encodes information about the eigenstate correlations, which serve as robust diagnostics of quantum chaos, as exemplified by the Eigenstate Thermalization Hypothesis~\cite{deutsch1991quantum, Srednicki, Rigol2008}. Analogous to Eq.~\eqref{eq:sff_overview}, in the large-$q$ limit, PSFF can be evaluated  via a mapping to the partition function of a system [Eq.~\eqref{eq:psff_overview}], with different Hamiltonians governing subsystem $A$, its complement $\overline{A}$, and their boundary $\partial A$. 
We analytically and numerically demonstrate that PSFF features a bump before decaying  exponentially in time after the corresponding Thouless time. The origin of this bump is analogous to that in the SFF, but in this case, the emergent statistical mechanical system is confined in the subregion $A$, within which time-parallel and time-reversed pairings of Feynman paths interact.

The many-body interactions between time-parallel and time-reversed pairings have implications on the dynamical signatures of the {\gqmbcs} with TRS. 
Two-point correlation functions are one of the simplest observables to characterise  quantum fluctuations~\cite{rickayzen1980greens}, and are expected to decay exponentially in {\gqmbcs} in the absence of conserved quantities~\cite{Rigol}. 
Consider the Hilbert space $\mathbb{C}^{q^L}$ of a quantum many-body system with $L$ sites, where each local Hilbert space is $\mathbb{C}^q$. At each site, we use the generalized Gell-Mann matrices -- explicitly defined in Eq.~\eqref{eq:gellmann} -- as operator basis $\{o_{\mu }\}$ with index $\mu =0,1,\dots, q^2-1$. This basis satisfies the follow properties: $o_{0}\equiv \mathbb{1}$, $o_{\mu }$ is traceless for $\mu \neq 0$, and the basis is orthonormal, i.e. $q^{-1} \Tr[o_{\mu} o_{\nu}] = \delta_{\mu \nu}$.  For $q=2$, this basis consists of Pauli matrices and the identity matrix. 
For the many-body Hilbert space, we take operator strings $O_{{\mu}} =\bigotimes_{i=1}^L o_{\mu_i}$  to form the basis states with $O_{0}\equiv \mathbb{1}$ and the integer $\mu = 0, 1,2,\dots, q^{2L}-1$ is now promoted to a vector $(\mu_1,\mu_2, \dots,\mu_L)$ with $\mu_i =0,1,\dots, q^2-1$. 
The two-point correlation function of the operator $O_\mu$ at time $t$ at infinite temperature is defined as $C_{\mu \nu}(t) := \mathcal{N}^{-1}\Tr[O_\mu(t) O_\mu(0)]$ where $O_\mu(t) = U(t) O_\mu U^\dagger(t)$ is the time-evolved operator in the Heisenberg picture. The \textit{two-point autocorrelation function} (2PAF) is defined as  $C_{\mu \mu }(t)$.

We show that averaged 2PAF can be evaluated by recognizing that   time-reversed pairings of Feynman paths dominate at sites where  $O_{\mu}$ has  operator support, while time-parallel pairings dominate on sites without  support [Eq.~\eqref{eq:2paf_partition}].
Crucially, unlike in the case of the SFF, time-parallel and time-reversed pairings are \textit{not} treated on equal footing in the 2PAF.
At site $i$ with non-trivial operator support, a unique time-reversed pairing of Feynman paths connects the two local operators at leading order in $q$, effectively circumventing the tracelessness condition [Lemma~\ref{lemma:leading2pcf} and Fig.~\ref{fig:2pcf_diag} (a)]. 
We represent this local pairing heuristically as   
\be \label{eq:2paf_cont_ill_v1}
\ba
\Tr[ 
\wick{
\c1 o_{\mu_i}(t) \c1 o_{\mu_i}(0) 
}]
 \, , \quad \text{TRS,} 
\ea
\ee
where $o_{\mu_i}$ is the support of the operator $O_\mu$ at site $i$ with $\mu_i=1,2,\dots, q^2-1$ . 
This contraction gives rise to a crucial sign that is dependent on the symmetricity of the local operator: a negative (positive) sign arises for antisymmetric (symmetric) operator [Fig.~\ref{fig:op_dependence}], e.g. Pauli-$Y$ (Pauli-$X$) operator at $q=2$. 
Further, the other subleading diagrams from time-reversed pairings and time-parallel pairings are  always positive without operator dependence. 
Consequently, 2PAF for antisymmetric operator is expected to have suppressed values relative to the symmetric counterparts, since the negative leading contributions are reduced in magnitude by the positive subleading contributions.  
The above behaviour of 2PAF contrasts sharply with the case without TRS, where the leading order contributions lack operator-dependent signs and are suppressed by an additional factor of $q$ for each site with operator support. 
By summing over 2PAF over a complete operator basis and properly accounting for the operator-dependent sign, we rederive the Ising scaling behaviour of SFF in {\gqmbcs} with global TRS [Eqs.~\eqref{eq:ising_scaling_2paf} and \eqref{eq:3approaches}]. 
The dominance of time-reversed pairings and the characteristic negative and suppressed values of the 2PAF of antisymmetric operators (compared to 2PAF of symmetric operators) are supported by two pieces of numerical evidence at finite \( q \). First, direct simulations of the 2PAF for a local antisymmetric operator show a clear negative and suppressed value, as illustrated in Fig.~\ref{fig:2pcf_numerics_antisym} for two Floquet models and a time-independent Hamiltonian model. Second, upon summing over both symmetric and antisymmetric operators supported on the same region, the 2PAF exhibits a scaling behaviour with operator support size that aligns with theoretical predictions, as shown in Fig.~\ref{fig:2pcf_numerics}.

\begin{table*}[ht]
\caption{\label{Tab:sff_3approaches} Effective Ising degrees of freedom in the three approaches used to derive the scaling behaviour of SFF in {\gqmbcs} with global TRS. See also \ref{fig:three_routes} for a schematic illustration.}
\begin{ruledtabular}
\begin{tabular}{lll}
Approach & Effective Ising degrees of freedom 
 \\
  \hline  
1. SFF as double Feynman path integrals  & Time-parallel pairings of  Feynman paths 
  \\
 &
Time-reversed pairings of Feynman paths 
 \\
  \hline  
2. SFF via space-time duality and & 
Identity contraction in the dual spatial direction
\\
\hspace{0.3cm} parity symmetric Ginibre ensembles
 & Parity SWAP contraction in the dual spatial direction
    \\
 \hline
3. SFF as a sum over correlation functions &
Time-parallel pairings for 2PAF on sites w/o operator support
\\
& 
Time-reversed pairings for 2PAF on sites w/ operator support
 \end{tabular}
\end{ruledtabular}
\end{table*}

Further, we demonstrate that the averaged 2PAF fluctuations or 2PAF second moment can be evaluated via a mapping onto the partition function of the three-state Potts model, where each site hosts three effective degrees of freedom arising from time-reversed and time-parallel pairings of Feynman paths. In the notation introduced above, these local degrees of freedom are [Lemma~\ref{lemma:2paf_fluc_v2}]
\begin{IEEEeqnarray}{ll} 
\Tr[
\wick{
\c1 o_{\mu_i}(t) \c1 o_{\mu_i}
} ]
\wick{
\Tr[
\c1 o_{\mu_i}(t) \c1 o_{\mu_i}
}
]  \, , \quad &\text{TRS,} \nonumber
\\   
\Tr [  
\wick{
\c2 o_{\mu_i}(t) \c1 o_{\mu_i}
]
\Tr [
\c1 o_{\mu_i}(t) \c2 o_{\mu_i}
}
]  
\, , \quad &\text{TRS,} \quad
\label{eq:2paf_cont_ill_v2}
\\
\Tr [
\wick{
\c1 o_{\mu_i}(t) \c2 o_{\mu_i}
]
\Tr[ 
\c1 o_{\mu_i}(t) \c2 o_{\mu_i}
}
]  
\, , \quad &\text{No symmetries.} \nonumber
 \end{IEEEeqnarray}
The emergence of three distinct states is \textit{fundamental}, arising from the combinatorial fact that there are three possible ways to pair four operators in the fluctuation of the 2PAF.
Crucially, in the absence of TRS, only the last pairing in Eq~\eqref{eq:2paf_cont_ill_v2} remains, leading to a striking contrast in the behaviours of 2PAF fluctuations: With TRS, the normalized fluctuations scale exponentially with operator support size at a rate governed by the three-state Potts model; while without TRS, the normalized fluctuations remain approximately constant as the operator support size grows. We support these claims with numerical simulations two quantum circuit models and a time-independent Hamiltonian model at finite-$q$ in Figs.~\ref{fig:2pcf_numerics_fluc_in_op_size} and \ref{fig:2pcf_numerics_fluc_growth_rate}.
Schematically, we  summarise the results for 2PAF fluctuations as
\be \label{eq:2paf_fluc_summary}
\ba
\lim_{q\to\infty} \overline{C_{\mu \mu}^2(t) } 
\propto 
\begin{cases}
     Z_A^{\text{cluster}}  \qquad & \text{No symmetries,}
     \\
     Z_A^{3\text{-Potts}} \qquad & \text{TRS.}
\end{cases} 
\ea
\ee
$Z_A^{3\text{-Potts}}$ is the partition function of an emergent statistical mechanical model \eqref{eq:Z_A} where the operator is supported in subregion $A$ with the emergent dynamics governed by the three-state Potts model. 
$Z_A^{\text{cluster}}$ is the partition function of a trivial statistical mechanical system described by a single state with Boltzmann weight depending on size of the operator support $|O|$ and its boundary $|\partial O|$.
 %
%
These results on 2PAF and its fluctuations highlight the significant role of the many-body interactions among time-reversed pairings in {\gqmbcs} with TRS.

Under the quantum dynamics in generic quantum many-body systems, information of local perturbation spreads and scrambles among non-local degrees of freedom of the system, and can be diagnosed using the \textit{out-of-time-ordered correlator} (OTOC) \cite{Shenker_2014, LarkinOvchinnikov}, defined at infinite temperature as $F_{\mu \nu}(t):=\frac{1}{\mathcal{N}}\Tr[O_\mu (t) O_\nu(0) O_\mu(t) O_\nu(0)]$, where  $O_\mu(t) = U(t) O_\mu U^\dagger(t)$ is the time-evolved operator in the Heisenberg picture as before. 
We show that the leading order behaviour of the OTOC between two operators is unaffected by TRS when the operators are spatially separated. However, when their supports overlap at the initial time, the OTOC becomes sensitive to TRS. In this case, TRS induces corrections arising from contractions analogous to those in Eqs.~\eqref{eq:2paf_cont_ill_v2} and ~\eqref{eq:2paf_fluc_summary}, with the sign of the correction depending on whether the local operator supports are identical.

TRS naturally arises in some of the simplest models of quantum many-body chaos studied in the literature. Recent works on TRS in spatially-extended many-body systems have explored entanglement and operator spreading in temporal and spatial random local TRS quantum circuits that lack global TRS and time translational symmetry~\cite{hunter2018operator, Kalsi_2022}. Additionally, the RMT behaviour has been derived in SFF in Floquet TRS chaotic systems with long-range interactions~\cite{Kos_2018, Roy_2020fermion, Roy_2022boson, kumar2025leadingleadingorderspectralform} and  Floquet TRS chaotic quantum circuits with dual unitary condition~\cite{bertini2018exact, Flack_2020, Bertini_2021du}. In dual unitary circuits, spatial propagations are generated by unitary operators, so that their spatial propagation can be treated as effective time evolution.
Consequently, dual unitary circuits exhibit RMT behaviour in SFF after only $O(1)$ time, and do so without a diverging many-body bump region discussed here. Additionally, their two-point correlation functions exhibit strictly confined to the light cone, vanishing exactly within the light cone. By contrast, this work investigates the spectral and dynamical properties of generic quantum many-body chaotic systems with TRS, which do not exhibit these features.

While completing this manuscript, \cite{khanna2025randomquantumcircuitstimereversal} appeared in the arXiv on random quantum circuits with global TRS. These two works are complementary to each other. \cite{khanna2025randomquantumcircuitstimereversal} focuses on the measurement-induced phase transition, while this work focuses on the SFF, PSFF, correlation functions, and OTOC.

 The main results of this paper are summarized as follows, with exact analytical results obtained in the limit of large local Hilbert space dimension $q$ unless stated otherwise: 
    \begin{itemize}
    \item Derivation of the leading order RMT universal behaviour in the spectral form factor in {\gqmbcs} with TRS after the Thouless time $\tth$ [Sec.~\ref{sec:sff_manybody}].
    \item Derivation of Ising  and generalized Potts universal scaling behaviours beyond the RMT in the spectral form factor in  {\gqmbcs} with TRS before the Thouless time $\tth$ [Eqs.~\eqref{eq:ising_scaling} and ~\eqref{eq:ftrs_obc_scaling}  with numerical evidences from two quantum circuit models in Fig.~\ref{fig:sff_gtrs_num} (c) and Fig.~\ref{fig:sff_ftrs_num} (c) respectively].
    \item A second derivation of the Ising scaling behaviour of the spectral form factor from non-Hermitian parity symmetric Ginibre ensembles of {\gqmbcs} without relying on the large-$q$ limit [Eq.~\eqref{eq:ginibre_scaling_form}]. 
    \item Exact results for the partial spectral form factor in {\gqmbcs} with TRS beyond the RMT regime [Eq.~\eqref{eq:psff_overview}]. 
    \item Exact results for the two-point autocorrelation functions in {\gqmbcs} with TRS [Eq.~\eqref{eq:2paf_partition}], which are  dominated by local time-reversed pairings, resulting in characteristic negative and suppressed values for autocorrelation functions of antisymmetric operators, compared to their symmetric counterparts [Eq.~\eqref{eq:trs_2pcf_onekindop} with numerical evidences in Fig.~\ref{fig:2pcf_numerics_antisym} from two quantum circuit models and a Hamiltonian model]. 
    
    \item A third derivation of the Ising scaling behaviour of the spectral form factor in {\gqmbcs} with TRS from summation of autocorrelation functions [Eq.~\eqref{eq:ising_scaling_2paf}].
    %
    %
    \item Exact results for the fluctuations of two-point autocorrelation functions in {\gqmbcs} with TRS [Eq.~\eqref{eq:2paf_fluc_stat_mech}], which exhibit exponential growth  in the operator support size governed by an emergent three-state Potts model [Eq.~\eqref{eq:2pcf_fluc_1d2} and numerical evidences in Figs.~\ref{fig:2pcf_numerics_fluc_in_op_size} and \ref{fig:2pcf_numerics_fluc_growth_rate} from two quantum circuit models and a Hamiltonian model].
    \item Discussions on how the out-of-time-ordered correlator between two operators depends on their spatial separation in the presence of TRS [Fig.~\ref{fig:otoc_diag}].
\end{itemize}
We emphasise that the Ising scaling behaviour of SFF in {\gqmbcs} with global TRS is derived via three complementary approaches [Fig.~\ref{fig:three_routes} and Table~\ref{Tab:sff_3approaches}]:
\begin{enumerate}
    \item Evaluation of SFF as double Feynman path integrals by accounting for the many-body interactions between time-parallel and time-reversed pairings of Feynman paths under the large-$q$ approximation [Eq.~\eqref{eq:ising_scaling}].
    \item Evaluation of the SFF via space-time duality using a parity symmetric Ginibre ensemble without relying on the large-$q$ approximation [Eq.~\eqref{eq:ginibre_scaling_form}]. 
    \item Evaluation of SFF as a sum of correlation functions under the large-$q$ approximation, which treat time-reversed and time-parallel pairings of Feynman paths asymmetrically [Eq.~\eqref{eq:ising_scaling_2paf}].  
\end{enumerate}
Together, these approaches give the equation
\be\label{eq:3approaches}
\kappa^{\text{SFF}}_{\text{g-TRS}}(x) = \kappa^{\text{Gin}}_{\text{g-TRS}}(x)
=
\kappa^{\text{2PAF}}_{\text{g-TRS}}(x) \, ,
\ee
where $\kappa$ are the SFF evaluated at the Thouless scaling limit computed via the three different routes. The effective Ising degrees of freedom in the three approaches are summarised in Table~\ref{Tab:sff_3approaches} and Fig.~\ref{fig:three_routes}.

\section{Models}\label{sec:models}
To model generic quantum many-body dynamics, we consider random quantum many-body circuits $U(t)$ acting on the Hilbert space of many qudits with local Hilbert space dimension $q$.  We consider three conditions on TRS and time translational symmetry:
\begin{enumerate}
    \item \textbf{Local TRS}: A local interaction described by some unitary operator $u$  is \textit{TRS} if $u$ satisfies
    \be\label{eq:ltrs_cond}
\mathcal{T} u \mathcal{T}^{-1} = u^\dagger \;.
\ee
We say that a quantum circuit $U(t)$ is \textit{local TRS} if $U(t)$ is composed of $u$ satisfying \eqref{eq:ltrs_cond}.
\item \textbf{Global TRS}: The time evolution operator $U(t)$ composed of local interactions $u$-s is called \textit{global TRS} if $U(t)$ satisfies
\be\label{eq:gtrs_cond}
\mathcal{T} U(t) \mathcal{T}^{-1} = U^\dagger(t) \;.
\ee
\item \textbf{Time translational symmetry (Floquet)}: $U(t)$  is a \textit{discrete time translational symmetric} or \textit{Floquet} operator if 
\be\label{eq:floq_cond}
U(t) = U^t \;.
\ee
\end{enumerate}
Rather intuitively, the condition Eq.~\eqref{eq:ltrs_cond} (and Eq.~\eqref{eq:gtrs_cond}) is imposing that time evolution followed by the associated time-reversed evolution is equal to identity, i.e. $\mathcal{T} u \mathcal{T}^{-1}u = \mathbb{1}$. The main advantage of modelling {\gqmbcs} with unitary time evolution operator, as opposed to, say, time-independent Hamiltonians, is its simplicity: these models can be constructed so that all symmetries and conserved quantities, including energy, are removed. These symmetries can be systematically added back to the systems, e.g.~\cite{khemani2018, vonKeyserlingk2017a, friedman2019, moudgalya2020spectral, chan2021trans}. Furthermore, such time evolution operators often have homogeneous density of states, which leads to particularly clean behaviour in observables, e.g. absence of the dip in SFF at early times, eliminating the need for unfolding, unlike the case of time-independent Hamiltonians. Lastly, unitary quantum circuits are naturally realized in digital gate-based quantum simulators and processors~\cite{fauseweh2024quantum}, see for example Google's experiment~\cite{arute2019quantum}.

In the following, we will define two random quantum circuits models, namely the Random Phase Model and the Haar-Random Model. 
For both models in later sections, we will provide analytical calculations in the limit of large local Hilbert space dimension, and numerical simulations at finite local Hilbert space dimension, so that the generality of the analytical results can be verified.
We will also define two models in random matrix theory. The first (second) one is based on the circular unitary (Ginibre) ensemble, acts (non-)unitarily in the time (space) direction, and captures the late-time (large-system-size) universal behaviour of {\gqmbcs}. 
For brevity, we will define the quantum many-body circuits using tensor networks diagrammatically in the main text, and provide their algebraic definitions in the Appendix~\ref{app:models}.

\subsection{Random quantum circuits}
\subsubsection{Random phase models (RPM)}
\begin{figure*}[ht]
    \centering
    \includegraphics[width=0.95\textwidth]{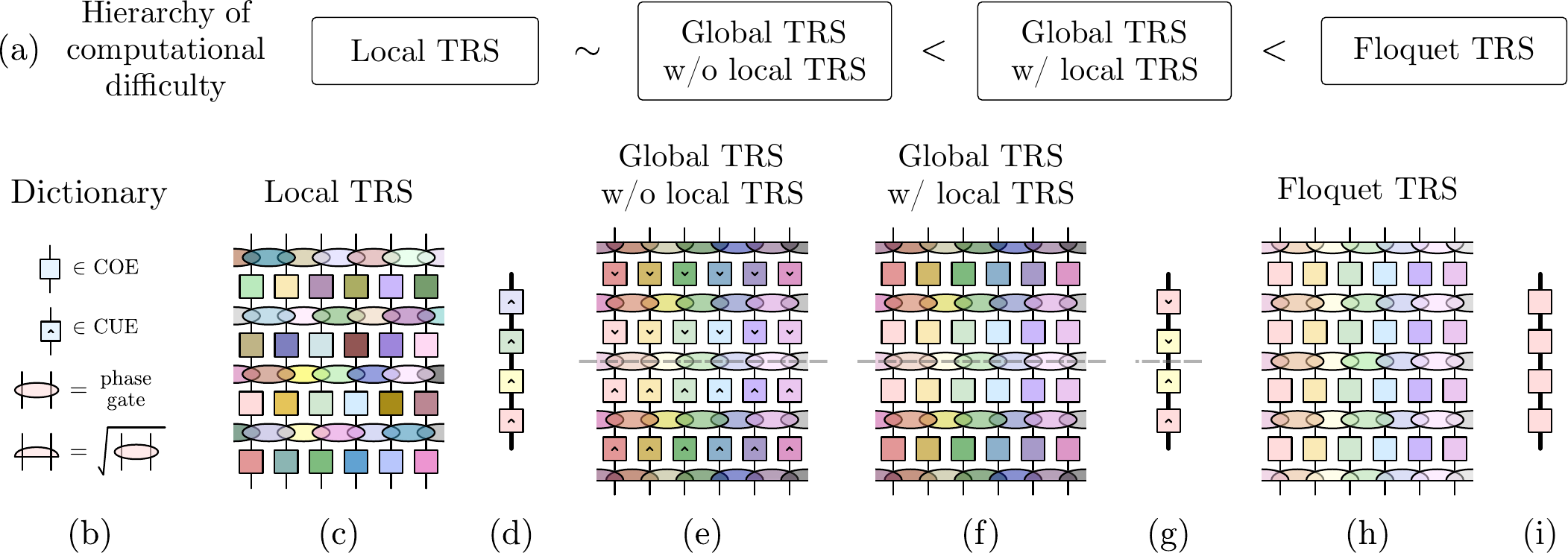}
    \caption{\textbf{TRS random phase models (RPM).} (a) Hierarchy of computational difficulty of TRS random quantum circuits. (b) Diagrammatical representation of two-site and one-site unitary gates in the RPM. 
    RPM with (c) \textbf{local TRS}, (e) \textbf{global TRS w/o local TRS}, (f) \textbf{global TRS w/ local TRS}, and (h) \textbf{Floquet TRS}. Gates illustrated in the same colour are being identified. The above models can be approximated in sufficiently late time by the coarse-grained models of random matrix theory (RMT)  in (d), (g), and (i) respectively. 
    Notice that even though 2-site unitary gates may be locally TRS, bi-layer of such gates may not have TRS, e.g. (c) (f).
For models with global TRS, the gates reflected across the time reversal axes (dashed lines) are identified. 
The gates at the top and bottom layers of (e), (f) and (h) are the half gates (see main text). Algebraic definitions of models are given in Appendix~\ref{app:models}.}
    \label{fig:model}
\end{figure*}

The \textit{Random Phase Models} (RPM) are many-body random quantum circuits  that act on the Hilbert space $\mathbb{C}^{q^L}$ of $L$ qudits. While many results derived for RPM are valid for general geometries in arbitrary dimensions, we will define RPM in one dimension for simplicity. The RPM is composed of 1-site gates $u_{\mathrm{COE}}$ drawn from the COE and/or CUE, and 2-site diagonal  gates $[u_{\mathrm{phase}}]_{a_j a_{j+1}}^{ a'_j a'_j+1}=\delta_{a_j, a'_j} \delta_{a_{j+1}, a'_{j+1}} \exp[i \varphi^{(j)}_{a_j,a_{j+1}} (t)]$ with  $a_j= 1,2\ldots, q$ that couples neighbouring sites with random phases. Each coefficient $\varphi_{a_j,a_{j+1}}^{(j)}(t)$ is an independent Gaussian random real variable with mean zero and variance $\epsilon$, which controls the coupling strength between neighboring spins. 
RPM is the first local and spatially-extended generic quantum many-body model that allows the analytical derivation of a universal deviation from RMT for time earlier than the Thouless time~\cite{chan2018spectral}.
By constraining the local gates in a quantum circuit to be site- or time-independent (or both), the RPM gives access to translational invariant in time and in space without TRS, as explored in \cite{chan2018spectral, mace2019quantumcircuitcriticality, chan2020lyap, Chan_2022, Shivam_2023, huang2023outoftimeorder, yoshimura2023operator, Yoshimura2025}.

To incorporate TRS, using the tensor networks illustrated in Fig.~\ref{fig:model}, we define the \textit{local TRS RPM} as the quantum circuit in Fig.~\ref{fig:model} (c) where all gates are independently drawn, such that condition 1 is satisfied, but not conditions 2 and 3. The \textit{global TRS RPM without local TRS}  is defined as the quantum circuit in Fig.~\ref{fig:model} (e), where gates reflected across the (horizontal) time reversal axis [dashed line in Fig.~\ref{fig:model} (e)] are identified, such that condition 2 is satisfied, but not condition 1 and 3. Similarly, the \textit{global TRS RPM with local TRS} is defined as the quantum circuit in Fig.~\ref{fig:model} (f), satisfying condition 1 and 2, but not 3. We will refer to these two models with global TRS collectively as the \textit{global TRS RPM}. Lastly, the \textit{Floquet TRS RPM} is defined as  the quantum circuit in Fig.~\ref{fig:model} (h), such that all three conditions are satisfied. Note that since the 2-site gates are diagonal, they trivially satisfy local TRS condition \eqref{eq:ltrs_cond}.
These models are defined algebraically in Appendix \ref{app:models}.
RPM without symmetries are chaotic only for $q \geq 3$ -- the same statement is found to be valid in TRS RPM.   Throughout this paper, we present exact analytics of RPM at large $q$ and general $\epsilon$, and numerics of RPM at $q=3$ and $\epsilon = 2$.

In order to satisfy the global TRS condition \eqref{eq:gtrs_cond} in quantum circuit geometries, we introduce the \textit{half gates} at the first and last time steps in Fig.~\ref{fig:model} (e,f,h) and also later in Fig.~\ref{fig:model_hrm_maintext} (b,c), illustrated by gates roughly half the size of gates in the bulks of the quantum circuits. 
For a unitary $u$ with TRS satisfying Eq.~\eqref{eq:ltrs_cond}, there always exists a $\mathcal{T}$-invariant basis such that $u$ is symmetric. Therefore, $u$ always has a spectral decomposition $u = \mathcal{S} D \mathcal{S}^{T}$ where $T$ denotes transposition, $D$ is the diagonal matrix of eigenvalues, and $\mathcal{S}^{T} \mathcal{S} = \mathbb{1}$. In turn,  we can decompose $u = v v^\mathrm{T} $ with $v = \mathcal{S} \sqrt{D} $ being the half gate.
Intuitively, one can think of the geometry of the circuits with global TRS [e.g. Fig.~\ref{fig:model} (e,f,h)] as the conventional geometry [e.g. Fig.~\ref{fig:model} (c)], but cyclically shifted half a time step upwards or downwards. These half gates can be thought of as the space-time-dual analogue of the gates that straddle between the left-most and right-most sites across the periodic boundary of a one-dimensional many-body circuits. In other words, a gate that straddle across boundaries in space would correspond to a half gate if the space and time axes are swapped. 
To define TRS RPM, we need the half gate associated to a 2-site random phase gate, which is obtained by simply taking the matrix square root of the diagonal random phase gate.  
To define TRS RPM in odd time (see Appendix \ref{app:models}) and TRS Haar-random model (see below), we also need the half gate associated to $u_{\mathrm{COE}}$ drawn from the COE. $u_{\mathrm{COE}}$ has a natural decomposition $u_{\mathrm{COE}} = u_{\mathrm{CUE}} u_{\mathrm{CUE}}^{T}$, and we take $u_{\mathrm{CUE}}$ as the half gate of $u_{\mathrm{COE}} $.

Heuristically, the computational difficulty of these models can be assessed by the moments of local gates $u_{\mathrm{CUE}}$ and  $u_{\mathrm{CUE}}^\dagger$ required to compute the moments of the entire circuit  $U(t,L)$ and $U^\dagger(t,L)$. By this criterion, the hierarchy of computational difficulty of random quantum circuits increases in the following order: local TRS {\gqmbcs} and global TRS w/o local TRS {\gqmbcs}, global TRS w/ local TRS {\gqmbcs}, and Floquet TRS {\gqmbcs} as in Fig.~\ref{fig:model} (a). For example, the computation of the 1st moment of $U(t,L)$ and $U^\dagger(t,L)$ requires respectively the 2nd, 2nd, 4th and $2t$-th moments of $u_{\mathrm{CUE}}$ and $u_{\mathrm{CUE}}^\dagger$. Note crucially that this implies that the SFF of local TRS {\gqmbcs}, global TRS w/o local TRS {\gqmbcs} only requires the 2nd moment of the Haar-random gates, and therefore can be efficiently simulated using the corresponding Clifford circuits by the Gottesman-Knill theorem~\cite{gottesman1998heisenberg}, which will be discussed in the upcoming work~\cite{upcoming}. 

\subsubsection{Haar-random models (HRM)}

\begin{figure}[ht]
    \centering
    \includegraphics[width=0.5 \textwidth]{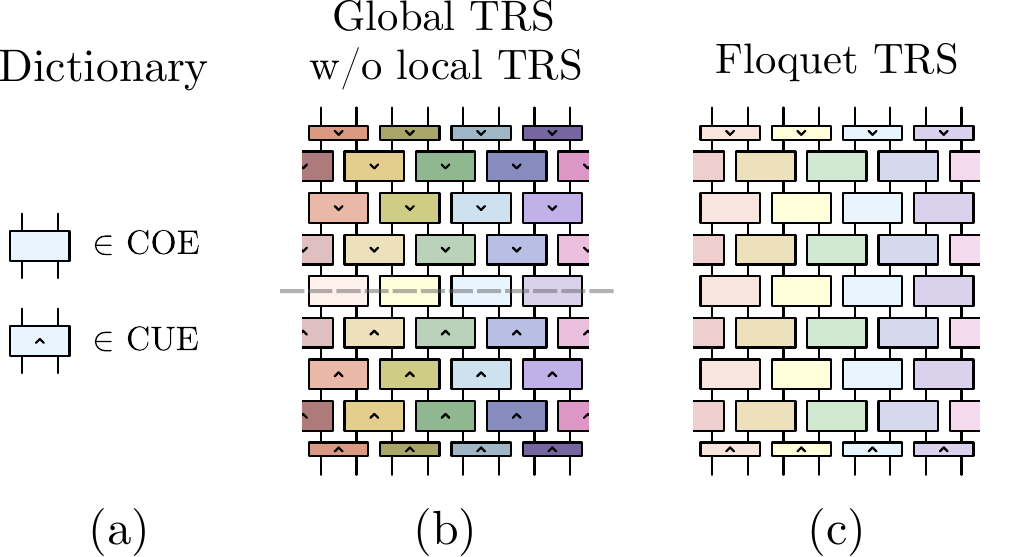}
    \caption{\textbf{TRS Haar-random models (HRM).} (a) Dictionary of the diagrammatical representation of COE and CUE unitary gates. HRM as models of {\gqmbcs} with (b) global TRS w/o local TRS and (c) Floquet TRS at time $t=4$. Gates illustrated in the same colour are being identified. The dashed line denotes TRS inversion axis. 
    Local TRS HRM (not illustrated) is similar to (c) except each unitary gate is independently drawn, and that the circuit geometry is shifted in time by half a time step. 
     Global TRS HRM w/ local TRS is defined as in (b), except that the CUE gates are replaced by COE gates, aside from the topmost and bottommost layers of CUE gates (the half gates). See the algebraic definitions in Appendix~\ref{app:models}. 
     }
     \label{fig:model_hrm_maintext}
\end{figure}

The \textit{Haar-random models} (HRM) are many-body random quantum circuits that act on the Hilbert space $\mathbb{C}^{q^L}$ of $L$ qudits.  HRM is composed of 2-site gates $u_{\mathrm{COE}}$ drawn from either the circular orthogonal ensemble (COE), or the circular unitary ensemble (CUE) [Fig.~\ref{fig:model_hrm_maintext} (a)].
Without TRS, the temporal- and spatial-random version of HRM has been studied extensively in the literature,  e.g. \cite{nahum2017, nahum2018, vonKeyserlingk2017a, keyserlingk2018, Skinner_2019,Li_2018}, and imposing time translational symmetry in HRM has allowed access to spectral statistics~\cite{chan2018solution, chan2020lyap, Chan_2022, huang2023outoftimeorder} and eigenstate correlations~\cite{cdc3}.
Closely analogous to the definitions of TRS RPM, we define the \textit{local TRS HRM} as in Fig.~\ref{fig:model_hrm_maintext} (c), except that all gates are drawn independently at each space time coordinate, and the geometry is shifted upwards by half a layer of gates. The \textit{global TRS HRM without local TRS} is defined as in Fig.~\ref{fig:model_hrm_maintext} (b). The \textit{global TRS HRM with local TRS} is defined similar to Fig.~\ref{fig:model_hrm_maintext} (b), except the CUE gates are replaced with COE gates aside from the topmost and bottommost layers. Lastly, the \textit{Floquet TRS HRM} is defined using Fig.~\ref{fig:model_hrm_maintext} (c). 
These models are defined algebraically in Appendix \ref{app:models}.
We will provide exact analytics of HRM at large $q$, and numerical simulations at finite $q=2$.

\subsubsection{3-parameter models (3PM)}
The \textit{3-parameter models} (3PM)~\cite{znidaric2022, huang2023outoftimeorder}  are many-body random quantum circuits  that act on the Hilbert space $\mathbb{C}^{q^L}$ of $L$ qudits. 
The  3PM shares the same brick-wall geometry as the HRM, but with two-site gates acting on the $j$-th and $(j+1)$-th qudits given by
  \begin{align}
  \label{eq:3pm_fullgate}
 	& u^{(j,j+1)} = \left[u^{(j)}_1 \otimes u^{(j)}_2\right]
   u^{(j)}_{\mathrm{XYZ}}
    \left[u^{(j)}_1 \otimes u^{(j)}_2\right]\, ,
    \\ \nonumber
    &u^{(j)}_{\mathrm{XYZ}}=
     \exp\left(i \sum_{\mu = x, y, z} a_\mu \sigma^\mu_{j}  \sigma^\mu_{j+1} \right) \,,
\end{align}
where for each $j$ and $i$, the unitary matrix $u^{(j)}_i$ is independently drawn from the $\text{COE}$ of $2$-by-$2$ unitary matrices in the presence of TRS.
$\sigma_j^{\mu}$ is the Pauli-$\mu$ matrix acting on site $j$, and we take $(a_x, a_y, a_z) = (0.3, 0.4, 0.5)$. In the presence of TRS, half gates for the 3PM are defined as 
\be
 \label{eq:3pm_halfgate}
 u_{\half}^{(j,j+1)} = 
    \exp\left( \frac{i}{2} \sum_{\mu = x, y, z} a_\mu \sigma^\mu_j \sigma^\mu_{j+1} \right)
    \left[u^{(j)}_1 \otimes u^{(j)}_2\right]\, ,
\ee
with $u^{(j)}_i \in$ COE, such that $u^{(j,j+1)} = \left[ u_{\half}^{(j,j+1)}\right]^T u_{\half}^{(j,j+1)}$. The Floquet TRS 3PM is defined as in Fig.~\ref{fig:model_hrm_maintext} (b) with full and half gates given by Eqs.~\eqref{eq:3pm_fullgate} and \eqref{eq:3pm_halfgate}, respectively.
Floquet 3PM (without TRS) is defined similarly to Floquet TRS 3PM, except that $u_i^{(j)}$ is drawn from the CUE instead of the COE, and that the brick-wall geometry does not require a shift. See the appendix for the algebraic definitions.
In this manuscript, we focus on Floquet TRS 3PM and Floquet 3PM (no time-reversal symmetry), although one may likewise introduce 3PM with local TRS or global TRS in direct analogy with the HRM.

\subsection{Hamiltonian models}
We use a time-independent Hamiltonian model to test the applicability of the results derived from random quantum circuits. 
Specifically, we define 
\be
\ba
H= & \, H_{\text{XYZ}} + H_{\text{NNN}}
\\
H_{\text{XYZ}}=& \, \sum_{\langle i,j \rangle} \left( 
a_x \sigma^x_i \sigma^x_j + a_y \sigma^y_i \sigma^y_j + a_z \sigma^z_i \sigma^z_j   
\right)
\\
H_{\text{NNN}} =& \, 
\begin{cases}
\sum_{\langle\!\langle i,j\rangle\!\rangle} h^{\mathrm{GOE}}_{i,j} \qquad \qquad & \text{TRS} 
\\
\sum_{\langle\!\langle i,j\rangle\!\rangle} h^{\mathrm{GUE}}_{i,j}
\qquad \qquad & \text{No symmetries} 
\end{cases}
\ea
\ee
where we take $(a_x, a_y, a_z) = \frac{1}{4}(1.2,1,0.8)$ at which the model displays chaotic behaviour. $h^{\mathrm{GOE/GUE}}_{i,j}$ are random matrices drawn from the Gaussian orthogonal ensemble (GOE) and Gaussian  unitary ensemble (GUE)  respectively. $\langle i,j \rangle$ and  ${\langle\!\langle i,j\rangle\!\rangle}$ denote all pairs of nearest-neighbour sites and next-to-nearest-neighbour (NNN) sites respectively. 
By adding the NNN terms drawn from the GOE, we break all symmetries of the XXZ model except the continuous time translational symmetry and TRS. When the NNN terms are drawn from the GUE instead, we additionally break the TRS.

\subsection{Random matrix theory (RMT)}\label{sec:mm_models}
Random matrix theory (RMT) is the studies of the statistical properties of an ensemble of random matrices, and can be used to model chaotic dynamics of many-body quantum chaotic dynamics in sufficiently late time scales or sufficiently small energy scales~\cite{bohigas1984characterization}. 
In this paper, we study two classes of models in RMT. The first class is based on the circular ensembles, acts unitarily in the time  direction, and captures the late-time universal behaviuor of {\gqmbcs}. 
More specifically, the first class of models in RMT are defined via time evolution operators with unitary $u$ drawn from the circular unitary ensemble (CUE) or the circular orthogonal ensemble (COE) of $N$-by-$N$ unitaries~\cite{Haake}. 
Using diagrammatical tensor network, we define the first class of RMT that models {\gqmbcs} with local TRS in Fig.~\ref{fig:model} (d); RMT that models {\gqmbcs} with global TRS in Fig.~\ref{fig:model} (g); and RMT that models Floquet {\gqmbcs} with TRS in Fig.~\ref{fig:model} (i). The models are algebraically defined in Appendix~\ref{app:models}. 
We will refer to the RMT models of circular ensembles simply as the ``RMT''. 

The second class of RMT models is based on the Ginibre ensemble, which is an ensemble of random matrices with independent complex Gaussian matrix elements~\cite{Mehta}.
It has recently been shown that {\gqmbcs} displays universal signatures of the  Ginibre ensemble~\cite{Shivam_2023}. In turn, the Ginibre ensemble can be used to model the coarse-grained \textit{non-unitary} evolution of the {\gqmbcs} in the \textit{spatial} direction. Since this connection between Ginibre ensemble and {\gqmbcs} is relatively recent, we will devote a separate section to define and analyse them below. We will refer to the second class of RMT models explicitly as the Ginibre ensembles.

\subsection{TRS-breaking mechanisms}
Symmetry-breaking mechanisms for the global TRS are important because in the emergent classical statistical mechanical problem dual to the SFF of random quantum circuits (see below), TRS-breaking mechanisms can be used to generate external magnetic fields, by favouring local time-parallel pairings of Feynman paths over the local time-reversed pairings of Feynman paths. Here we will describe two TRS-breaking mechanisms.

The first TRS-breaking mechanism can be used in {\gqmbcs} defined by time-independent Hamiltonians or Floquet operators.  
The computation of SFF (see below) amounts to analyse a sum of many-body diagrams which can be organised by the order in $q^{-1}$, according to Lemma \ref{lemma:leading_sff_diag}, the leading local SFF diagrams are of two types: The ladder diagrams [e.g. Fig.~\ref{fig:sff_diag} (e)], corresponding to the possible time-parallel pairings of Feynman paths, i.e. between a Feynman path and its complex conjugate counterpart, both traversing in the same direction in time; and the twisted diagrams [e.g. Fig.~\ref{fig:sff_diag} (b)], corresponding to the possible time-reversed pairings of Feynman paths, i.e. between a Feynman path and its time-reversed complex conjugate counterpart. 
Suppose that there are non-TRS unitary operators appearing in the time evolution of a quantum many-body system. We observe that ladder contractions  and twisted contractions [Fig.~\ref{fig:contraction}] in SFF diagrams (and in other observables) can form a loop around the product of the non-TRS (non-symmetric) unitary operator $u_\mathrm{b}$ and its complex conjugate $u_\mathrm{b}^*$ [Fig.~\ref{fig:trs_breaking}]. For the ladder case, this leads to $\Tr[u_\mathrm{b}^T u^*_\mathrm{b} ]= \Tr[\mathbb{1}] = q$. For the twisted case, this leads to $\Tr[u_\mathrm{b} u^*_\mathrm{b}]$. Consequently, the insertion of non-TRS unitary operators generally biase one type of contractions over the other. 
This non-TRS unitary operator $u_\mathrm{b}= \exp(ih_\mathrm{b})$ can be generated from a Hermitian matrices $h_\mathrm{b}$ such that $h_\mathrm{b}^T \neq h_\mathrm{b}$ is not symmetric, and can generally be introduced in a time-independent TRS quantum many-body Hamiltonians.

\begin{figure}[ht]
    \centering
    \includegraphics[width=0.47 \textwidth]{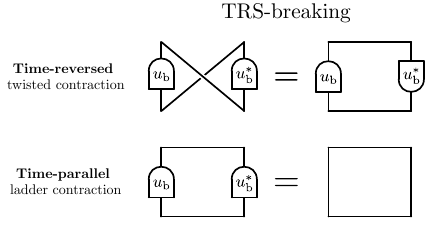}
     \caption{\textbf{TRS-breaking.} A time-independent global TRS-breaking mechanism can be constructed in quantum dynamics by noting that the ladder (time-parallel) and twisted (time-reversed) contractions [Fig.~\ref{fig:contraction}] of non-TRS (non-symmetric) unitary operator $u_\mathrm{b}$ and its complex conjugate $u_\mathrm{b}^*$  leads to $\Tr[u_\mathrm{b}^T u^*_\mathrm{b} ]= \Tr[\mathbb{1}] = q$ in the ladder or time-parallel case, and $\Tr[u_\mathrm{b} u^*_\mathrm{b} ]$ in the twisted or time-reversed case. $u_\mathrm{b}$ can be generated from a Hermitian matrices, and introduced in a time-independent TRS quantum many-body Hamiltonian. An alternative global TRS-breaking mechanism is to introduce explicit non-TRS time-dependent unitary evolution. See TRS-broken RPM model in Fig.~\ref{fig:trs_breaking_RPM}.
    }
    \label{fig:trs_breaking}
\end{figure}

The second TRS-breaking mechanism can be used in {\gqmbcs} described by time-dependent Hamiltonian or unitary time evolution operator. To this end, we introduce time-dependent gates to explicitly break the global TRS, so that the second half of the quantum circuit above time inversion axis (labelled by the dashed grey line in Fig.~\ref{fig:model} (e)) is no longer the transpose of the first half of the quantum circuit below the time inversion axis. Specifically, we define the \textit{global TRS RPM with TRS-breaking} via the tensor network illustrated in Fig.~\ref{fig:trs_breaking_RPM}, where each 1-site CUE or COE gate in global TRS RPM is now followed by a 1-site diagonal random phase gate $[u_{\mathrm{phase}}]_{a_j }^{ a'_j }=\delta_{a_j, a'_j}  \exp[i \phi^{(j)}_{a_j}(t)]$ with  $a_j= 1,2\ldots, q$. Each coefficient $\phi_{a_j}^{(j)}(t)$ is an independent Gaussian random real variable with mean zero and variance $b$, controlling the strength of the TRS-breaking mechanism.  These additional 1-site phase gates are independently drawn at each space-time coordinate, which breaks the global TRS. This is  analogous to the 2-site diagonal random phase gates in RPM. Much like the RPM, the simplicity of the model lies in the fact that these phase gates are diagonal, which implies that they are local TRS. 

Let's briefly compare the two TRS-breaking mechanisms. The first TRS-breaking mechanism utilizes unitary gates that are not local TRS, while the second mechanism relies on the insertion of time-dependent unitary gates to break TRS, and in fact these latter gates are chosen to be local TRS in Fig.~\ref{fig:trs_breaking_RPM} for simplicity. The first mechanism is suitable for a wide-range of {\gqmbcs} including time-independent Hamiltonian or Floquet many-body systems, while the second mechanism, in exchange for simplicity, can be applied only to time-dependent many-body systems.

\begin{figure}[ht]
    \centering
    \includegraphics[width=0.4 \textwidth]{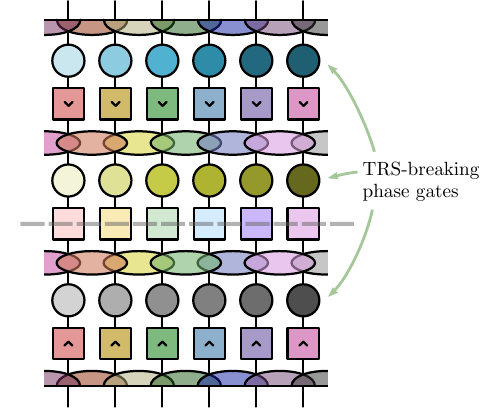}
     \caption{\textbf{Global TRS RPM with TRS-breaking mechanism} is defined identically as the global TRS RPM in Fig.~\ref{fig:model} (e), except that each 1-site CUE or COE gate in global TRS RPM is now followed by a 1-site diagonal random phase gate (analogous to the 2-site diagonal random phase gates). These additional 1-site phase gates are independently drawn at each space-time coordinate, which breaks the global TRS. The illustration is for time $t=3$. Gates with the same colour are identified.
    }
    \label{fig:trs_breaking_RPM}
\end{figure}

\section{Diagrammatical approach}\label{sec:diag}
\begin{figure}[ht]
    \centering
    \includegraphics[width=0.35 \textwidth]{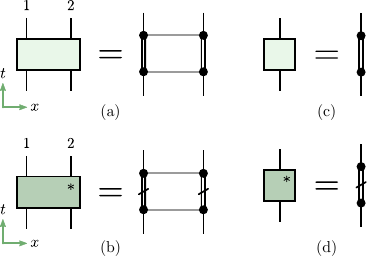}
     \caption{\textbf{Diagrammatical representation.} The diagrammatical representation of unitary matrix drawn from the COE for (a) a 2-site gate, (b) its conjugation, (c) a 1-site unitary gate, and (d) its conjugation. 
    }
    \label{fig:dictionary}
\end{figure}

\begin{figure}[ht]
    \centering
    \includegraphics[width=0.46 \textwidth]{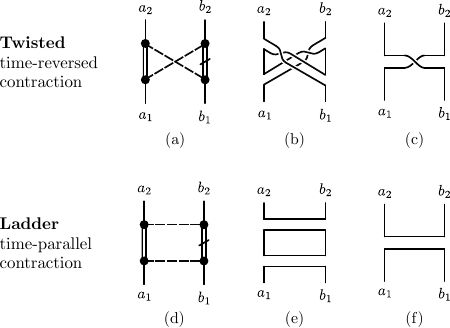}
     \caption{\textbf{Time-reversed and time-parallel contractions.} For Eq.~\eqref{eq:coe_formula} with $n=m=1$, we have $\overline{u_{a_1 a_2}u^*_{b_1 b_2}} = \wg[\mathbb{1}, N] \left(  \delta_{a_1 b_2} \delta_{a_2 b_1} +  \delta_{a_1 b_1} \delta_{a_2 b_2}  \right)$ where $u$ is drawn from the COE. (a) Diagrammatical representation of the first term, the \textit{twisted} or \textit{time-reversed contractions} with $\sigma=\text{SWAP} \in S_2$, which can be represented using 't Hooft's double line notation in (b)~\cite{tHooft1973alw}. To put it in its simplest form, the twisted term contains a twist, namely $\delta_{a_1 b_2} \delta_{a_2 b_1}$, as illustrated in (c). In contrast, the second term contains \textit{ladder} or \textit{time-parallel contractions} $\sigma=\iden_2 \in S_2$ as illustrated in (d), (e), and  (f). Note that ladder contractions appear in ensemble averages of both {\gqmbcs} with or without TRS. Note that we use the term ``twisted diagrams'' instead of ``maximally crossed diagrams'' in the literature of disordered systems. This choice emphasises  the non-planar nature of the contractions~\cite{upcoming}, and in the context of random quantum circuits,  ``crossed diagrams'' composed of ladder contractions have been used ~\cite{chan2021trans}. } 
    \label{fig:contraction}
\end{figure}

\begin{figure}[ht]
    \centering
    \includegraphics[width=0.48 \textwidth]{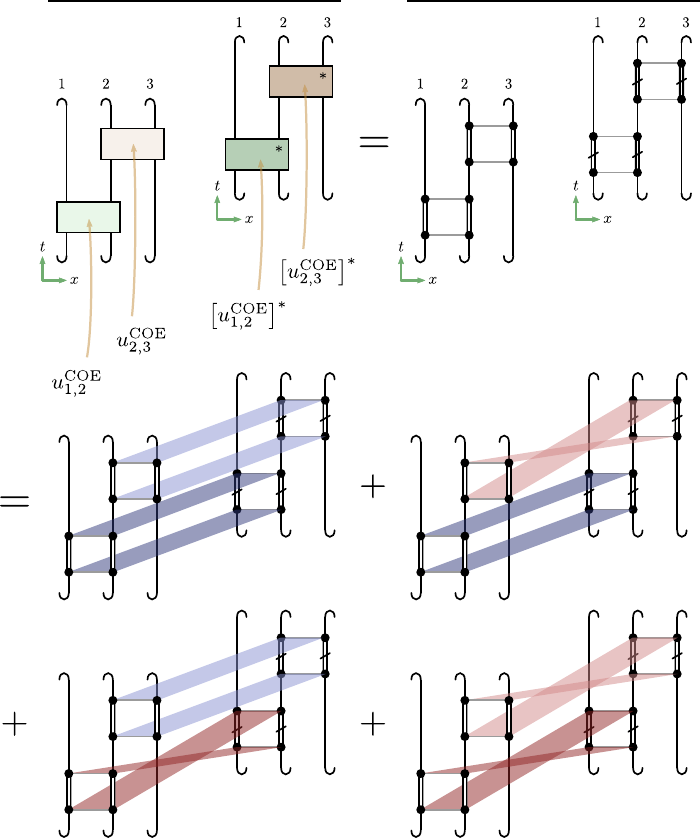}
     \caption{\textbf{Diagrammatical approach example} of $\overline{\left|\Tr \left[u^{\COE}_{2,3} u_{1,2}^{\COE}\right] \right|^2}$, where $u^{\COE}_{2,3}$ and $u^{\COE}_{1,2}$ are drawn from the COE. Here we use (blue and red) strips to denote multi-site contractions (dashed lines in Fig.~\ref{fig:contraction} for the single-site case). The contractions $(\sigma_{1,2}, \sigma_{2,3}) \in S_2 \times S_2$ for the four terms are $(\iden_2, \iden_2)$, $(\iden_2, \text{SWAP})$, $(\text{SWAP}, \iden_2)$, and $(\text{SWAP}, \text{SWAP})$ respectively. }
    \label{fig:diag_eg}
\end{figure}

Here we review the integration over $N$-by-$N$ unitary matrix $u$ and its complex conjugate $u^*$ drawn from the circular orthogonal ensemble (COE)~\cite{samuel1980wg1, creutz1978wg2, Brouwer1996, MATSUMOTO_2012}, which is given by
\be \label{eq:coe_formula}
\begin{aligned}
\overline{
[u]_{a_1 a_2} \dots 
[u]_{a_{2n-1} a_{2n}} 
[u^*]_{b_1 b_2} \dots 
[u^*]_{b_{2m-1} b_{2m}}}
\\
= \delta_{n,m} \sum_{\sigma \in S_{2n}} \wg_{\mathrm{COE}}[\sigma; N] \prod_{i=1}^{2n} \delta_{a_i, b_{\sigma(i)}} \;,
\end{aligned}
\ee
where $\wg_{\mathrm{COE}}[\sigma; N]$ is the Weingarten function of the COE, taking arguments from $\sigma \in S_{2n}$ in the symmetry group of $2n$ objects. Define $\sigma_\mathrm{o}\in S_n$ by $\sigma_\mathrm{o}(n) = \left\lceil \sigma (2n-1)/2 \right\rceil$, and similarly $\sigma_\mathrm{e}\in S_n$ by $\sigma_\mathrm{e}(n) = \left\lceil \sigma (2n)/2 \right\rceil$, where $\left\lceil  \cdot \right\rceil$ is the ceiling function. The Weingarten function is found to depend only on the cycle structure, $\{c_1, \dots, c_k \}$,  with $c_i$ the length of the $i$-th cycle of $\sigma_\mathrm{o}^{-1} \sigma_\mathrm{e}$. With an abuse of notation, we write  $\wg_{\mathrm{COE}}[\sigma; N]= \wg_{\mathrm{COE}}[\{c_1, \dots, c_k\}; N]$. The Weingarten function can be generated under a recursive relation~\cite{Brouwer1996, MATSUMOTO_2012}, which we provide in the App.~\ref{app:wg_coe}.   
As examples, for $n=1$, we have $\wg[\{ 1\}, N ]= 1/(N+1)$, and for $n=2$, we have 
$\wg[\{1, 1\}, N ]= (N+2)/N(N+1)(N+3)$, and
$\wg[\{2\}, N ]= -1/N(N+1)(N+3)$.
For our purpose, it is useful to note that $\wg[(c_1, c_2, \dots , c_k); N] = O(N^{k - 2\sum_{i=1}^k c_i})$, i.e. for fixed total loop length, $\sum_{i=1}^k c_i$, the Weingarten function is the largest in the order of $N$ when the number of cycles is maximized.   

We employ the diagrammatical approach for random matrix theory~\cite{Brouwer1996} and for quantum many-body systems~\cite{chan2018solution}, as follows:
\begin{itemize}
    \item[(i) ] Diagrammatical representation: Each random unitary matrix from the COE is represented with two dots, single lines and double lines as in Fig.~\ref{fig:dictionary}. In particular, the dots represent the indices of the unitary matrix in \eqref{eq:coe_formula}. Their complex conjugates are similarly represented except that we add a slash to distinguish the conjugation. These unitaries can act on multiple sites or qudits [Fig.~\ref{fig:dictionary} (a,b)], or a single site or qudit [Fig.~\ref{fig:dictionary} (c,d)]. Observables like the SFF are represented by connecting the single lines using the standard tensor network notation, see e.g. \cite{Or_s_2014}. 
    \item[(ii) ] Contractions: To average over $n$ pairs of unitaries and conjugates, one generates a sum of all possible contractions $\sigma \in S_{2n}$ between $2n$ indices or dots of the unitaries and $2n$ indices or dots of the conjugates. We represent the contractions with dashed lines (except for Fig.~\ref{fig:diag_eg}, where the multi-site contractions are represented using thick lines in red and blue for clarity). See Fig.~\ref{fig:contraction} for an example with $n=1$ and $\sigma \in S_2 = \{\iden_2, \text{SWAP}\}$ where SWAP is the swap permutation of two objects.
    \item[(iii) ] many-body interactions: In many-body systems, diagrammatical rules are used to account for the interactions between different sites or qudits depending on the model. For HRM with COE gates, we impose the \textit{bond constraint}, which simply refers to the fact that contractions of a given unitary matrix that act on multiple sites, have to be identical across all sites. For the RPM, a random phase gate must take on indices consistent with the diagrams of the two qudits it acts upon. 
    \item[(iv) ] Translation to algebraic terms: The diagrams are translated to algebraic terms by accounting the delta and Weingarten functions in \eqref{eq:coe_formula}. In HRM, $\wg[\sigma; q^2]$ is assigned at each set of identical random unitaries on each bond, while the delta functions typically giving rise to factors of $q^{\ell}$, with $\ell$ denoting the number of loops of single and dashed lines. For RPM, $\wg[\sigma; q]$ is assigned at each site, and upon ensemble averaging, the integral over random phases give rise to factors dependent on the variance $\epsilon$ of the random phases, see ~\cite{chan2018spectral} and examples below. 
    \item[(v)] Identifying dominant contributions: The ensemble-averaged observable is a sum of algebraic or diagrammatic contributions, organised in terms of the order in the local Hilbert space dimension $q$. The order of a diagram is the product of the orders of local diagrams associated to Haar-random unitaries, which in turn is obtained by evaluating and inspecting delta functions and Weingarten functions in Eq.~\eqref{eq:coe_formula}.  
\end{itemize}
In Fig.~\ref{fig:diag_eg}, using the diagrammatical approach, we provide a 3-site example of the evaluation of 
\be
\begin{aligned}
\overline{ \left|\Tr \left[u^{\COE}_{2,3} u_{1,2}^{\COE} \right] \right|^2}=& \; \wg_{\mathrm{COE}}[\{1\}; q^2]^2  2 q^3(q+1) 
\\
=& \; 2 q^3(q+1)/(q^2+1)^2 \,, 
\end{aligned}
\ee
where $u^{\COE}_{2,3} $ and $u^{\COE}_{1,2} $ are $q^2$-by-$q^2$ unitary matrices independently drawn from the COE, and act on pair of sites $(2,3)$ and $(1,2)$ respectively. 
The distinctiveness of the physics of TRS {\gqmbcs} lies in the existence of \textit{twisted} or \textit{time-reversed contractions}, as illustrated in Fig.~\ref{fig:contraction} (a-c), and the many-body interactions between twisted diagrams and ladder diagrams, leading to interesting features in SFF and correlation functions, as we will see below.

\section{Spectral form factor}\label{sec:sff}
The \textit{spectral form factor} (SFF) is the Fourier transform of the two-level correlation function, which describes the probability of finding two eigenenergies separated by a certain distance in the energy spectrum.
The SFF is arguably the simplest analytically tractable quantity to capture the universal spectral fluctuation of quantum systems up to an arbitrary energy scale, and consequently, SFF has been instrumental in multiple frontiers of physics, such as the semi-classical approach to quantum chaos~\cite{Hannay1984PeriodicOA, berry1985semiclassical, Sieber_2001, muller_2004, muller_2005}, the studies of black holes~\cite{kitaev2015simple, garciagarcia2016, Cotler_2017, Gharibyan_2018, Saad2019semiclassical, Altland_2021, Saad_2022kfe}, many-body quantum chaos~\cite{chan2018solution, kos_sff_prx_2018, chan2018spectral, bertini2018exact} and its transition to prethermal many-body localization (MBL)~\cite{basko2006metal, oganesyan2007localization, chaoschallengeMBL, prakash2021}, dynamics in open quantum systems~\cite{li2021spectral, fyodorov1997, dsff2022kulkarni,
Garc_a_Garc_a_2023dsff, li2024spectral}, and more \cite{Cotler_2017complexity, del_Campo_2017, ALTLAND201845, Liu_2018, liaogalitski2020, chan2020lyap, roy2020random, SUNTAJS2021168469, prakash2021, winer2022hydrodynamic, sfffilter2022delcampo, Winer_2022glass, santos2022, suntajs_2022, Dag2023, Fritzsch_2023rxt, yoshimura2023operator, vikram2024exact, das2024proposal, fritzsch2024eigenstatecorrelationsdualunitaryquantum}.  Recently, signatures of quantum chaos in SFF have been experimentally measured in quantum simulators for up to five qubits~\cite{Dong_2025}.
Formally, the SFF can be defined as 
\be\label{eq:sff_def}
K(t):=  \sum_{a,b} e^{i(E_a-E_b)t} = \,  
\left| \mathrm{Tr}\, U(t) 
\right|^2 \, ,
\ee
where $\{ E_{a} \}$ are the (quasi-)energies of the system, and $\overline{(\dots )}$ denotes the ensemble average over statistically-similar systems. It is useful to adopt a generalized formulation of SFF such that  the SFF
is defined directly through the time evolution operator $U$ of a system of interest. This formulation of SFF allows the study of temporal-random systems, which in turn allows us to isolate important physical aspects of global TRS, see below. 
 Note that we adopt the convention where $K(t)$ is normalized such that trivially $K(0)=\Nhil^2$ for all models we consider, with $\Nhil$ the dimension of the Hilbert space.

We will show that the SFF of a minimal model of  {\gqmbcs} with TRS can be mapped to the partition functions $Z_{\mathrm{classical}}$ of certain  classical spin chains described by the Hamiltonian $H_{\text{classical}}$ in the limit of large local Hilbert space dimension $q$, i.e.
\begin{equation}\label{eq:general_mapping}
\ba
 &   \lim_{q \to \infty}
    \overline{K_{\mathrm{quantum}}} 
    = \, \, Z_{\mathrm{classical}}  \,,
    \\
&Z_{\mathrm{classical}} =  \sum_{\{\sigma_i\}} e^{-\beta H_{\text{classical}}(\{ \sigma_i\})}  \,,
    \ea
\end{equation}
where $\{\sigma_i\}$ are the classical degrees of freedom. It is also often convenient to compute the partition function  $Z_{\mathrm{classical}} =\sum_{\{\sigma_i\}} \prod_{\langle i,j\rangle} \mathcal{B}^{\text{classical}}_{\sigma_i, \sigma_j}$ as a product over pairs of coupled sites $\langle i,j\rangle$ using the Boltzmann factor $\mathcal{B}^{\text{classical}}_{\sigma_i, \sigma_j}$.
This approach was first introduced in \cite{chan2018spectral} for a Floquet many-body quantum chaotic system without symmetries, which enabled the analytical derivation of the existence of an extensive ``bump'' region in generic quantum many-body chaos, and subsequently observed and investigated in \cite{friedman2019, Gharibyan_2018, moudgalya2020spectral, chan2020lyap, chaoschallengeMBL, chan2021trans, Garratt2021prx, yoshimura2023operator, Dag2023}.

\subsection{SFF diagrams}

\begin{figure*}[ht]
    \centering
    \includegraphics[width=1 \textwidth]{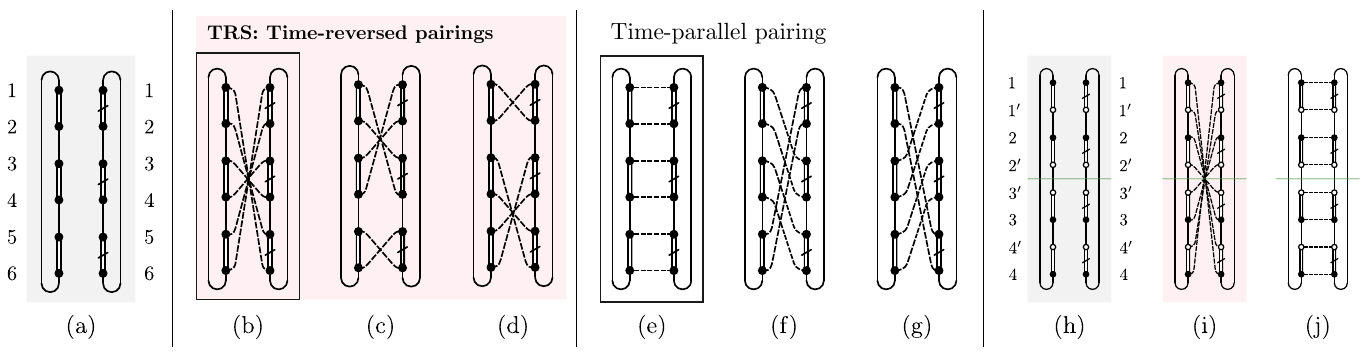}
    \caption{\textbf{SFF diagrams.} (a) The diagrammatical representation of a local diagram of SFF, $K(t=3)$, at $t=3$ for quantum many-body circuit  before ensemble averaging, for all cases of TRS dynamics except for global TRS without local TRS dynamics. 
    For \textbf{Floquet TRS} chaotic systems, the leading local diagrams of SFF after ensemble averaging are given by the twisted and ladder diagrams in (b-g), while for Floquet chaotic models without TRS, the leading diagrams are given by only the ladder diagrams in (e-g).
    For \textbf{global TRS w/ local TRS} chaotic systems, the leading local diagrams are (b) and (e). Note that for temporal-random systems without TRS, the leading local diagram is (e). 
    For \textbf{global TRS w/o local TRS} chaotic systems,  the diagrammatical representation of SFF for $K(t=4)$ and its leading diagrams are given in (h) and (i) to (j) respectively.
    For \textbf{local TRS} chaotic systems, the leading diagram is simply the ladder diagram (e). 
    The boxed diagrams, (b) and (e), form the emergent effective Ising degrees of freedom  for global TRS chaotic systems.}
    \label{fig:sff_diag}
\end{figure*}

In this subsection, we determine the leading and leading subleading local SFF diagrams in the order of local Hilbert space dimension, which will allow us to compute the leading contributions to SFF in the many-body setting later. 
Under the diagrammatical approach introduced above, the ensemble-averaged SFF is a sum of many-body SFF diagrams, whose order is the product of the orders of local SFF diagrams. 
Each local SFF diagram is represented as a permutation $\sigma$ in \eqref{eq:coe_formula}, representing the contractions between the dots of $\Tr[U(t)]$ and those of $\Tr[U^\dagger(t)]$ with a convention given in Fig.~\ref{fig:sff_diag} (a) for the Floquet TRS RPM.

\begin{lemmaalt} \label{lemma:leading_sff_diag}
\textbf{Leading local spectral form factor (SFF) diagrams for generic quantum many-body  chaotic systems with TRS.} Consider the SFF \eqref{eq:sff_def} for (i) TRS random phase model (RPM), (ii) TRS Haar-random model (HRM), and (iii) the random matrix model (RMT) for TRS {\gqmbcs}.
\begin{itemize}
    \item \textbf{Floquet TRS (TRS with discrete time translational symmetry):}
    In the order of the local Hilbert space dimension $q$ for (i) and (ii), and in matrix dimension $N$ for (iii), the $2t$ leading SFF diagrams of order $O(1)$ are given by $t$ twisted diagrams [Fig.~\ref{fig:sff_diag} (b-d)],
\be\label{eq:twisted_sff}
\sigma_{\mathrm{twisted}}^{(m)}(i) = (m - i +1) \Mod{2t}  \quad \in S_{2t}
 \,, 
\ee
for $m=2,4, \dots, 2t$, and $t$ ladder diagrams [Fig.~\ref{fig:sff_diag} (e-g)],
\be\label{eq:ladder_sff}
\sigma_{\mathrm{ladder}}^{(m)}(i) = (m + i -1) \Mod{2t} \quad \in S_{2t}
 \,, 
\ee
where $m=1,3,\dots, 2t-1$ with $i=1,2, \dots, 2t$ labelling the dots in Fig.~\ref{fig:sff_diag}(a).
    \item \textbf{Global TRS:} The two leading SFF diagrams of order $O(1)$ are given by $ \sigma_{\mathrm{twisted}}^{(m=2t)}$ and $\sigma_{\mathrm{ladder}}^{(m=1)}$.

    \item \textbf{Local TRS:} The single leading SFF diagram of order $O(1)$ is given by $\sigma_{\mathrm{ladder}}^{(m=1)}$.
\end{itemize}

\end{lemmaalt}

\noindent\textit{Proof.} The proof and technical details on the parametrisation of these permutations for different models are given in Appendix~\ref{app:sff_diag_proof}.

\vspace{0.25cm}

As examples, for the twisted diagrams in SFF for Floquet TRS RPM  on each site at $t=3$, see Fig.~\ref{fig:sff_diag} (d), (c), and (b)  for  $\sigma_{\mathrm{twisted}}^{(m)}$ with $m=2,4,$ and $6$ respectively. For the ladder SFF diagrams of the same models, Fig.~\ref{fig:sff_diag} (e), (f), and (g)  for  $\sigma_{\mathrm{ladder}}^{(m)}$ with $m=1,3,$ and $5$ respectively. 

Note that for SFF diagrams for global TRS RPM without local TRS, the parametrization of the permutations needs to be modified as in Fig.~\ref{fig:sff_diag} (h). For HRM, a local SFF diagram at site $j$ requires additional colouring of dots from unitary gates acting on the bonds $(j,j-1)$ and $(j,j+1)$. See Sec.~\ref{sec:diag} and Appendix~\ref{app:sff_diag_proof} for details.

The leading local SFF diagrams, Eq.~\eqref{eq:twisted_sff} and \eqref{eq:ladder_sff}, coincide with contributions identified through the diagonal approximation in the periodic orbit theory~\cite{Hannay1984PeriodicOA, berry1985semiclassical}.
It is important to emphasise that these diagrams appear for \textit{local} quantum degrees of freedom, i.e. a single site for RPM or a single bond for HRM. Different subregions of the system can take different leading SFF diagrams, and the many-body interactions among these diagrams give rise to signatures of quantum many-body chaos.
Indeed, we will see that the onset of RMT behaviour in SFF, parametrised by the many-body Thouless time, is located when  the SFF is dominantly contributed by many-body SFF diagrams where all local SFF diagrams are identical.
Before the Thouless time, different subregions of the system can take on different values of local leading diagrams, and effectively behave like patches of RMT~\cite{chan2018spectral}. These contributions give rise to the bump before the Thouless time on top of the linear ramp, as we will derive in Sec.~\ref{sec:sff_manybody}.

Now we introduce a natural definition of irreducible SFF diagrams and topological equivalence classes of SFF diagrams, which allows one to compactly describe leading SFF diagrams. In addition, these definitions allow one to identify and prove the leading subleading SFF diagrams for {\gqmbcs} with TRS, which will be discussed in detail in an upcoming work~\cite{upcoming}.  Readers who are mainly interested in the non-technical results may skip to the next section. 

\begin{definitionalt}
    \textbf{Irreducible SFF diagrams.} An irreducible SFF diagram of an SFF diagram $\sigma \in S_t$ collapses a sequence of twisted contractions, i.e. $\sigma(i)= (m-i) \Mod{t}$  for some connected domain of $i \in [a,a+r]$ to a single twisted irreducible contraction, 
    as illustarted in Fig.~\ref{fig:sff_diag_v2} (a); and collapses a sequence of ladder contractions $\sigma(i)=(m'+i) \Mod{t}$ for some connected domain of $i \in [a',a' +r']$  to a single irreducible contraction as illustrated in Fig.~\ref{fig:sff_diag_v2} (b). The irreducible SFF diagram is labelled by a permutation $\tilde{\sigma} \in S_n$ of $n$ irreducible contractions with an additional vector $v \in \mathbb{Z}_2^n$ labelling whether the irreducible contractions are twisted or not. 
\end{definitionalt}

\noindent \textit{Examples.} Here we use the notation where $\sigma \in S_n$ is defined by $\sigma(\vec{a}) = (\sigma(a_1), \sigma(a_2), \dots, \sigma(a_n))$. Consider the SFF diagram $\sigma \in S_6$ in Fig.~\ref{fig:sff_diag} (b) defined by $\sigma(\vec{a}) =(6,5,4,3,2,1)$. Its irreducible SFF diagram is given in Fig.~\ref{fig:sff_diag_v2} (c) labelled by $\tilde{\sigma}\in S_1$ with $\tilde{\sigma}(1)= 1$ and $v=(1)$. Consider the SFF diagram  $\sigma \in S_6$ in Fig.~\ref{fig:sff_diag} (e) defined by $\sigma(\vec{a}) =(1,2,3,4,5,6)$. Its irreducible SFF diagram is given in Fig.~\ref{fig:sff_diag_v2} (d) labelled by $\tilde{\sigma}\in S_1$ with $\tilde{\sigma}(1)= 1$ and $v=(0)$. Consider rightmost SFF diagram $\sigma \in S_6$  in Fig.~\ref{fig:sff_diag_v2} (e) defined by $\sigma(\vec{a}) =(5,3,2,1,6,4)$. Its irreducible SFF diagram is given in the leftmost diagram in Fig.~\ref{fig:sff_diag_v2} (e) labelled by $\tilde{\sigma}\in S_2$ with $\tilde{\sigma}(1)= 1$ and $\tilde{\sigma}(2)= 2$, with $v=(0,1)$.

\begin{figure*}[ht]
    \centering
    \includegraphics[width=1 \textwidth]{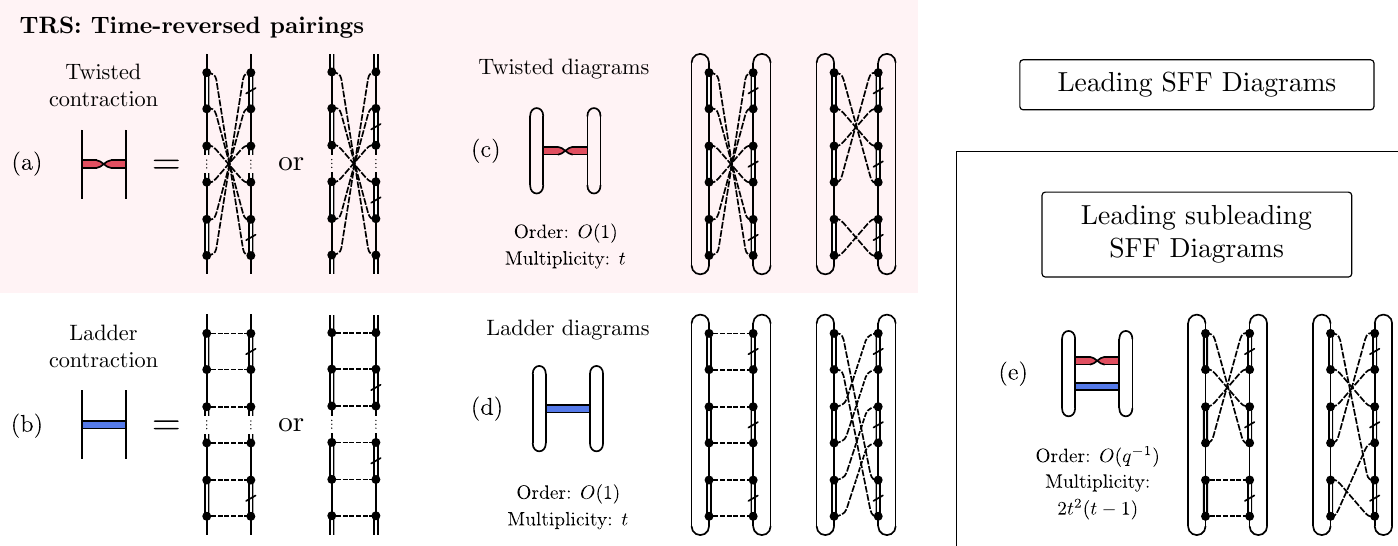}
    \caption{\textbf{Leading and leading subleading SFF  diagrams} can be expressed using the notation of (a) twisted contraction and (b) ladder contraction. The leading diagrams are (c) twisted diagrams with order $O(1)$ and multiplicity $t$; and (d) ladder diagrams with order $O(1)$ and multiplicity $t$. (e) The leading subleading diagram has a pair of twisted and ladder contraction with order $O(q^{-1})$ and multiplicity $2t^2 (t-1)$. On the right of (c), (d) and (e) two examples of each diagram are provided. The leading subleading diagram is proven and discussed in detail in an upcoming work~\cite{upcoming}.
    }
    \label{fig:sff_diag_v2}
\end{figure*}

\begin{definitionalt}\label{app_def:topo}
    \textbf{Topological equivalence classes of SFF diagrams.} Two SFF diagrams $\sigma, \sigma' \in S_n$ belong to the same \textit{topological equivalent class} if $\sigma'$ can be transformed from $\sigma$ under a sequence of operations $A_i$, i.e. $\sigma' = \prod_i A_{i} [\sigma]$. The permitted operations $A_i$ are
    \begin{itemize}
        \item[(i)] Cyclic rotation of $\Tr[U(t)]$: $A^{(r)}[\sigma](i)= \sigma((i+r)\Mod{n})$ for all $i$ and any $r$; 
        \item[(ii)] Reflection between $\Tr[U(t)]$ and $\Tr[U^\dagger(t)]$: $A[\sigma](i) =\sigma^{-1}(i)$; and 
        \item[(iii)] Twist of $\Tr[U^\dagger(t)]$ with respect to $\Tr[U^\dagger(t)]$: $A[\sigma](i)= \sigma((n+1-i)\Mod{n})$ for all $i$.
    \end{itemize}
\end{definitionalt}

\noindent \textit{Examples.} For (c) in Fig.~\ref{fig:sff_diag_v2}, we provide two diagrams topologically equivalent to the irreducible SFF diagrams labelled by $\sigma = \iden \in S_1$ with $v=(1)$. Similarly, for (d) in Fig.~\ref{fig:sff_diag_v2}, we provide two diagrams topologically equivalent to the irreducible SFF diagrams labelled by $\sigma = \iden \in S_1$ with $v=(0)$.

Under (iii) in \ref{app_def:topo}, the SFF diagrams in  Fig.~\ref{fig:sff_diag_v2} (c) and Fig.~\ref{fig:sff_diag_v2} (d) belong to the same equivalence class. However, to emphasise the many-body interactions between these diagrams (see later), we will discuss these diagrams separately. 
Note also that cyclic rotation of $\Tr[U^\dagger(t)]$ can be performed by using a combination of (i) and (ii).

\begin{lemmaalt}\label{app_lemma:sublead}
    \textbf{Leading subleading SFF diagrams with TRS and time translational symmetry.}  Consider the SFF \eqref{eq:sff_def} for (i) Floquet TRS random phase model (RPM), (ii) Floquet TRS Haar-random model (HRM), and (iii) the random matrix model (RMT) for Floquet TRS {\gqmbcs}.
For (i) and (ii) at each physical site in the order of the local Hilbert space dimension $q$, and for (iii)  in the order of matrix dimension $N$, the leading subleading SFF diagrams are of order $O(q^{-1})$ for (i) and (ii), and $O(N^{-1})$ for (iii). The leading subleading SFF diagrams are the topological equivalent class equivalent to
\be\label{app_eq:sublead}
\tilde{\sigma}= \iden \in S_2  \text { with } v=(0,1)\,,
\ee
as illustrated in Fig.~\ref{fig:sff_diag_v2} (e).
\end{lemmaalt}

\noindent\textit{Proof.} The proof for \ref{app_lemma:sublead}, for higher order subleading SFF diagrams are provided in an upcoming work~\cite{upcoming}.

\noindent \textit{Examples.} In Fig.~\ref{fig:sff_diag_v2} (e), we provide two diagrams belonging to the topological equivalence class of the subleading irreducible SFF diagram in Eq.~\eqref{app_eq:sublead}.
\vspace{0.2cm}

The leading subleading local SFF diagrams in Eq.~\ref{app_eq:sublead} and Fig.~\ref{fig:sff_diag_v2} (e)  indeed coincide with the Sieber-Ritchler pairs in the studies of periodic-orbit theory~\cite{Sieber_2001}, and ~\cite{upcoming} will discuss the signatures of such leading subleading SFF diagrams in the presence of quantum many-body interactions.

\subsection{Many-body quantum chaos}\label{sec:sff_manybody}
In this subsection, we systematically evaluate the SFF in the TRS random matrix models (RMT), random phase models (RPM) and  Haar-random models (HRM).
\subsubsection{RMT}
The Bohigas-Giannoni-Schmidt conjecture postulates that random matrix theory (RMT) behaviour emerges from chaotic system for sufficiently late time scales and sufficiently small energy scales. 
Consequently, the SFF behaviour of {\gqmbcs} converge to the SFF behaviour of  RMT at late time, and as we see below.
We define  the RMT models for different TRS dynamics  in Fig.~\ref{fig:model} and  App.~\ref{app:matrixmodel}. The SFF for RMT for local TRS and global TRS {\gqmbcs} are given by
\be\label{eq:rmt_non_floq}
\overline{K_{\text{TRS}}^{\mathrm{RMT}}(t, N)} 
=\begin{cases}
1 \, ,  \quad & \text{RMT for local TRS {\gqmbcs}}\,, \\
2 \frac{N}{N+1}  \, ,  \; \;  & \text{RMT for global TRS {\gqmbcs}}\,, 
\end{cases} 
\ee
for $t \neq 0$. The case of RMT for local TRS {\gqmbcs} can be trivially computed, while the derivation for RMT for global TRS {\gqmbcs} is given in App.~\ref{app:global_trs_mm_sff}. 
The SFF for RMT for Floquet TRS {\gqmbcs} is given by~\cite{Haake}, 
\be
\ba \label{eq:rmt_floq}
\nonumber
&\overline{K_{\text{f-TRS}}^{\mathrm{RMT}}(t, N)} \equiv 
\overline{K_{\text{COE}}^{\mathrm{RMT}}(t, N)} 
\\
\quad  =&
\begin{cases}
2t - t \sum_{m=1}^t \frac{1}{m + (N-1)/2} \, ,  \quad & 1 \leq t \leq N \,, \\
2N - t \sum_{m=1}^N \frac{1}{m + t - (N+1)/2} \, ,  \quad & N \leq t \,, 
\end{cases} 
\ea
\ee
for finite $N$ and $t\neq 0$, and the infinite-$N$ expression can be found in~\cite{Haake}.

\subsubsection{Local TRS}
In the absence of TRS or other symmetries, the SFF \eqref{eq:sff_def} of temporal-random many-body quantum circuits systems is trivially and exactly computed as unity~\cite{chan2021trans}.
In the presence of local TRS, the SFF for temporal-random systems can exceed unity, due to the additional twisted contractions that are local in space and time. However, such local twisted contractions are suppressed in $q$. 
In the limit of large-$q$, for the local TRS HRM, we obtain exactly
\be
\lim_{q \to \infty} \overline{K^{\text{HRM}}_{\text{l-TRS}} (t,L)}
=1 \,,
\ee
coinciding with the RMT result in Eq.~\eqref{eq:rmt_non_floq}.
This result is derived by identifying a single leading order many-body SFF diagram, whose local SFF diagram is given by $\sigma^{(m=1)}_{\mathrm{ladder}}$ according to Lemma~\ref{lemma:leading_sff_diag}.
For the local TRS RPM, a similar reasoning leads to
\be
\lim_{q \to \infty} \overline{K^{\text{RPM}}_{\text{l-TRS}}(t,L)}
=1 \,.
\ee
In Fig.~\ref{fig:sff_gtrs_num}(a) inset, we provide the finite-$q$ numerics of the local TRS HRM, and we observe that the local TRS SFF at finite-$q$ converges to the large-$q$ results after a time independent of the system sizes.

\begin{figure*}[ht]
    \centering
    \includegraphics[width=1\textwidth]{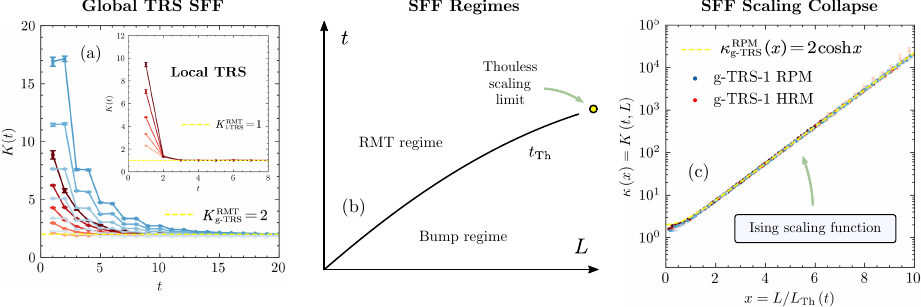}
    \caption{\textbf{SFF for Global and local TRS quantum many-body systems.} 
 (a) $K(t)$ versus $t$  for the 1D global TRS HRM (red) for $q=2$ and $L= 4,6,8,10,12$ and  1D global TRS RPM (blue) for $q=3$ and $L= 3, 4,5,6,7,8$  in increasingly dark shades. Both models have pbc and do not have local TRS. For both cases, the SFF converge to the RMT behaviour $K_{\text{g-TRS}}^\mathrm{RMT}\approx 2$ after the Thouless times, which increase in system size. (a) Inset:  $K(t)$ versus $t$  for 1D local TRS HRM for $L=4,6,8,10,12$ in increasingly dark shades. (b) SFF regimes in time $t$ and system size $L$ for {\gqmbcs} with  global TRS. In the RMT regime, the SFF of the spatially-extended {\gqmbcs} exhibit the RMT behaviour, while in the bump regime, the SFF takes values higher than the RMT ones due to the effects of many-body interactions. The Thouless scaling limit is the limit where $t$ and $L$ are sent to infinity, while $x=L/\Lth(t)$ is kept fixed. (c) Global TRS SFF scaling function $\kappa_{\mathrm{g-TRS}}(x) = K_{\mathrm{g-TRS}}(t,L)$ versus $x=L/\Lth(t)$ for the same models and parameters in (a). The scaling collapse for global TRS SFF numerical data obtained at finite $q$ shows excellent agreement with the Ising scaling function \eqref{eq:ising_scaling} obtained at infinite $q$. The scaling collapse for global TRS SFF with local TRS is provided in App.~\ref{app:numerics} and also shows agreement with the prediction.}
    \label{fig:sff_gtrs_num}
\end{figure*}

\subsubsection{Global TRS}

Here we evaluate that the SFF of global TRS HRM and RPM. According to Lemma~\ref{lemma:leading_sff_diag}, in the presence of global TRS, there are two leading local SFF diagrams  $\sigma_{\mathrm{twisted}}^{(m=2t)}$  [Fig.~\ref{fig:sff_diag} (b)]  or $\sigma_{\mathrm{ladder}}^{(m=1)}$  [Fig.~\ref{fig:sff_diag} (e)]  regardless of the presence of local TRS.
For the global TRS HRM with or without local TRS, we have
\be
   \lim_{q\to \infty}  \overline{K^{\text{HRM}}_{\text{g-TRS}} (t,L) }= 2\, ,
\ee
which counts two leading order many-body SFF diagrams, in which all local SFF diagrams either take $\sigma_{\mathrm{twisted}}^{(m=2t)}$ or $\sigma_{\mathrm{ladder}}^{(m=1)}$. Since many-body SFF diagrams containing two different local SFF diagrams will be suppressed in $q$~\cite{chan2018solution}. As such, the bump observed in SFF cannot be derived in the HRM in the large-$q$ limit, but it has been numerically observed in finite-$q$~\cite{chan2018spectral}. 
To circumvent this issue while retaining the analytical tractability in large-$q$, the RPM was introduced such that the coupling strength is controlled by the additional parameter $\epsilon$. The SFF of global TRS HRM without local TRS can be treated exactly in finite-$q$. This will be discussed together with other finite-$q$ analysis in \cite{upcoming}.

To evaluate the SFF in the global TRS RPM, we perform the ensemble averages over the one-site Haar-random gates, two-site coupling random phase gates, and one-site TRS-breaking random phase gates in sequence. Firstly, the Haar-random averages lead to two leading local SFF diagrams at each physical site as given by Lemma~\ref{lemma:leading_sff_diag}. 
As a result, we can define an effective classical degree of freedom $\sigma_i= \sigma_{\mathrm{ladder}}^{(m=1)}, \sigma_{\mathrm{twisted}}^{(m=2t)}$ at the $i$-th site of the original quantum circuit.
Secondly, we perform the ensemble average with respect to the random phase gates which act on two coupled sites $i$ and $j$. This involves performing Gaussian integrals on the phase contribution $\sum_{t'=1}^t [\varphi_{a(t') b(t')}(t') - \varphi_{a(\sigma_i(t')) b(\sigma_j(t'))}(t')]$. If $\sigma_i = \sigma_j$, the phases of the coupling gates from $\Tr[U(t)]$ and $\Tr[U^\dagger(t)]$ annihilate each other entirely, leading simply to a factor of $1$. If $\sigma_i \neq \sigma_{j}$, $2\alpha(t)$ terms in the summation of phases remain, leading to a factor of $e^{-\epsilon \alpha(t)}$ upon the Gaussian integral in the large-$q$ limit, where $\alpha(t) =  t-1 - \delta_{0, t\Mod{2}}$. Thirdly, we perform the ensemble average over the TRS-breaking gates  with $\sigma_i$ at site $i$. Similar to the previous step, this involves performing Gaussian integrals on the phase contribution $\sum_{t'=1}^t [\phi_{a(t')}(t') - \phi_{a(\sigma_i(t'))}(t')]$, which gives $1$ if $\sigma_i = \sigma^{(1)}_{\mathrm{ladder}}$ and $e^{-b \alpha(t)}$ if $\sigma_i = \sigma^{(2t)}_{\mathrm{twisted}}$. 
In Appendix \ref{app:sff_diag_proof}, we present an alternative diagrammatic representation of the quantum-classical mapping using a folding construction [Fig.~\ref{fig:fold}]. In this approach, the unitary circuits are folded so that all identically-sampled unitary gates are  stacked on top of each other. This formulation makes it clear that the leading local SFF diagrams in Lemma~\ref{lemma:leading_sff_diag} correspond to the identity and SWAP elements of the symmetric group $S_2$.
Lastly, note that the above derivation can be applied to the RPM with general geometry in arbitrary dimension. The SFF of this model in higher dimension will be discussed in~\cite{upcoming}. 

The two effective Ising degrees of freedom $s_i$ at the $i$-th site are
\begin{enumerate}
    \item $s_i=1$ denoting the time-parallel pairing of Feynman paths in SFF given by $\sigma^{(1)}_{\mathrm{ladder}}$, and 
    \item $s_i=-1$ denoting the time-reversed pairing of Feynman paths in SFF given by $\sigma^{(2t)}_{\mathrm{twisted}}$.
\end{enumerate}
See Table~\ref{Tab:sff_3approaches} and later sections for comparisons of the effective Ising degrees of freedom with the other approaches.
For global TRS RPM in a general geometry in arbitrary dimension, we have exactly
\be \label{eq:ising_mapping}
\ba
\lim_{q\to \infty}  &\overline{K^{\text{RPM}}_{\text{g-TRS}}(t,L)} 
\\
=& \sum_{\{s_i \}} \exp\left[ - \epsilon t \sum_{\langle i, j \rangle} (1- \delta_{s_i , s_j})  - b t \sum_i  \delta_{s_i, -1 } \right]  ,
\ea
\ee
in the large-$q$ limit, and therefore, we write
\be
\ba\label{eq:ising_model}
\lim_{q\to \infty}  \overline{K^{\text{RPM}}_{\text{g-TRS}}(t,L,\epsilon)} \, \propto&
\, \, Z_{\text{Ising}}(\beta J \propto \epsilon t, \beta h \propto bt, L) \,,
\\
\text{with }\, \, H_{\text{Ising}}:= & \,  J \sum_{\langle i, j \rangle} s_i s_j  + h \sum_i  s_i \, ,
\ea
\ee
i.e. the SFF of global TRS RPM is dual to the partition function of a classical Ising model with ferromagnetic interactions.
Note that the SFF of RPM with discrete time translational symmetry but without TRS maps to an emergent $t$-state Potts model, where the number of state on each site in the emergent problem increases with $t$. In contrast, the mapping of SFF for global TRS RPM strips off even the feature of the linear ramp, making it an ideal test ground to understand SFF of {\gqmbcs}~\cite{upcoming}.

%
%
%
In particular, for one-dimensional global TRS RPM (with or without local TRS), the SFF  can be evaluated exactly using the Ising model in large-$q$ limit as 
\begin{widetext}
\begin{equation}\label{eq:trs_psff_gtrs_full}
\ba
\lim_{q\to \infty} \overline{K^{\text{RPM}}_{ \text{g-TRS}}(t,L)}
= 
\begin{cases}
\lambda _1^L+\lambda _2^L \quad  & \text{pbc,}
   \\
  \frac{1}{1+ e^{-\epsilon  \alpha (t)}} 
 \left( \lambda _1^L+\lambda _2^L \right)
 +
  \frac{ e^{-\epsilon  \alpha (t)} +e^{-(b+\epsilon) \alpha (t)} }{\zeta (1+e^{-\epsilon  \alpha (t)}) } \left( \lambda _1^L-\lambda _2^L \right)
   \quad \quad \quad  &\text{obc,}
\end{cases}
\ea
\end{equation}
\end{widetext}
for the periodic boundary condition (pbc) and open boundary condition (obc). Here, $\alpha(t) = t-1 - \delta_{0, t \, (\mathrm{mod} 2)}$, and $\lambda_1$ and $\lambda_2$ are given by
\be\label{eq:eval_ising_mag}
\ba
\lambda_1 =& \,\, \frac{1}{2} \left(1+e^{-b\alpha(t)}+\zeta \right) \, , 
\\
\lambda_2 =& \, \,  \frac{1}{2} \left(1+e^{-b\alpha(t)}-\zeta\right) \, ,
\\
\zeta= & \, \,   \sqrt{4 e^{-b \alpha (t)-2 \epsilon  \alpha (t)}+e^{-2 b \alpha (t)}-2 e^{-b \alpha (t)}+1} \, .
\ea
\ee
For global TRS RPM with $b=0$, i.e. the TRS-breaking mechanism is turned off, SFF simplifies to 
\begin{equation} \label{eq:sff_gtrs_nomag}
\begin{aligned}
   \lim_{q\to \infty}  & \overline{K^{\text{RPM}}_{\text{g-TRS}}(t,L)} 
   \\
   =& \, \, 
   \begin{cases}
\left[1+e^{-\epsilon \alpha (t)}\right]^L + \left[1-e^{-\epsilon \alpha (t)}\right]^L\,  \quad \quad  & \text{pbc,}  \\
2 \left[1 + e^{-\epsilon \alpha (t)}\right]^{L-1} \,   \quad  & \text{obc.}
   \end{cases} 
\end{aligned}
\end{equation}
We define \textit{Thouless time} as the time after which  the {\gqmbcs}  behaves indistinguishably from the corresponding RMT in terms of its connected SFF. From \eqref{eq:sff_gtrs_nomag}, we see that the SFF of the global TRS {\gqmbcs} approaches the RMT value $K^{\mathrm{RMT}}_{\text{g-TRS}}\approx 2$ in Eq.~\eqref{eq:rmt_non_floq} after $\tth(L) = \ln L / \epsilon$.
Let us interpret the behaviour of SFF in {\gqmbcs} in terms of the emergent classical statistical mechanical problem with reference to the Thouless time. 
For $ 0 \lesssim t \ll \tth(L)$, all Boltzmann weights are close to unity regardless of the configurations. Therefore, all configurations of the Ising model contribute similarly to the partition function, resulting in an SFF that exceeds the RMT value. 
For $t \gg \tth(L)$,  the cost of having a domain wall is so large that the dominant contributions to the partition function are the two ferromagnetic ground states, thus recovering the RMT value. 
As $t$ approaches $\tth(L)$ from above, domain wall configurations begin to contribute to the SFF, and the SFF begins to deviate from the RMT behaviour. 

To test the universality of the result in the \eqref{eq:sff_gtrs_nomag}, we take the thermodynamic limit, where the microscopic details of the model should become irrelevant. If $t, L$ are sent large, while keeping $t\gg \tth(L)$, the SFF is in the familiar ramp-plateau regime. If $t, L$ are sent large, while $t$ is fixed, SFF displays exponential growth characterised by the spectral Lyapunov exponent studied in~\cite{cdclyap}. Instead, we can take the \textit{Thouless scaling limit}~\cite{chan2021trans}, where $t, L$ are sent to infinity, with $L/\Lth(t)$ fixed where the \textit{Thouless length} $\Lth(t)= e^{\epsilon t}$ is the inverse function of $\tth(L)$. This limit is illustrated in the SFF regime diagram of global TRS {\gqmbcs} represented by a yellow dot in Fig.~\ref{fig:sff_gtrs_num} (b). Indeed, universal behaviour in SFF and other observables have been observed in this limit in~\cite{chan2021trans, Shivam_2023, christopoulos2024universaldistributionsoverlapsunitary}.
The exact solution in one-dimensional global TRS RPM allows us to compute a scaling function at the Thouless scaling limit,
\be\label{eq:ising_scaling}
\kappa^{\text{SFF}}_{\text{g-TRS}}(x)\equiv    \lim_{\substack{L,t \to \infty\\ x = L/\Lth}}  \overline{K^{\text{RPM}}_{\text{g-TRS}}} =\begin{cases}
2 \cosh x \, ,  \quad & \text{pbc}\,, \\
2 \, e^{x} \, , \quad & \text{obc} \,.
\end{cases} 
\ee
This is the first route to obtain the universal Ising scaling behaviour of SFF in global TRS {\gqmbcs}, and we obtain the same result in two other routes, see  Eq.~\eqref{eq:ginibre_scaling_form} and Eq.~\eqref{eq:ising_scaling_2paf}.
Numerically, we compute the SFF for global TRS HRM at $q=2$ and global TRS RPM at $q=3$ in Fig.~\ref{fig:sff_gtrs_num} (a), where we observe a convergence to the SFF RMT behaviour at sufficiently late time, and that the Thouless time increases with the system sizes for both models. 
In Fig.~\ref{fig:sff_gtrs_num} (c), we observe an excellent agreement between the scaling function \eqref{eq:ising_scaling} derived from the Ising model in large-$q$, and the numerical data from global TRS HRM and RPM obtained at finite-$q$. 

\begin{figure*}[ht]
    \centering
    \includegraphics[width=1 \textwidth]{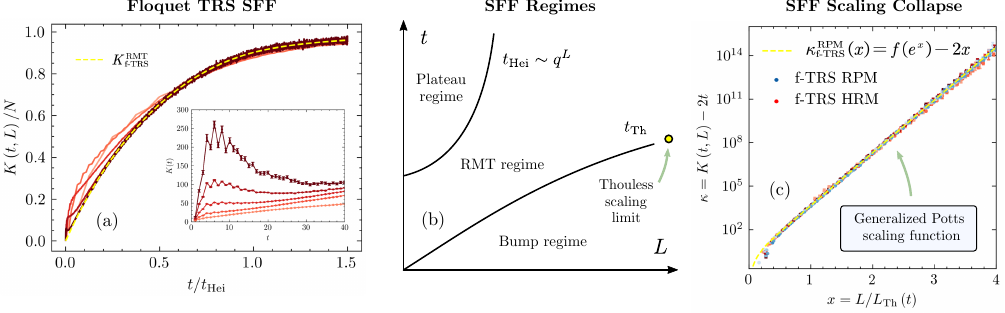}
    \caption{\textbf{SFF for Floquet TRS quantum many-body systems.} 
 (a) The rescaled SFF $K(t)/N$ versus $t/\thei$ for the 1D Floquet TRS HRM for $L= 6,8,10,12,14$ in increasingly dark shades. $N=\thei = q^L$ is the Hilbert space size and the Heisenberg time. SFF has been rescaled to demonstrate convergence to the RMT behaviour after the Thouless times, which increases in system size without rescaling, as shown from the plot of  $K(t)$ versus $t$ in the inset. 
(b) SFF regimes in time $t$ and system size $L$ for {\gqmbcs} with  global TRS. In the RMT and plateau regimes, the SFF of the spatially-extended {\gqmbcs} exhibit the RMT behaviour, while in the bump regime, the SFF takes values higher than the RMT ones due to the effects of many-body interactions. The Thouless scaling limit is the limit where $t$ and $L$ are sent to infinity, while $x=L/\Lth(t)$ is kept fixed.
  (c) Floquet TRS SFF scaling function $\kappa_{\mathrm{f-TRS}}(x) = K_{\mathrm{f-TRS}}(t,L)- 2t$ at odd times versus $L/\Lth(t)$ for the 1D Floquet TRS HRM with pbc for $q=2$ and $L= 4,6,8,10,12$, and for the 1D Floquet TRS with pbc for $q=3$ and $L=3,4,5,6,7,8$ in increasingly dark shades. 
    The scaling collapse for data obtained at finite-$q$ shows excellent agreement with the generalized Potts model scaling function \eqref{eq:ftrs_pbc_scaling} obtained at infinite-$q$.  }
    \label{fig:sff_ftrs_num}
\end{figure*}

\subsubsection{Floquet TRS}
Here we evaluate that the SFF of Floquet TRS HRM and RPM.
According to Lemma~\ref{lemma:leading_sff_diag}, in the presence of Floquet TRS, there are $2t$ leading local SFF diagrams: $t$ twisted diagrams $\sigma_{\mathrm{twisted}}^{(m)}$ with $m=2,4,\dots, 2t$ [e.g. Fig.~\ref{fig:sff_diag} (d,c,b)]  and  $t$ ladder diagrams $\sigma_{\mathrm{ladder}}^{(m)}$ with $m=1,3,\dots, 2t-1$ [e.g. Fig.~\ref{fig:sff_diag} (e,f,g)]. 
Use the notation $\sigma^{(m)}(i)= m+ (-1)^m (-i +1)$ which  combines \eqref{eq:twisted_sff} and \eqref{eq:ladder_sff}. 
We define a basis with these $2t$ SFF diagrams on each site,
$\omega = \{\sigma^{(1)}, \sigma^{(3)}, \dots \sigma^{(2t-1)}, \sigma^{(2t)}, \sigma^{(2t-2)}, \dots, \sigma^{(2)} \dots \}$, so that $\omega(a) = \sigma^{(2a-1)}$ for $a= 1,2,\dots, t$, and $\omega(a) = \sigma^{2(2t+1-a)}$ for $a= t+1,t+2,\dots, 2t$.
These $2t$ diagrams represent the effective degrees of freedom of the statistical mechanical model at each site. 

For the Floquet TRS HRM, using methods described above, we can directly evaluate the SFF as
\be
   \lim_{q\to \infty}  \overline{K^{\text{HRM}}_{\text{f-TRS}} (t,L) }= 2t\, ,
\ee
recovering the leading order behaviour of the SFF RMT behaviour in \eqref{eq:rmt_floq}. This result follows from the fact that many-body SFF diagrams in HRM composed of multiple kinds of local leading diagrams are suppressed in $q$, and therefore all the leading order SFF contributions must occupy the same local diagram at all sites.

For Floquet TRS RPM, following the identical approach described 
above Eq.~\eqref{eq:ising_mapping}, we derive the quantum-classical mapping for Floquet TRS SFF for general geometries,
\be\label{eq:floq_trs_mapping}
\ba
&\lim_{q\to \infty}  \overline{K^{\text{RPM}}_{\text{f-TRS}}(t,L)}   = Z_{2t\text{-gPotts}}
(\beta J \propto \epsilon t) \, ,
\ea
\ee
where $Z_{2t\text{-gPotts}}= \sum_{\{ a_i=1, \dots, 2t \}} \prod_{\langle i,j \rangle} \mathcal{B}_{a_i a_j}$  denotes the partition function of a  generalized  classical Potts model. In this Potts model, each site possesses $2t$ degrees of freedom, labeled as $a=1,2,\dots, 2t$  according to the basis above. The product extends over all lattice bonds where the RPM have 2-site coupling gates. The Boltzmann weight for the bond between site $i$ and $j$ is given by $\mathcal{B}_{a_i a_j}$, with two distinct types:
\begin{itemize}
\item \textbf{Domain walls of type I} are the domain walls between two time-parallel pairings of Feynman paths, i.e. two ladder diagrams, or domain walls between two time-reversed pairings, i.e. two twisted diagrams. The corresponding Boltzmann weights $\mathcal{B}_{a b}$ are given by 
\be \label{eq:dw_typ1}
 A_{ab} = \; \delta_{ab} + (1 - \delta_{ab}) \, e^{-\epsilon t} \,, 
\ee
with $a, b= 1,2,\dots t$ for  ladder-ladder interactions, or with $a, b= t+1, t+2,\dots 2t$ for  twisted-twisted interactions. Note that $A_{ab}$ are the Boltzmann weights that describe the emergent $t$-state Potts model from the SFF of {\gqmbcs} in the absence of symmetries~\cite{chan2018spectral}.

\item \textbf{Domain walls of type II} are domain walls between a time-parallel (ladder diagram) and a time-reversed (twisted diagram) pairing of Feynman paths, with the Boltzmann weight given by  
\be\label{eq:dw_typ2}
    [B(t)]_{ab} = \;
    \begin{cases}
    e^{-\epsilon (t-1)}, 
    & \quad t \;\; \text{odd}\,,
       \\
    e^{-\epsilon t + 2 \epsilon [(a-b)\mod 2]}, & \quad  t \;\; \text{even}\,.
        \end{cases}
        \ee
 where $a \in  \{ 1,2,\dots t \}$ and $b \in  \{ t+1, t+2,\dots 2t \}$, or $a \in  \{ t+1, t+2,\dots 2t \}$ and $b \in  \{ 1,2,\dots t \}$.  
\end{itemize}
Now we evaluate the SFF of Floquet TRS RPM in the one-dimensional setting
\be \label{eq:k_floq_trs}
\ba
\lim_{q\to \infty} \overline{K^{\text{RPM}}_{\text{f-TRS}}(t,L)} =& 
\begin{cases}
 \tr[T_{\mathrm{gPotts}}^L] \,,  & \text{ pbc} \,,
    \\
\bra{\eta} T_{\mathrm{gPotts}}^{L-1} \ket{\eta}  \,,  & \text{ obc} \,,
\end{cases}
\ea
\ee
where $\ket{\eta} = (1,1,\dots, 1)$, and  $T_{\mathrm{gPotts}}$ is a $2t$-by-$2t$ transfer matrix   between site $i$ and $i+1$ given by
\be
\ba
T_{\mathrm{gPotts}} = & \, 
 \begin{pmatrix}
    A & B \\
    B & A 
    \end{pmatrix} \, .
    \ea
\ee
For obc, we have for  odd $t>0$, 
\be \label{eq:ftrs_sff_oddt_obc}
\ba
&
\lim_{q\to \infty} 
\overline{K^{\text{RPM}}_{\text{f-TRS}}(t \in 2\mathbb{Z}+1)} 
\\
& \qquad = 2 t \left[1+ 
t e^{-\epsilon(t-1)}
+ (t-1) e^{-\epsilon t}
\right]^{L-1}\,,
\ea
\ee
and even $t>0$,
\be
\ba
& 
\lim_{q\to \infty} \overline{K^{\text{RPM}}_{\text{f-TRS}}(t \in 2\mathbb{Z})}
\\
& \qquad 
= 2 t \left[1+ \frac{t}{2}  e^{-\epsilon (t-2)}+\left(\frac{3 t}{2}-1\right) e^{-\epsilon t}\right]^{L-1} \,. 
\ea
\ee
For pbc, for odd $t$,  we have 
\begin{equation}\label{eq:k_floq_scal_oddt}
\ba
&  
\lim_{q\to \infty} 
\overline{K^{\text{RPM}}_{\text{f-TRS}}(t\in 2\mathbb{Z} + 1,L)} =
(2 t-2) \left(1-e^{-\epsilon t}\right)^L
\\ & \qquad +\left[1-t e^{-\epsilon (t-1)}+(t-1) e^{-\epsilon t} \right]^L
\\ & \qquad +\left[1+t e^{-\epsilon (t-1)}+(t-1) e^{-\epsilon t}\right]^L
 \;,
   \ea
\end{equation}
while for even~$t$, we have 
\begin{equation} \label{eq:ftrs_event_pbc}
\ba
& 
\lim_{q\to \infty} 
\overline{ K^{\text{RPM}}_{\text{f-TRS}}(t \in 2\mathbb{Z},L)} = (2 t-4) \left(1-e^{-\epsilon t}\right)^L
  \\
  & \qquad  +2 \left[1-\frac{t}{2}  e^{-\epsilon (t-2)}+\left(\frac{t}{2}-1\right) e^{-\epsilon t}
  \right]^L
  \\
  & \qquad  +
  \left[1+\frac{t}{2}  e^{-\epsilon (t-2)}
  -\left(\frac{t}{2}+1\right) e^{-\epsilon t}\right]^L 
    \\
  & \qquad  +\left[1+ \frac{t}{2}  e^{-\epsilon (t-2)}+\left(\frac{3 t}{2}-1\right) e^{-\epsilon t}\right]^L
  \;.
\ea
\end{equation}
Again, the Thouless time is  the time after which  the {\gqmbcs}  behaves indistinguishably from the corresponding RMT in terms of its connected SFF, and in this case, $\tth(L)$ is determined by $\tth L e^{-\epsilon \tth} = O(1)$, and its inverse function, the Thouless length, is given by $\Lth(t) = e^{\epsilon t}/t$.  
The SFF can be interpreted with respect to $\tth(L)$ in terms of the statistical mechanical problem  as in the global TRS case, except that the emergent statistical mechanical problem has $2t$ states per site. 
For $t \gg \tth(L)$, the Boltzmann weights of any configurations other than the $2t$ ferromagnetic ground states are small, and hence we recover the leading behaviour of $2t$ ramp of SFF. 
For $ 0 \lesssim t \ll \tth(L)$, all Boltzmann weights are close to unity regardless of the configurations. Therefore, all configurations of the generalized Potts model contribute similarly to the partition function, resulting in an SFF that exceeds the RMT value. 
As $t$ approaches $\tth(L)$ from above,  configurations with domain walls of type I and type II begin to contribute to the SFF, and the SFF begins to deviate from the RMT behaviour.

As in the global TRS case, we look for universal behaviour in the Thouless scaling limit \cite{chan2021trans} where $t$ and $L$ are sent to infinity, while $x= L/ \Lth(t)$ is fixed.  
For the obc, we have  the scaling form for Floquet TRS RPM as 
\be \label{eq:ftrs_obc_scaling}
\kappa^{\text{RPM}}_{\text{f-TRS}}(x)\equiv    \lim_{\substack{L,t \to \infty\\ x = L/\Lth}}  \frac{ \overline{K_{\text{f-TRS}}} }{2t} =
g_\epsilon(e^x)
\,. 
\ee
where $g_\epsilon(y) = y^{(1+e^{\epsilon})}$ for odd $t$, and $g_\epsilon(y) = y^{\frac{1}{2}(3+e^{2\epsilon})}$ for even $t$.
For the pbc, we obtain the scaling form for Floquet TRS RPM as
\be \label{eq:ftrs_pbc_scaling}
\kappa^{\text{RPM}}_{\text{f-TRS}}(x)\equiv   \lim_{\substack{L,t \to \infty\\ x = L/\Lth}} \overline{K^{\text{RPM}}_{\text{f-TRS}}} -2t =
f_\epsilon(e^{x}) -2x
\,, 
\ee
where $f_\epsilon(y)= y^{1+e^{\epsilon}} + y^{1-e^{\epsilon}} -2 $ for odd $t$, and 
$f_\epsilon(y)= 2 y^{\frac{1}{2}(1-e^{2\epsilon})} + 
y^{\frac{1}{2}(e^{2\epsilon} -1)}+
y^{\frac{1}{2}(3+e^{2\epsilon})} -4$ for even $t$.
Numerically, we compute the SFF for Floquet TRS HRM at $q=2$ and Floquet TRS RPM at $q=3$ in Fig.~\ref{fig:sff_ftrs_num} (a), where we observe a convergence to the SFF RMT behaviour at sufficiently late time, and that the Thouless time increases with the system sizes for both models, as shown in the inset. 
In Fig.~\ref{fig:sff_ftrs_num} (c), we observe an excellent agreement between the scaling function \eqref{eq:ftrs_pbc_scaling} derived from the generalized Potts model in large-$q$, and the numerical data from global TRS HRM and RPM obtained at finite-$q$.

These result should be compared with the results obtained for Floquet RPM~\cite{chan2018spectral} which serves as a minimal model for many-body quantum chaotic systems without symmetries. This model is defined similarly to the Floquet TRS RPM, except that the COE gates are replaced with CUE gates. The SFF is given by
\be \label{eq:sff_cue}
\ba
&\lim_{q\to \infty} \overline{K^{\text{RPM}}_{\text{no-sym}}(t,L)} 
\\
& 
= 
\begin{cases}
(t-1)(1-e^{-\epsilon t})^L + [1+(t-1)e^{-\epsilon t}]^L\,,  & \text{ pbc} \,,
    \\  t [1+(t-1) e^{-\epsilon t}]^{L-1} \,,  & \text{ obc} \,.
\end{cases}
\ea
\ee
Taking the Thouless scaling limit with the Thouless length $\Lth(t) = e^{\epsilon t}/t $, we obtain the scaling form of SFF for {\gqmbcs} without TRS with the pbc as~\cite{chan2021trans} 
\be \label{eq:nosym_pbc_scaling}
\kappa^{\text{RPM}}_{\text{no-sym}}(x)\equiv    \lim_{\substack{L,t \to \infty\\ x = L/\Lth}}   \overline{K^{\text{RPM}}_{\text{no-sym}}} -t  = e^x - x -1 \, ,
\ee
and for obc, we have the scaling form
\be\label{eq:nosym_obc_scaling}
\kappa^{\text{RPM}}_{\text{no-sym}}(x)\equiv   \lim_{\substack{L,t \to \infty\\ x = L/\Lth}}   \frac{ \overline{K^{\text{RPM}}_{\text{f-TRS}}} }{t} =
e^{x}
\,.
\ee
In comparison with case without TRS, due to the many-body interactions with time-reversed pairings of Feynman paths, the presence of TRS leads to additional factors of $2$ appear with $t$ and $x$ in ~\eqref{eq:ftrs_pbc_scaling} and ~\eqref{eq:ftrs_obc_scaling} compared to  in ~\eqref{eq:nosym_pbc_scaling} and ~\eqref{eq:nosym_obc_scaling}, and that the type I and II domain walls modify the scaling forms via the functions $f$ and $g$.

\subsection{Ginibre ensemble}

\begin{figure}[ht]
    \centering
    \includegraphics[width=0.4 \textwidth]{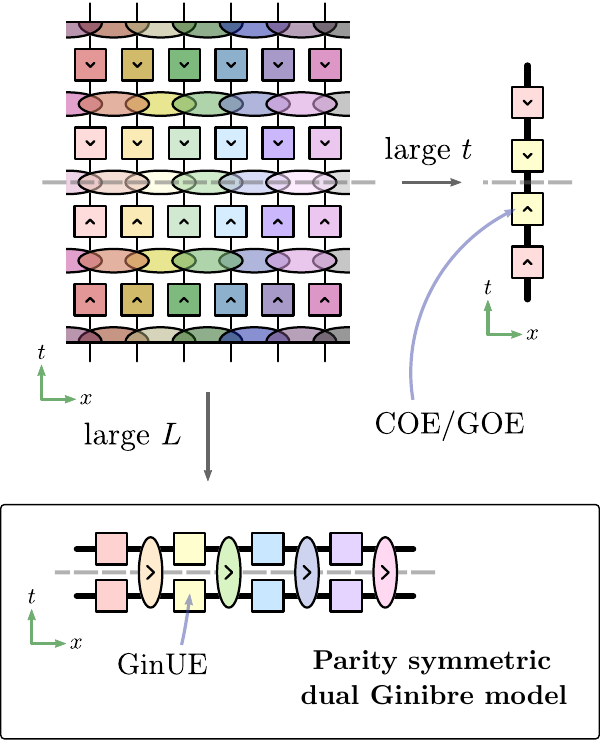}
    \caption{\textbf{TRS Dual Ginibre models are parity symmetric.} For sufficiently large time $t$, the TRS random quantum circuits displays behaviour that can be modelled by a RMT model of unitary or Hermitian matrices drawn from the COE or GOE (right). In sufficiently large system size $L$, the TRS quantum circuits can be modelled by the dual Ginibre model with spatial parity symmetry, that is, a RMT model of non-Hermitian matrices drawn from the Ginibre ensemble that is parity symmetric in the dual Hilbert space (bottom).  The dashed lines refer to the TRS inversion axes. 
    }
    \label{fig:ginue}
\end{figure}

The Ginibre ensemble is a subclass of RMT that contains random matrices with independent complex Gaussian matrix elements~\cite{Mehta}. In particular, unlike the Circular and Gaussian ensembles of RMT, unitarity and hermiticity conditions are not imposed in the Ginibre ensemble.  
In quantum many-body physics, space-time duality or rotation refers to the fact that one can learn about a quantum many-body system by studying the spatial propagation of circuit in space, rather than the unitary time evolution in time~\cite{guhr2016kim, bertini2018exact,chan2018spectral, chan2020lyap, lerose2020influence, Ippoliti_2022}. 
It has recently been shown that the dual transfer matrix of spatially-extended generic many-body quantum chaotic systems belongs to the class of universality class of the Ginibre ensemble~\cite{Shivam_2023}. This emergence provides a new set of tools to tackle the problem of quantum many-body chaos, leading to analytical results on understanding the spectral form factor of generic and local spatially-extended quantum many-body systems beyond the infinite-$q$ approximation. 

For simplicity, we focus on the dual Ginibre model [Fig.~\ref{fig:ginue}] for the global TRS {\gqmbcs}, one of the simplest setting with TRS.
Consider a system with two (dual) coarse-grained sites, each with Hilbert space size $\mathbb{C}^{\NN}$ such that the total Hilbert space has the size $\mathbb{C}^{\NN^2}$. 
The non-unitary evolution in the spatial direction is defined by
\be\label{eq:gin_model}
\ba
V_{ \text{g-TRS}}^{\mathrm{Gin}}(L) = &\,\, \prod_{x=1}^L v(x) \, ,
\quad v(x) =  v_2 (x) v_1 (x) \,,
\\
v_1 (x) =& \, \, G_1(x;\NN) \otimes G_1(x;\NN)\,,
\\
v_2 (x) =& \,\,   G_2(x;\NN^2) + S G_2(x;\NN^2)  S  \, ,
\ea
\ee
where $G_2(x;{N})$ and $G_1(x;\NN^2)$ are independently drawn from the complex Ginibre ensemble of $\NN$-by-$\NN$ non-Hermitian random matrices, such that $\overline{[G_i]_{a a'}(x;\NN ) [G^*_i]_{b b'}(x;\NN)}=\delta_{ab} \delta_{a' b'} \sigma_i^{-2}(\NN)$ with $\NN$-dependent variance $\sigma_1^2=\NN$ and $\sigma_2^2=2\NN^2$.
$S$ is the two-site swap operator, defined by its action $S\ket{a_1, a_2}= \ket{a_2, a_1}$ where $\ket{a_1,a_2}$ with $a_i = 1,2,\dots, N$ is the computational basis in the Hilbert space. 
 For spatially-random systems, $v(x)$ does not correlate with $v(x')$ for $x\neq x'$. 
To incorporate global TRS, we crucially choose the dual Ginibre model to be parity symmetric along the spatial direction, i.e. $Sv(x)S^{-1} = v(x)$.
Note that dual Ginibre models can be constructed for Floquet  and/or translational invariant {\gqmbcs}. See details in ~\cite{Shivam_2023}.

The SFF of the dual Ginibre model for global TRS {\gqmbcs} is defined by
\be\label{eq:sff_gin}
\begin{aligned}
K^{\mathrm{Gin}}_{\text{g-TRS}}(N, L) =
\begin{cases}
\left| \Tr [V_{ \text{g-TRS}}^{\text{Gin}}(L) ] \right|^2  \, ,  \quad & \text{pbc,}
\\
\left| \bra{\eta} [V_{ \text{g-TRS}}^{\text{Gin}}(L)  \ket{\eta} \right|^2 \, ,  \quad & \text{obc,}
\end{cases}
\end{aligned}
\ee
where in the dual or space-time-rotated set-up, $\ket{\eta} = (1,1)$ accounts for the obc, and the trace accounts for the pbc in the original setting. 
Firstly, we perform the ensemble average over $G_1(x)$ in \eqref{eq:gin_model}. This average is performed over a pair of $G_1(x)$ and a pair of $G_1^\dagger (x)$, leading to two possible Wick contractions: Identity or SWAP, which we label with 0 and 1 respectively. See Fig.~\ref{fig:three_routes} for a schematic illustration, where circled dots represent $G_1(x)$, and the dashed lines represent Wick contractions. Using these two Wick contractions as the basis states of a transfer matrix, we can write the SFF without using the large-$N$ approximation as  
\be\label{eq:gtrs_gin_step1}
\ba
\overline{K_{ \text{g-TRS}}^{\mathrm{Gin}}(N, L)} = & \, \, \Tr \left[ (T_{ \text{g-TRS}}^{\mathrm{Gin}})^L \right] \,, 
\\
[T_{ \text{g-TRS}}^{\mathrm{Gin}}]_{ab} = & \, \, \frac{1}{\sigma_1^2}\overline{\Tr \, ( v_2^\dagger S^{a} v_2 S^{b} )} \,,
\ea
\ee
where $a,b = 0,1$ are the indices for the basis states, and  $S$ the two-site swap operator with $S^0 \equiv \mathbb{1}$.
Using $\overline{\Tr \, (G_2 G_2^*)}= N^4/ \sigma^2_2$, $\overline{\Tr \, (SG_2 SG_2^*)}=N^2 /\sigma^2_2$, and $\overline{\Tr \, (G_2 G_2^*S)}= \overline{\Tr \, (G_2 SG_2^*)}= N^3 /\sigma^2_2$, we obtain
\begin{equation}
\ba
    T_{ \text{g-TRS}}^{\mathrm{Gin}}(N) 
    = 
    \begin{pmatrix}
    1 + \frac{1}{N^2} &  \frac{2}{N}  \\
     \frac{2}{N}  & 1 + \frac{1}{N^2}
    \end{pmatrix} 
    \;,
    \ea
\end{equation}
which coincides with the transfer matrix for ferromagnetic Ising model. 
The effective Ising degrees of freedom are
\begin{enumerate}
    \item The identity contraction along the dual spatial direction; and 
    \item The parity SWAP contraction along the dual spatial direction.
\end{enumerate}
See Table~\ref{Tab:sff_3approaches} for a comparison of the effective Ising degrees of freedom with the other approaches.
Using the Ising transfer matrix, we obtain
\be
\ba
& \overline{K_{ \text{g-TRS}}^{\mathrm{Gin}}(N, L)} \, \,
\\
& \quad =
\begin{cases}
    \left(1 +\frac{2}{N}+ \frac{1}{N^2}\right)^L+\left(1-\frac{2}{N}+ \frac{1}{N^2}\right)^L, \quad & \text{pbc,}
    \\
    2 \left( 1 + \frac{2}{N} + \frac{1}{N^2}\right)^{L-1}, \quad & \text{obc,}
\end{cases}
\ea
\ee
which tends to 2 in large $N$ as expected.
In the dual Ginibre model, the Thouless length can be read off as $L_{\mathrm{Th}}= N/2$.  Taking the thermodynamic and Thouless scaling limit where $x= L/ L_{\mathrm{Th}}$ is fixed but $N$ and $L$ are sent to infinity, we arrive
\be\label{eq:ginibre_scaling_form}
\kappa^{\text{Gin}}_{\text{g-TRS}}(x)\equiv  \lim_{\substack{L,N \to \infty\\ x = L/\Lth}}  \overline{ K^{\text{Gin}}_{\text{g-TRS}} }=\begin{cases}
2 \cosh x \, ,  \quad & \text{pbc,} \\
2 \, e^{x} \, , \quad & \text{obc,}
\end{cases} 
\ee
and therefore, in the presence of global TRS, the Ginibre SFF scaling function coincides with the global RPM SFF scaling function, i.e. $\kappa^{\text{SFF}}_{\text{g-TRS}}(x) = \kappa^{\text{Gin}}_{\text{g-TRS}}(x)$. This is the second route to derive the emergent Ising scaling behaviour of SFF in {\gqmbcs} with global TRS. 

The use of Ginibre universality is, as far as the authors know, the only finite-$q$  methods known to describe the spectral statistics of generic, local, and spatially-extended many-body chaotic systems. This method therefore complements the large-$q$ analysis done with the RPM and HRM. Beyond spectral statistics, the use of Ginibre ensembles have also found other applications in generic quantum many-body systems~\cite{christopoulos2024universaldistributionsoverlapsunitary, Bulchandani_2024, deluca2024universalityclassespurificationnonunitary}.

\section{Partial spectral form factor}\label{sec:psff}
The SFF is the modulus square of the trace of the time evolution operator for a quantum system of interest. The partial spectral form factor is a generalization of the SFF where only a subregion of a quantum many-body system is being traced out.
More precisely, the \textit{partial spectral form factor} (PSFF) for subregion $A$ of a quantum many-body system with time evolution operator $U(t)$ is defined as~\cite{Gong_2020, Garratt2021prx, ZollerSFF2021}
\be
K_A(t) = \Nhil_{\overline{A}}^{-1} \Tr_{\overline{A}}[ \Tr_A[U(t)] \Tr_A[U^\dagger(t)] ]
\ee
where $\Tr_A$ is the partial trace over region $A$, $\overline{A}$ is the complement of $A$, and $\Nhil_{\overline{A}}$ is the Hilbert space dimension of $\overline{A}$. We define $K_A \equiv K$ if $A$ is the whole system.  PSFF has recently been analytically studied in quantum many-body systems without symmetries~\cite{Garratt2021prx, ZollerSFF2021, yoshimura2023operator} and with the dual unitarity condition~\cite{fritzsch2024eigenstatecorrelationsdualunitaryquantum}. 
Importantly, in addition to eigenvalue correlations, PSFF contains information on the eigenstate correlations, which is one of the robust diagnostic of quantum chaos, as exemplified by the Eigenstate Thermalization Hypothesis~\cite{deutsch1991quantum, Srednicki, Rigol2008}. This connection can be seen by taking sufficiently large time $t$, when PSFF $K^A(t)$ tends to the averaged purity of the reduced energy eigenstate~\cite{ZollerSFF2021}.

We expect {\gqmbcs} to display RMT behaviour in late time in PSFF as in SFF and 2PAF. In particular, for Floquet TRS {\gqmbcs}, we expect the late-time behaviour to be of the Gaussian or Circular Orthogonal Ensemble.
The PSFF of these ensembles exhibit a shift-ramp-plateau behaviour, which can be obtained using the fact that the distributions of eigenvalues and eigenvectors of random matrices decouple~\cite{ZollerSFF2021}
\be \label{eq:psff_rmt}
\overline{K^{\mathrm{RMT}}_{A, \text{f-TRS}}(t)} =  c^{(1)}_A + c^{(2)}_A \overline{K^{\mathrm{RMT}}_{\text{f-TRS}}(t)} \, ,
\ee
where $c_A^{(1)}= \Nhil_A^2(\Nhil_{\Abar}^2 + \Nhil_{\Abar} -2)/((\Nhil - 1 ) (\Nhil +2 ))$ and $c_A^{(2)}= (\Nhil + \Nhil_{\Abar} +1) (\Nhil_{A} + 1)/ ({\Nhil}_{\Abar} (\Nhil - 1 ) (\Nhil +2 ))$ with Hilbert space dimension $\Nhil = \Nhil_A \Nhil_{\Abar}$. 
For $\Nhil_A, \Nhil_{\Abar} \gg 1$, we have $\overline{ K^{\mathrm{RMT}}_{A, \text{f-TRS}}(t)} \approx 1 +  \overline{ K^{\mathrm{RMT}}_{ \text{f-TRS}}(t)} / \Nhil_{\Abar}^{2}$. In other words, the RMT PSFF behaviour coincides with the rescaled RMT SFF behaviour after a shift. 

In observables such as the PSFF and the 2PAF (see below), the quantum many-body system is partitioned into subregion $A$ and its complement $\overline{A}$. We will derive quantum-classical mappings between the ensemble averaged of such an observable and the partition functions of a classical statistical mechanical system
\begin{IEEEeqnarray}{l}  
Z_A := \sum_{\{\sigma_i\}} \exp\left[-\beta \left(H_A + H_{\overline{A}} + H_{A\overline{A}} \right)\right]  \,,  \label{eq:Z_A}
\\
= \sum_{\{ \sigma \}, \{\bar{\sigma} \} }
\prod_{\langle i,j\rangle \in A} \mathcal{B}^{A }_{\sigma _i \sigma _j}
\prod_{\langle k,l\rangle \in \overline{A}} \mathcal{B}^{\overline{A} }_{\bar{\sigma }_{k} \bar{\sigma }_{l}}
 \prod_{\langle m,n\rangle \in \partial A} \mathcal{B}^{A \overline{A}}_{\sigma _m \bar{\sigma }_n}  \nonumber
\end{IEEEeqnarray}
where the bulks of subregion $A$ and its complement $\overline{A}$ are governed by the Hamiltonian $H_A$ and $H_{\overline{A}}$ respectively, and the boundaries between the subregions $\partial A$ are governed by the Hamiltonian $H_{A\overline{A}}$.
It is also convenient to represent the Boltzmann factor using $ 
\exp(-\beta H_A)= 
\prod_{\langle i,j\rangle \in A}  \mathcal{B}^{A }_{\sigma _i \sigma _j}$, where the product is over the bonds of coupled sites, and similarly for the bulk $\overline{A}$ and the boundary $\partial A$.
We will see that for 2PAF and PSFF, $H_{\overline{A}}$ describes a trivial 1-level system in large $q$, consistent with the fact that in $\overline{A}$, many pairs of unitary gates and their Hermitian conjugate annihilate each other due to unitarity.
Then, the many-body dynamical behaviour of such observables can be interpreted in terms of the statistical mechanical properties of the emergent system in subregion $A$ and the interactions on the boundaries between $A$ and $\overline{A}$.

\begin{figure}[t]
    \centering
    \includegraphics[width=0.4 \textwidth]{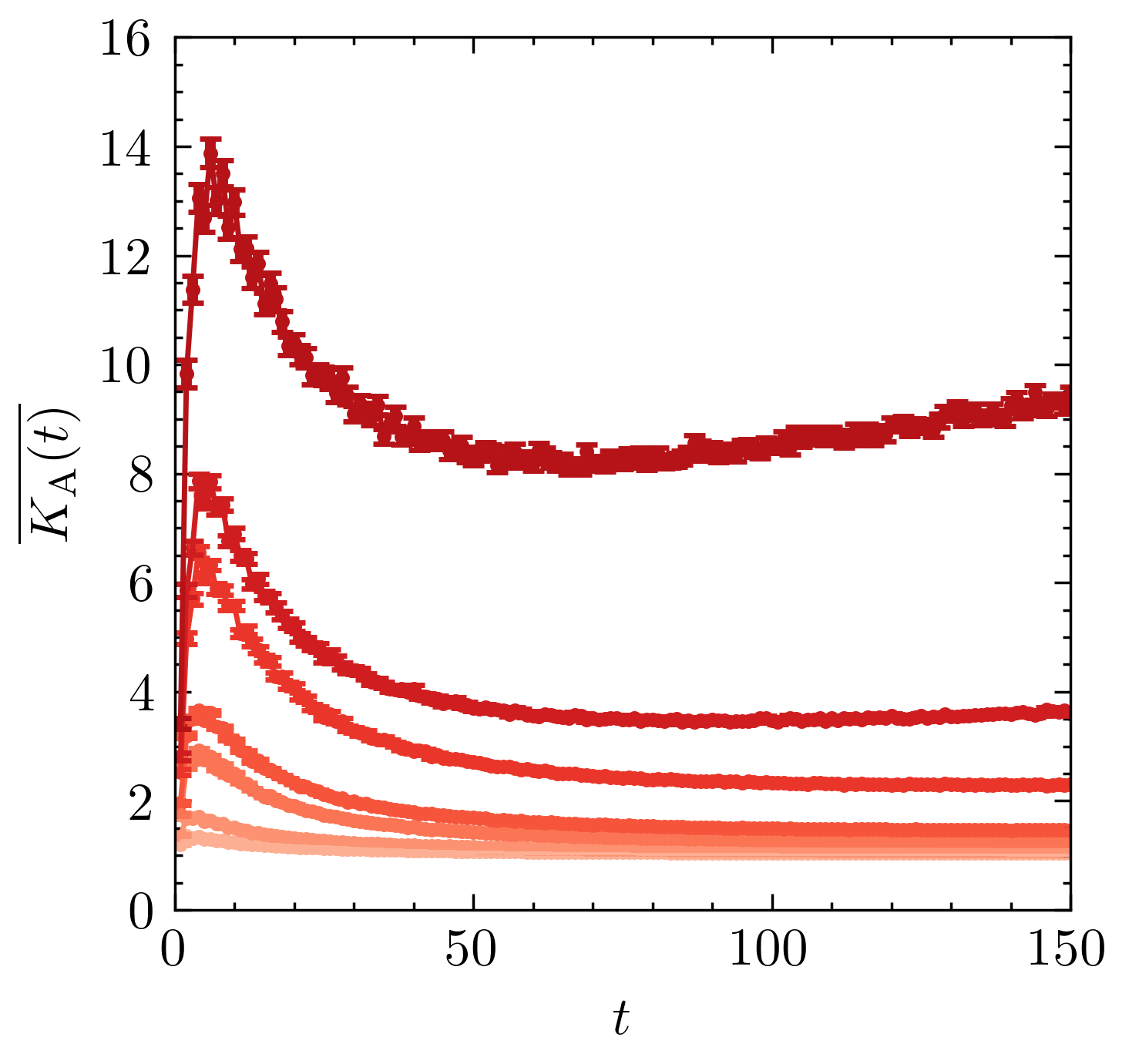}
    \caption{\textbf{PSFF against $t$.} Numerical simulations of PSFF $\overline{K_A(t)}$ against $t$ for the 1D Floquet TRS HRM for $q=2$ and $L=10$ with $L_A=1,2,\dots, 7$ in increasingly dark shades. The HRM is chosen to have obc with the region $A$ connected to the boundary. 
    }
    \label{fig:psff_num}
\end{figure}

The evaluation of PSFF in general geometries with arbitrary dimensions follows from the calculation of SFF and leads to the mappings 
\be\label{eq:psff_overview}
\ba
\lim_{q\to\infty} \overline{K^{\text{RPM}}_A} \propto 
\begin{cases}
     Z_A^{t\text{-Potts}}   \qquad & \text{Floquet,}
    \\
      1 \qquad & \text{Local TRS,}
     \\
     Z_A^{\text{Ising}} \qquad & \text{Global TRS,}
    \\
     Z_A^{2t\text{-Potts'}} \qquad & \text{Floquet TRS,}
\end{cases} 
\ea
\ee
using definitions in \eqref{eq:ising_model}, \eqref{eq:floq_trs_mapping}, and \eqref{eq:Z_A}. We will make these statements precise below in one-dimensional settings. For both local TRS RPM and HRM, in the limit of large local Hilbert space dimension $q$,  we have 
\be\label{eq:trs_psff_ltrs}
\ba
\overline{K_{A, \text{l-TRS}}(t)}= 
1  \, ,
\ea
\ee
since there are only a single leading local diagram on each site: Fig.~\ref{fig:2pcf_diag} (b) for sites in $\overline{A}$, and Fig.~\ref{fig:sff_diag} (e) for sites in $A$ according to Lemma \ref{lemma:leading_sff_diag}. For global TRS RPM, there is still a single leading diagram Fig.~\ref{fig:2pcf_diag} (b)  at sites in $\overline{A}$, but two leading diagrams Fig.~\ref{fig:sff_diag} (b) and (e) at sites in $A$. 
For one-dimensional global TRS RPM with and TRS breaking mechanism, this gives
\begin{equation}\label{eq:trs_psff_gtrs_full}
\ba
&\lim_{q\to \infty} \overline{K^{\text{RPM}}_{A,\, \,  \text{g-TRS}}(t)}=
\\
&\quad \prod^n_{i=1} \left( \frac{ \zeta -y+1
}{2 \zeta }\lambda _1^{L_{A_i}+1}
+
\frac{
 \zeta +y-1
}{2 \zeta }\lambda _2^{L_{A_i}+1} \right)
\ea
\end{equation}
where $i$ labels the $i$-th connected subregion of $A$, and 
the variables $\lambda_1$, $\lambda_2$ and $\zeta$ are given in Eq.~\eqref{eq:eval_ising_mag}.
In the simplest case where the TRS mechanism (therefore, the emergent magnetic field) is turned off with $b=0$, and where region $A$ is connected with $n=1$, we have
\be\label{eq:trs_psff_gtrs}
\ba
\lim_{q\to \infty}
\overline{K^{\text{RPM}}_{A, \text{g-TRS}}(t,L)}
=& 
\frac{1}{2} \left[\left(1-e^{-\epsilon \alpha(t)}\right)^{L_{A}+1} \right.
\\
& \quad \quad
\left. +\left(1+e^{-\epsilon \alpha(t)}\right)^{L_{A}+1}\right] \, .
\ea
\ee
\eqref{eq:trs_psff_gtrs} tends to one for time after the Thouless time  $t_{\mathrm{Th}, A}= \ln L_A / \epsilon$, which also defines a corresponding Thouless subsystem size $L_{\mathrm{Th}, A}(t) = e^{\epsilon  t}$. To seek universal behaviour, as in the SFF case, we take the Thouless scaling limit, except that $L_A$ in PSFF now plays the role of $L$ in SFF. More precisely, we send $t$ and $L_A$ (and therefore $L$) to infinity, while keeping $x= L_A/ L_{\mathrm{Th}, A}$ fixed. 
This results in the scaling form 
\be
\kappa^{\text{RPM}}_{A, \text{g-TRS}}(x)\equiv  \lim_{\substack{L_A,t \to \infty\\ x = L_A/L_{\mathrm{Th}, A} }}   \overline{K^{\text{RPM}}_{A, \text{g-TRS}}}=
 \cosh x \, ,  
\ee
which coincide with Eq.~\eqref{eq:ising_scaling} aside from a factor of $1/2$, which arises due to the fact that there is a single leading diagram $\sigma^{(1)}_{\mathrm{ladder}}$ in $\overline{A}$.

 For one-dimensional fTRS RPM, in the limit of large local Hilbert space dimension $q$, we have for odd time, 
\be\label{eq:trs_psff_ftrs_oddt}
\ba
&\lim_{q\to \infty} \overline{K^{\text{RPM}}_{A, \text{f-TRS}}(t \in 2 \mathbb{Z}+1 )}= 
\\
& \qquad \prod_{i}^n \frac{1}{2t} \Bigl\{   
(2t-2)[1-e^{- \epsilon t}]^{L_{A_i}+1}
\\
&\qquad +
[1 -t e^{-\epsilon (t-1)} + (t-1)e^{-\epsilon t }]^{L_{A_i}+1}
\\
&\qquad 
 +
[1 +t e^{-\epsilon (t-1)} + (t-1)e^{-\epsilon t }]^{L_{A_i}+1}
\Bigr\}
 \, ,
 \ea
\ee
and for even time, 
\be\label{eq:trs_psff_ftrs_event}
\ba
& \lim_{q\to \infty} \overline{K^{\text{RPM}}_{A, \text{f-TRS}}(t\in 2 \mathbb{Z})}= 
\\
& \quad
\prod_{i=1}^n
\frac{1}{2t} 
\Biggl\{   
(2t-4)(1-e^{- \epsilon t})^{L_{A_i}+1}  
\\
 & \quad +2 \left[1-\frac{t}{2}  e^{-\epsilon (t-2)}+\left(\frac{t}{2}-1\right) e^{-\epsilon t}
  \right]^{L_{A_i}+1}
  \\
  & \quad   +
  \left[1+\frac{t}{2}  e^{-\epsilon (t-2)}
  -\left(\frac{t}{2}+1\right) e^{-\epsilon t}\right]^{L_{A_i}+1} 
  \\
&  \quad   +\left[1+ \frac{t}{2}  e^{-\epsilon (t-2)}+\left(\frac{3 t}{2}-1\right) e^{-\epsilon t}\right]^{L_{A_i}+1}
\Biggr\}
 \, .
 \ea
\ee
where again $i$ labels the $i$-th connected subregion of $A$. The Thouless subsystem length is given by  $L_{A, \mathrm{Th}}(t) = e^{\epsilon t}/t$. The Thouless scaling limit for PSFF can be taken analogous to the SFF case, except that the role of $L$ is taken by $L_A$, i.e. we take $t$ and $L_A$ (and therefore $L$) to  infinity, while keeping $x= L_A/ L_{A, \mathrm{Th}}(t)$  fixed. In this limit, we obtain the scaling form for Floquet TRS RPM as
\be \label{eq:psff_ftrs_pbc_scaling}
\ba
\kappa^{\text{RPM}}_{A, \text{f-TRS}}(x) & \,\equiv \,  \lim_{\substack{L,t \to \infty\\ x = L/\Lth}} 2t \, \overline{K^{\text{RPM}}_{\text{f-TRS}}} -2t
\\
& \, \qquad   = f_\epsilon (e^x) - 2x
\,, 
\ea
\ee
where $f_\epsilon (y)= y^{1+e^{\epsilon}} + y^{1-e^{\epsilon}} -2 $ for odd $t$, and 
$f_\epsilon(y)= 2 y^{\frac{1}{2}(1-e^{2\epsilon})} + 
y^{\frac{1}{2}(-1+e^{2\epsilon})}
y^{\frac{1}{2}(3+e^{2\epsilon})} -4$ for even $t$. This coincides with the scaling form for SFF with pbc given in Eq.~\eqref{eq:ftrs_pbc_scaling}, except for an extra factor of $2t$.

\begin{figure}[t]
    \centering
    \includegraphics[width=0.4 \textwidth]{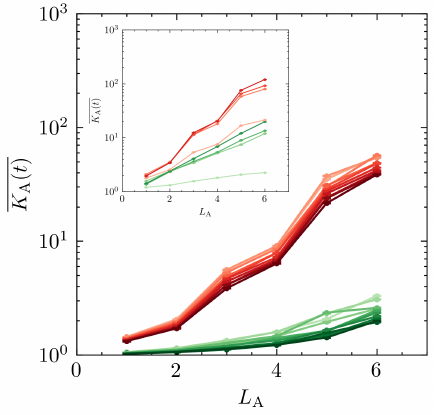}
    \caption{\textbf{PSFF against $L_A$.}   Numerical simulations of PSFF $\overline{K_A(t)}$ against $L_A$ for the 1D Floquet HRM with TRS (green) and without TRS (red) for $q=2$ and $L=10$ with $t\in [15,24]$ (main panel), and $t\in [1,4]$ (inset) in increasingly dark shades. The HRM is chosen to have obc with the region $A$ connected to the boundary. 
    }
    \label{fig:psff_num_against_La}
\end{figure}

The one dimensional Floquet RPM without TRS or other symmetries~\cite{chan2018spectral} can be evaluated as the partition function of $t$-state Potts model given by~\cite{yoshimura2023operator} 
\be
\ba
\lim_{q\to \infty} \overline{K^{\text{RPM}}_{A, \text{no-sym}}(t)}=
\prod_{i=1}^n
\frac{1}{t}\left[(1+(t-1)e^{-\epsilon t})^{L_{A_i}+1}  \right.
\\
\left.  \quad + (t-1) (1 - e^{-\epsilon t})^{L_{A_i} +1} \right] \,.
\ea
\ee
For simplicity, take the region $A$ to be a single connected region. Then in the Thouless scaling limit, we send $t$ and $L_A$ (and therefore $L$) to infinity, while keeping $x= L_A/ L_{\mathrm{Th}, A} = L_A t e^{-\epsilon t}$ fixed. This gives the scaling form
\be
\kappa^{\text{RPM}}_{A, \text{no-sym}}(x)\equiv   \lim_{\substack{L_A,t \to \infty\\ x = L_A/L_{\mathrm{Th}, A} }}
t  \overline{K^{\text{RPM}}_{A, \text{no-sym}}}- t=
 e^x -x-1 \, ,
\ee
coinciding with the SFF scaling form for RPM without symmetries and with pbc, except for a factor of $t$, and except that the role of Thouless length is taken by the Thouless subsystem size $L_{\mathrm{Th}, A}$. 
We numerically simulate the PSFF for Floquet TRS HRM and Floquet TRS 3PM with system size $L=10$. 
We consider obc and choose the region $A$ to be connected to the boundary. 
In Fig.~\ref{fig:psff_num}, we plot PSFF $K_A(t)$ against $t$ for the 1D Floquet TRS HRM for $q=2$ and $L=10$ with various $L_A$. 
We observe that PSFF exhibits a characteristic bump  which is qualitatively described in the large-$q$ solution of the Floquet TRS RPM in Eqs.~\eqref{eq:trs_psff_ftrs_oddt} and \eqref{eq:trs_psff_ftrs_event}. 
In Fig.~\ref{fig:psff_num_against_La}, we plot PSFF against $L_A$ for the same model and find that PSFF grows exponentially with $L_A$, with a higher growth  rate in the presence of  TRS compared to the case without TRS. 
In Appendix~\ref{app:numerics}, we provide numerical data of PSFF for Floquet 3PM which show qualitatively similar behavior to PSFF of Floquet HRM.

\section{Correlation functions}\label{sec:corr}

\begin{figure*}[ht]
    \centering
    \includegraphics[width=1 \textwidth]{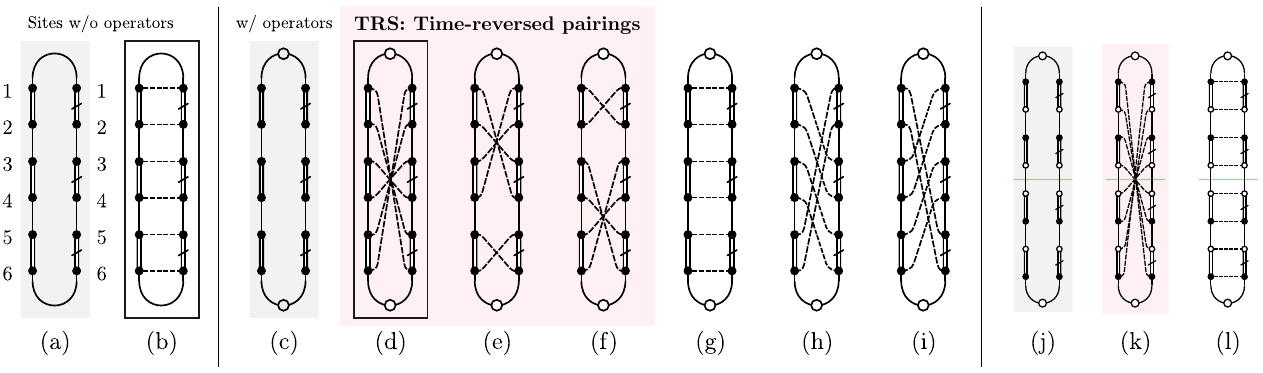}
    \caption{\textbf{2PAF diagrams.} (a,b) Diagrammatical representation and the leading local diagram on sites without operator support respectively for all TRS dynamics (for models with global TRS without local TRS, the black dots have to be partitioned into black and white dots, but the diagram is the same). (c) Diagrammatical representation on sites with operator support with (d) a single leading 2PAF diagrams with time-reversed pairing, (e-f) subleading diagrams with time-reversed pairing, and (h-i) subleading  2PAF diagrams with time-parallel pairings for \textbf{Floquet TRS} chaotic systems. (g) A 2PAF diagram with time-parallel pairing that vanishes due to tracelessness condition of  operators. 
    For \textbf{global TRS w/ local TRS} chaotic systems, the same diagrams appear except that (e,f,h,i) are absent. 
    (j-l) are the analogues of (c), (d) and (g) respectively for \textbf{global TRS w/o local TRS} chaotic systems.
    For \textbf{local TRS} systems, only (a), (b), (c) and (g) remain. 
(l) vanishes due to tracelessness condition of operators.
    If the operators are identity and the 2PAF is locally normalized with $q^{-1}$ as in Eq.~\eqref{eq:2paf_def}, then the diagrams are of order $O(1)$ for (b,g,l); $O(q^{-1})$ for (d) and (k); and $O(q^{-2})$ for (e,f,h,i), i.e. the twisted and ladder diagrams are not on equal footing. The boxed diagrams, (b) and (d), are the leading local diagrams and they form the emergent effective Ising degrees of freedom for global TRS many-body chaotic systems.}
    \label{fig:2pcf_diag}
\end{figure*}

In this section, we derive exact results for the two-point autocorrelation function (2PAF) and its fluctuation in {\gqmbcs} with TRS in the large-$q$ limit. 
We demonstrate that, due to TRS and the many-body interactions between time-reversed pairings of Feynman paths [Fig.~\ref{fig:contraction}], both the 2PAF and its fluctuations exhibit distinct signatures of many-body quantum chaos, in contrast to the case without TRS.
Specifically, we show that the 2PAF are  dominated by time-reversed pairings of Feynman paths on physical sites with local operator support, resulting in characteristic negative and suppressed values of 2PAF of antisymmetric operators, as compared to the case with symmetric operators. By summing 2PAF over a complete operator basis (using Eq.~\eqref{eq:sff_2paf_connection} below), we rederive the emergent Ising scaling behaviour of the SFF in global TRS  {\gqmbcs}. Furthermore, we find that the fluctuations of the 2PAF are governed by an emergent three-state Potts model, leading to exponential growth of the operator support size in the presence of TRS.
Consider the Hilbert space $\mathbb{C}^{q^L}$ of a quantum many-body system with $L$ sites, where each local Hilbert space is $\mathbb{C}^q$. At each physical site, we use the generalized Gell-Mann matrices to form an operator basis.
 To this end, let $E_{j,k}$ denote a matrix with the $(j,k)$-th matrix element given by 1, and 0 otherwise. The Gell-Mann matrices are defined by
\begin{IEEEeqnarray}{ll}   
\quad o_{\mu } \equiv o_{j,k } =
\label{eq:gellmann}
\\ \nonumber
\begin{cases}
\sqrt{\frac{2}{j (j+1)}} \left( \, 
\sum_{\ell=1}^j E_{\ell,\ell}  - j E_{j+1, j+1} 
\right), \,  &  j=k \in [1, q-1], 
\\ 
 E_{j,k} + E_{j,k} \, ,  &  j>k ,
\\
 -i (E_{j,k} - E_{j,k}) \, , &  j<k, 
\end{cases}
\end{IEEEeqnarray}
such that there are $q+q(q-1)/2+q(q-1)/2=q^2-1$ basis states, and we have used $\mu$ to label the 2-tuple $(j,k)$. 
For convenience, we  further define $o_{q,q}\equiv \mathbb{1}$.
This operator basis is traceless for $(j,k) \neq (q,q)$, and is orthonormal, i.e. $q^{-1} \Tr[o_{\mu} o_{\nu}] = \delta_{\mu \nu}$.  For $q=2$, this basis consists of Pauli matrices. 
 Note that the set of 2-tuples $(j,j)$ for $j \in [1, q-1]$ labels all diagonal operator basis states other than the identity, and the set of 2-tuples $(j,k)$ with $j>k$ ($j<k$) labels all off-diagonal and (anti-)symmetric operator basis states.
Importantly, we will see that the behaviour of 2PAF and its fluctuation of a given operator depend on whether the local operator support at time zero is diagonal and symmetric.

\begin{figure*}[ht]
    \centering
    \includegraphics[width=1 \textwidth]{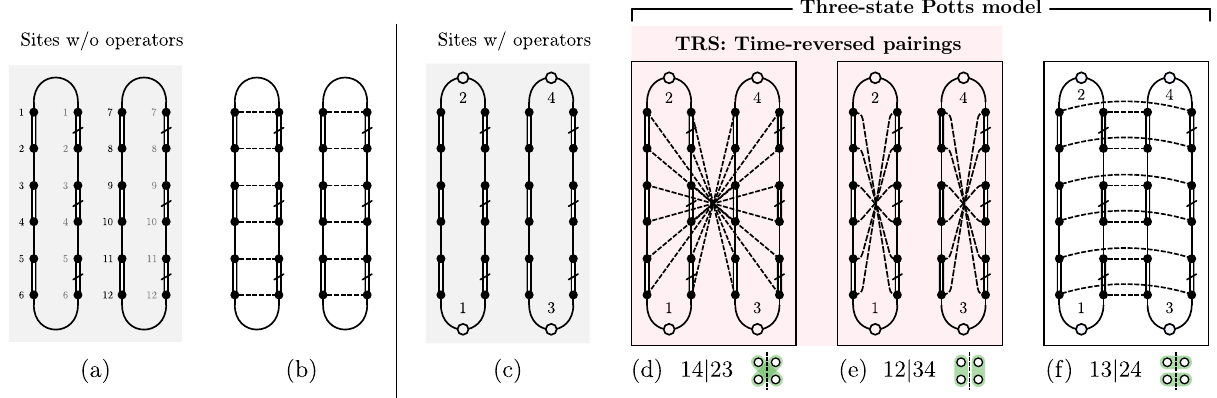}
    \caption{\textbf{2PAF fluctuations diagrams.} (a,b) Diagrammatical representation and the leading local diagram on sites without operator support respectively for all TRS dynamics. (c) Diagrammatical representation on sites with operator support, with leading diagrams given by (d,e,f) for \textbf{Floquet TRS} and \textbf{global TRS} {\gqmbcs}. (d,e,f) corresponds to the three ways of pairing up four operators, forming the effective degrees of freedom for the three-state Potts model, see discussion in the main text and in App.~\ref{app:2paf_proof}. 
    For {\gqmbcs}   with \textbf{local TRS} or \textbf{without TRS}, (f) is the only leading diagrams.
    If the operators are identity and the 2PAF is locally normalized with $q^{-1}$ as in Eq.~\eqref{eq:2paf_def}, then the diagrams are of order $O(1)$ for (b); and $O(q^{-2})$ for (d,e,f).
    The boxed diagrams, (d,e,f), form the emergent effective  degrees of freedom  of three-state Potts model for global TRS chaotic systems.
    }
    \label{fig:2pcf_fluc_diag}
\end{figure*}

For the many-body Hilbert space, we take operator strings $O_{{\mu}} =\bigotimes_{i=1}^L o_{\mu_i}$ to form the basis states, where we have abused the notations and denoted  $\mu=(\mu_1,\mu_2, \dots,\mu_L)$.
The \textit{two-point correlation function} at time $t$ at infinite temperature is defined by
\be\label{eq:2paf_def}
C_{\mu \nu }(t)  = \frac{1}{\mathcal{N}} \Tr[ O_{\mu}(t) O_\nu(0)] \, ,
\ee
where $O_\mu(t) = U(t) O_\mu U^\dagger (t)$ is the time-evolved operator under the Heisenberg picture, and $\Nhil$ is the Hilbert space dimension. The \textit{two-point autocorrelation function} (2PAF) of operator $O_\mu$ is defined as $C_{\mu \mu }(t)$. 
Connections between  SFF and correlation functions have been pointed out in \cite{Cotler2017complexity, Gharibyan_2018}. Specifically, it has been shown that SFF or partial SFF is equal to the sum of 2PAF of operators with all possible operator supports in the system or in subregion $A$, i.e. 
\be \label{eq:sff_2paf_connection}
\ba
K(t) = & \,\sum_{\mu \in \mathcal{P}} C_{\mu \mu} (t)  \, , 
\ea
\ee
where $\mathcal{P}$ is a complete set of basis operators acting on the system including the identity operator. For PSFF, we have similarly $K_A(t) = \, \sum_{\mu \in \mathcal{P}_A} C_{\mu \mu} (t)$ with $\mathcal{P}_A$ the complete set of basis operators acting on $A$. 
Note that these identities hold generally, independent of the properties of the quantum many-body systems $U(t)$, and without the ensemble average.

Two-point correlation functions play a crucial role in  characterizing quantum fluctuations and understanding non-equilibrium dynamics in quantum many-body systems~\cite{Lieb, rickayzen1980greens, Rigol2008, Rigol, Luitz_2016, Cheneau_2012}. 
%
%
%
The behaviour of correlation functions in {\gqmbcs} as modeled by random quantum circuits has recently been studied in temporal-random HRM~\cite{Nahum2022correlator}, Floquet RPM~\cite{yoshimura2023operator}, and Floquet HRM~\cite{chan2018solution}. 
For temporal random {\gqmbcs}, the 2PAF decays rapidly and the SFF approaches unity on a timescale of order one. 
For {\gqmbcs} with time translational symmetry~\cite{yoshimura2023operator}, the system tends to exhibit full random matrix behaviour at late times, displaying the ramp-plateau structure as captured by \( \overline{C_{\mu\mu}} \approx [\overline{K_{\mathrm{RMT}}(t)} - 1]/(q^{2L} - 1) \).
At sufficiently early times, the 2PAF displays a bump preceding the onset of the ramp, arising from a mechanism analogous to that in the SFF, as {\gqmbcs} resemble patches of random matrices.
Here, we demonstrate that the 2PAF and its fluctuations in {\gqmbcs} exhibit several generic behaviours that are uniquely tied to the presence of time-reversal symmetry (TRS).
We begin by stating a lemma identifying the leading diagrams contributing to the 2PAF and its fluctuations in Sec.~\ref{sec:2paf_diag}, followed by the computation of the 2PAF in Sec.~\ref{sec:2paf_calc} and its fluctuations in Sec.~\ref{sec:2paf_fluc}.

We note that in the large-\( q \) limit, fluctuations of the 2PAF in our random quantum circuit models can be computed analytically from the first two moments. Throughout this work, we will use the term `fluctuations' to refer interchangeably to both the variance and the second moment of the 2PAF, with the latter offering a simpler analytical expression.

\subsection{2PAF diagrams}\label{sec:2paf_diag}
Here we give a lemma identifying the leading local 2PAF diagrams and 2PAF fluctuation diagrams in {\gqmbcs} with Floquet or global TRS, and also in {\gqmbcs} without symmetries. 

\begin{lemmaalt}
\label{lemma:leading2pcf}
\textbf{Leading local two-point autocorrelation function (2PAF) for generic quantum many-body quantum chaotic systems with or without TRS.} Consider the 2PAF \eqref{eq:2paf_def} with local diagram  labeled as in Fig.~\ref{fig:2pcf_diag} (a). 
For  (i) random phase model (RPM), (ii) Haar-random model (HRM), and (iii) the random matrix model (RMT) with Floquet TRS / with global TRS / without symmetries.
For (i) and (ii) at each physical site in the order of the local Hilbert space dimension $q$, and for (iii)  in the order of matrix dimension $N$, 
 the leading local 2PAF diagram is of order $O(1)$ given by the ladder contraction in $S_{2t}$   [Fig.~\ref{fig:2pcf_diag} (b)] 
\be
\sigma_{\mathrm{ladder}}^{(m=1)}(i) =  i \, ,   \quad \text{TRS / No symmetries, }\ee
for sites without local operator support with or without TRS. 
For sites with non-trivial operator support, the leading local 2PAF diagram  in the presence of TRS is of order $O(q^{-1})$ or $O(N^{-1})$, and is given by the twisted contractions  in $S_{2t}$   [Fig.~\ref{fig:2pcf_diag} (d)] 
\be \label{eq:2paf_diag_v1}
\sigma_{\mathrm{twisted}}^{(m=2t)}(i) = 2t - i +1 \, , \quad \text{TRS.}
\ee
For sites with non-trivial operator support, for RPM, HRM, and RMT in the absence of TRS, the leading local 2PAF diagram are of order $O(q^{-2})$ or $O(N^{-2})$, and are given by ladder contractions in $S_{2t}$ [Fig.~\ref{fig:2pcf_diag} (h,i)] 
\be \label{eq:2paf_diag_notrs}
\sigma_{\mathrm{ladder}}^{(m)}(i) = (m + i -1) \Mod{2t}\, ,  \quad \text{No sym.,} 
\ee
for $m=3,5,\dots, 2t-1$.
\end{lemmaalt}

\noindent\textit{Proof.} The proof for \ref{lemma:leading2pcf} is provided in the Appendix~\ref{app:2paf_proof}.

 In regions of operator support, there is a single leading local 2PAF diagram in the presence of TRS, Eq.~\eqref{eq:2paf_diag_v1}, of order $O(q^{-1})$. 
In contrast, there are $(t-1)$ leading local 2PAF diagrams \cite{yoshimura2023operator} of order $O(q^{-2})$ in Eq.~\eqref{eq:2paf_diag_notrs}, with $\sigma_{\mathrm{ladder}}^{(m=1)}$ vanishing due to the traceless condition of the operator by definition.

\begin{lemmaalt}\label{lemma:2paf_fluc_v2}
\textbf{Leading local two-point autocorrelation function (2PAF) fluctuations diagrams for generic quantum many-body quantum chaotic systems with or without TRS.} Consider the 2PAF fluctuations with local diagram labeled as in Fig.~\ref{fig:2pcf_fluc_diag} (a).
For (i) random phase model (RPM), (ii) Haar-random model (HRM), and (iii) the random matrix model (RMT) with Floquet TRS / with global TRS / without symmetries.
For (i) and (ii) at each physical site in the order of the local Hilbert space dimension $q$, and for (iii)  in the order of matrix dimension $N$,  the leading local 2PAF fluctuations diagram is of order $O(1)$ given by the ladder contraction in $S_{4t}$   [Fig.~\ref{fig:2pcf_fluc_diag} (b)] 
\be
\sigma_{\mathrm{ladder}}^{(m=1)}(i) =  i \, ,   \quad \text{TRS / No symmetries, }\ee
for sites without local operator support with or without TRS. 
 In the presence of Floquet TRS or global TRS, for sites with non-trivial operator support, the leading local 2PAF fluctuations diagrams are of order $O(q^{-2})$ or $O(N^{-2})$, and are given by contractions in $S_{4t}$ [Fig.~\ref{fig:2pcf_fluc_diag} (d,e,f) respectively] 
\begin{IEEEeqnarray}{rll} 
\sigma_{(14|23)}(i) = & \,(4t +1 - i) \Mod{4t}  \, , 
\quad &    \nonumber
\\
\sigma_{(12|34)}(i) = & \,  (2t +1 - i) \Mod{4t} \, ,  
\quad &  \text{TRS.}   \label{eq:2paf_fluc_diag_trs}
\\
\sigma_{(13|24)}(i) = &  \, (2t+ i) \Mod{4t}  \, , 
\quad &   \nonumber
 \end{IEEEeqnarray}
In the absence of TRS, for sites with non-trivial operator support, the leading local 2PAF fluctuations diagram is of order $O(q^{-2})$ or $O(N^{-2})$, and is given by contractions in $S_{4t}$ [Fig.~\ref{fig:2pcf_fluc_diag} (f)]
 \begin{IEEEeqnarray}{rll} 
\sigma_{(13|24)}(i) = &  \, (2t+ i) \Mod{4t}  \, 
,  \quad & \text{No sym.}  \label{eq:2paf_fluc_notrs_diag}
 \end{IEEEeqnarray}
\end{lemmaalt}
\noindent\textit{Proof.} The proof for Lemma \ref{lemma:2paf_fluc_v2} is provided in Appendix~\ref{app:2paf_proof}.

\vspace{0.3cm}
An intuitive understanding of why there are three leading diagrams Eq.~\eqref{eq:2paf_fluc_diag_trs} in the presence of TRS is described in the introduction -- time-reversed pairing of Feynman paths allows all three ways to group four objects into pairs. 
Note that one can derive Eq.~\eqref{eq:2paf_fluc_notrs_diag} for the case without TRS simply by observing that only $\sigma_{(13|24)}$ in Eq.~\eqref{eq:2paf_fluc_diag_trs} does not involve time-reversed pairings.

\begin{figure*}[t]
    \centering
    \includegraphics[width=0.85 \textwidth]{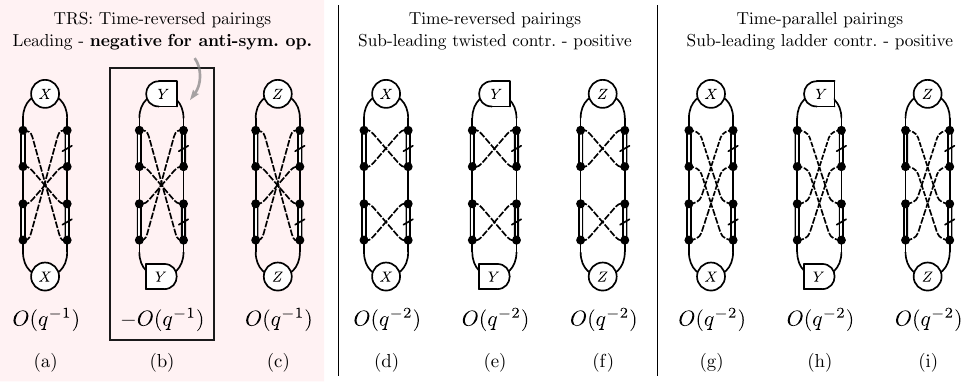}
    \caption{\textbf{Time-reversed pairings in 2PAF of antisymmetric operators give negative leading contributions. } Consider local 2PAF diagrams at time $t=2$ for sites where operators are supported.
     The 2PAF has leading contributions at order $O(q^{-1})$ from time-reversed pairings of Feynman paths, giving  (a, c) positive contributions for symmetric operators, e.g. Pauli-$X$ and -$Z$; and (b) negative contributions for antisymmetric operators, e.g. Pauli-$Y$. 
    A subset of subleading contributions of 2PAF of $O(q^{-2})$ with (d, e, f)  time-reversed pairings and (g,h,i) time-parallel pairings are positive, regardless of the symmetricity of the opoerator. The full set of subleading contributions are discussed in a upcoming work in~\cite{upcoming}.
In the absence of TRS, the leading contributions to the 2PAF appear at order $O(q^{-2})$, corresponding to (d, e, f). These contributions are positive and independent of whether the operators are symmetric or antisymmetric.
    }
    \label{fig:op_dependence}
\end{figure*}

\subsection{2PAF evaluations}\label{sec:2paf_calc}
\subsubsection{Time-reversed pairings}
In the case of SFF, the presence of TRS introduces both time-reversed pairings of Feynman paths (the twisted diagrams) and time-parallel pairings (the ladder diagrams). These two types of diagrams contribute to SFF on an equal footing. In contrast, in TRS {\gqmbcs}, time-reversed pairings of Feynman paths lead to contributions in 2PAF that are more \textit{dominant} -- by an order in $q$ at each site with non-trivial operator support -- than the leading contributions without TRS [Lemma~\ref{lemma:leading2pcf}]. This imbalance indicates that time-reversed and time-parallel pairings are \textit{not} on equal footing in 2PAF.

Here we evaluate the 2PAF of a given operator $O_\mu$ in the global TRS or Floquet TRS RPM. As before, we perform the ensemble averages over the one-site Haar-random gates and two-site coupling random phase gates in sequence. 
In regions with non-trivial operator support, Lemma~\ref{lemma:leading2pcf} states that the ensemble averages over local Haar-random gates lead to a single local leading twisted diagram $\sigma_{\mathrm{twisted}}^{(m=2t)}$ [Fig.~\ref{fig:2pcf_diag} (d)] of 2PAF of order $O(q^{-1})$. 
In contrast, in the case without TRS~\cite{yoshimura2023operator}, where the leading diagrams are the $t-1$  ladder diagrams [Fig.~\ref{fig:2pcf_diag} (h,i)] of order $O(q^{-2})$, where the missing ladder diagram [Fig.~\ref{fig:2pcf_diag} (g)] vanishes due to the traceless property of the operator.
In regions without operator support, it is straightforward to show that the leading local diagram is $\sigma_{\mathrm{ladder}}^{(m=1)}$ [Fig.~\ref{fig:2pcf_diag} (a)].

Importantly, due to the time-reversed pairing of Feynman paths, the dominant local 2PAF diagram [Eq.~\eqref{eq:2paf_diag_v1}] contributes a factor of \( \Tr[O_{\mu_i}^T O_{\mu_i}] = \pm q \) for each site with operator support. The sign of this contribution depends on whether the local operator \( O_{\mu_i} \) is symmetric (positive) or antisymmetric (negative).  
In fact, the negative values of the 2PAF for antisymmetric operators are essential for explaining several key phenomena. A number of remarks are in order. First, we emphasize that this sign structure arises from time-reversed pairings, which are inherently tied to the presence of TRS in {\gqmbcs}, and can be diagrammatically represented in Fig.~\ref{fig:op_dependence}. Second, the presence of this sign implies that the leading order 2PAF is negative for antisymmetric operators. This prediction is directly confirmed by numerical simulations of the 1D Floquet TRS HRM (not RPM), shown in Fig.~\ref{fig:2pcf_numerics_antisym} for two Floquet models and a Hamiltonian model. Third, this sign is crucial for the relation Eq.~\eqref{eq:sff_2paf_connection} to be satisfied. Suppose, for contradiction, that \( C_{\mu \mu}(t) = O(q^{-|O|}) \) remained strictly positive for all operators. Since there are \( O(q^{2|O|}) \) operators with support size \( |O| \), this would lead to a total contribution of \( O(q^{|O|}) \) to the spectral form factor (SFF). However, the SFF is known to scale as \( O(1) \) in \( q \), resulting in a contradiction. Therefore, the scaling behaviour of the 2PAF with respect to \( |O| \), as shown in Fig.~\ref{fig:2pcf_numerics}, provides a non-trivial consistency check of Eq.~\eqref{eq:trs_2pcf}. Fourth, we show below that the Ising scaling behaviour of the SFF in {\gqmbcs} with global TRS can be rederived by summing the 2PAF over a complete operator basis after accounting for the negative contributions from antisymmetric operators. Given universal nature of time-reversed pairings and the required consistency between the SFF and the 2PAF, we expect this operator-dependent structure of the 2PAF to be a generic dynamical feature of {\gqmbcs} in the presence of TRS.

The ensemble average of the two-site coupling gates closely follows to the average performed for SFF described above Eq.~\eqref{eq:ising_mapping}, except that the many-body interactions between local diagrams on site $a$ and $b$ depend on whether the local operators at $a$ and $b$ are off-diagonal. 
Explicitly, for Floquet TRS RPM and Global TRS RPM with general geometries, in the leading order in $q^{-1}$, the 2PAF of operator $O_{\mu}$ is given by 
\begin{IEEEeqnarray}{rl} \label{eq:2paf_full}
& \lim_{q\to \infty} \overline{C_{\mu \mu, \text{TRS}}^{\text{RPM}}(t)} 
= \bigg( \, 
\prod_{a } \mathcal{B}_{\mu_a \mu_a}^{\text{s}}
\bigg)
\bigg( \, 
\prod_{\langle a,b \rangle} \mathcal{B}_{\mu_a, \mu_b}^{\text{b}}
\bigg)
\\ \nonumber
&\mathcal{B}^{\text{s}}_{\mu_1, \mu_2 }=
\delta_{j_1 j_2} \delta_{k_1 k_2}
\big\{
\delta_{j_1 q } \delta_{k_1 q} 
+ 
(1- \delta_{j_1 q } \delta_{k_1 q} ) \delta_{j_1, k_1} q^{-1}
\\ \nonumber
& \quad  \quad 
(1- \delta_{j_1, k_1})
\left[ \Theta(j_1 - k_1)q^{-1} 
-\Theta(k_1 - j_1)q^{-1} 
\right]
\big\}
\\ \nonumber
&\mathcal{B}^{\text{b}}_{\mu_1, \mu_2 }=
\delta_{j_1 j_2} \delta_{k_1 k_2}
+
(1-\delta_{j_1 j_2 \delta_{k_1 k_2}})
\\ \nonumber
&
 \quad \quad  \big\{
 \delta_{j_1 q} \delta_{k_1 q}
 [\delta_{j_2 k_2} e^{-\epsilon \alpha(t) } + (1- \delta_{j_2 k_2}) e^{-\epsilon (\alpha(t) +1)}]
\\  \nonumber
 &
 \quad\quad 
 \delta_{j_2 q} \delta_{k_2 q}
 [\delta_{j_2 k_2} e^{-\epsilon \alpha(t) }  + (1- \delta_{j_2 k_2})
  e^{-\epsilon (\alpha(t) +1)} ]
 \\ \nonumber
 & 
\quad \quad +
 (1-\delta_{j_1 q} \delta_{k_1 q})
 (1-\delta_{j_2 q} \delta_{k_2 q})
\big\}
\end{IEEEeqnarray}
where $\mu = (\mu_1, \mu_2, \dots, \mu_L)$ with $\mu_i = (j_i, k_i)$, and $\alpha(t) =  t-1 - \delta_{0, t\Mod{2}}$ as before. The first product in \eqref{eq:2paf_full} is over all sites, and the second is over all bonds. 
Each site is associated with a Boltzmann weight $\mathcal{B}^{\text{s}}$ which is diagonal in its indices $\mu_1$ and $\mu_2$. At each site, $\mathcal{B}^{\text{s}}$ assigns a value of 1 to identity operators, $q^{-1}$ to symmetric operators, and $-q^{-1}$ to antisymmetric operators.
Each bond is associated with the Boltzmann weight $\mathcal{B}^{\text{b}}$, which describes the many-body interactions between time-reversed and time-parallel pairings of Feynman paths at different sites. 
Specifically, $\mathcal{B}^{\text{b}}=1$ between two sites with time-reversed pairings or two sites with time-parallel pairings. $\mathcal{B}^{\text{b}}$ is exponentially suppressed in $t$ for bonds connecting time-reversed  and time-parallel pairings.

\begin{figure*}[ht]
    \centering
    \includegraphics[width=1 \textwidth]{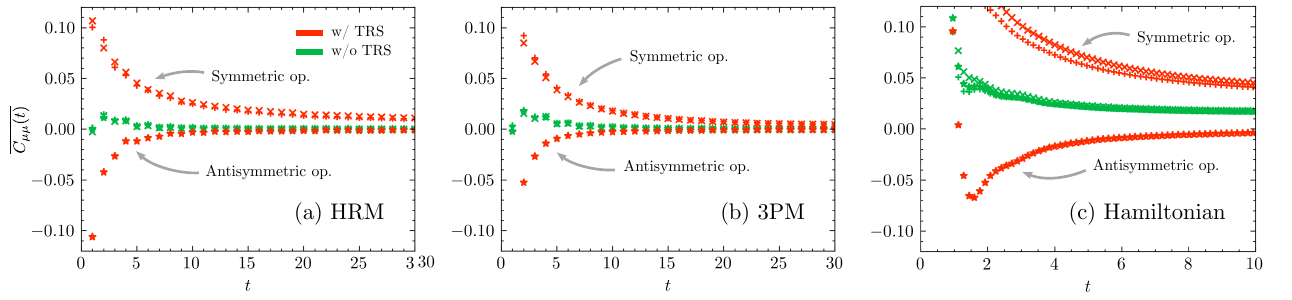}
    \caption{\textbf{2PAF of symmetric and antisymmetric operators} for one-dimensional (a) Floquet HRM, (b) Floquet 3PM, and (c) time-independent chaotic Hamiltonian model, in the presence (red) and absence (green) of TRS with $q=2$ and $L=10$.
 %
    %
The operators are single-site Pauli-$Z$ (cross symbols), Pauli-$X$ (x symbols), and  Pauli-$Y$  (star symbols) operators placed in the middle of the system.
    The 2PAF of antisymmetric (symmetric) operators is negative (positive) in the presence of TRS, with a smaller absolute value compared to the symmetric case, consistent with the result in Eq.~\eqref{eq:trs_2pcf_onekindop} and Fig.~\ref{fig:op_dependence}.
}   \label{fig:2pcf_numerics_antisym}
\end{figure*}

%
For $t \gg 1$, the exponential factors in Eq.~\eqref{eq:2paf_full}  can be approximated with $e^{-\epsilon t}$, and $\mathcal{B}^{\text{b}}_{\mu_1, \mu_2}$ depends only on whether the operators $\mu_{i}$ are identity or not. The 2PAF of an operator $O_\mu$ in the Floquet TRS RPM or Global TRS RPM in general dimensions is given by
\begin{IEEEeqnarray}{rl}\label{eq:trs_2pcf_onekindop}
    & \quad \lim_{q\to \infty}  \overline{C^{\text{RPM}}_{\mu\mu, \text{TRS}} (t)} = \Omega^{|O|} \Gamma^{|\partial O|} f(O_\mu) \,,
\\
\nonumber
&   \Omega = \, \,  \left(q^{2}-1 \right)^{-1} , \quad 
\Gamma =  \, \,  e^{-\epsilon  t } , \quad f(O_\mu) = \frac{1}{\mathcal{N}} \Tr[O_{\mu}^T O_{\mu}] ,
\end{IEEEeqnarray}
where  $|O|$ and $|\partial O|$ respectively denote the total size and boundaries of the operator support  (of possibly disconnected regions). 
$f(O_\mu)$ gives 1 if $O_\mu$ is symmetric, and $-1$ if  $O_\mu$ is antisymmetric.  
Further, we can average 2PAF over the set $\MM_A$ of all operators with non-trivial support at all sites in the region $A$ (in addition to the ensemble average over the quantum many-body systems). 
For sufficiently large $t$, for Floquet TRS and Global TRS RPM in general dimensions, we obtain 
\be\label{eq:trs_2pcf}
\ba
& \lim_{q\to \infty} \overline{C^{\text{RPM}}_{A, \text{TRS}} (t)} \equiv
\lim_{q\to \infty} \mathbb{E}_{\mu \in \MM_A} \left[ \overline{C^{\text{RPM}}_{\mu\mu, \text{TRS}} (t)}\right] \, ,
\\
& \, \,  \, \, = \Omega^{|O|} \Gamma^{|\partial O|} ,
 \quad \Omega = \, \,  \left(q^{2}-1 \right)^{-1} , \quad 
\Gamma =  \, \,  e^{-\epsilon  t } , 
\ea
\ee
where  the expectation value $\mathbb{E}$ denotes the average over $\MM_A$, the set of operators with non-trivial operator support at each site in $A$.

In the absence of TRS (and other symmetries),  2PAF for 1D Floquet RPM can be evaluated using the transfer matrix of SFF as~\cite{yoshimura2023operator}
\begin{equation}\label{eq:2pcf_nosym_1d}
\ba
& \lim_{q\to \infty} \overline{C^{\text{RPM}}_{\mu \mu, \,  \text{no-sym}} (t)} 
\\ & \quad = \, \,  q^{-2|O|} 
e^{-2 n \epsilon t} (t-1)^{n} [ 1 + (t-2) e^{-\epsilon t}]^{|O| - n} \, ,
\ea
\end{equation}
where $n$ is the number of connected region of operator support of $O_{\mu}$, and can be identified as $n=|\partial O|/2$ for one-dimensional {\gqmbcs} in PBC. 
The Thouless time is the time when $ |O| \tth e^{-\epsilon \tth }=O(1)$ holds, and the Thouless operator support size is $|O|_{\mathrm{Th}} = e^{\epsilon t}/t$.
For simplicity, take $n=1$, i.e. there is only a single connected region of operator support. Taking the Thouless scaling limit where $t$ and $|O|$ are sent large, and $x= |O| / |O|_{\mathrm{Th}}$ is fixed, we obtain the scaling form of 2PAF of operator $O_{\mu}$ as 
\be\label{eq:2paf_nosym_scale}
\ba
\gamma^{\text{2PAF}}_{\text{no-sym}}(x)  := & \, 
  \lim_{\substack{t,|O| \to \infty \\ x= |O|/|O|_{\mathrm{Th}} }} 
q^{2|O|} \, t\,  |O|^2  \, \overline{C^{\text{RPM}}_{\mu \mu, \text{no-sym}}(t)} 
\\
& \qquad \qquad  \quad =  x^2 e^x \,. 
\, 
\ea
\ee

In the presence of TRS, we can interpret the 2PAF \eqref{eq:trs_2pcf} of operator $O_\mu$ as the partition function of a trivial statistical mechanical system described by a single state with Boltzmann weight depending on $|O|$ and $|\partial O|$, except that there is an additional sign dependent on whether $O_\mu$ is antisymmetric.  
In the absence of TRS in the one-dimensional case, by writing $n = |\partial O|$, we can cast  \eqref{eq:2pcf_nosym_1d} in the form of \eqref{eq:Z_A} with an emergent Potts model of $t-1$ states  living in region where the operator has non-trivial support. The number of states of the emergent Potts model is not $t$ (as  in the case for SFF in {\gqmbcs} with no symmetries), because the 2PAF diagram Fig.~\ref{fig:2pcf_diag} (g) vanishes due to the traceless condition of the operator.
Together, we write
\be\label{eq:2paf_partition}
\ba
\lim_{q\to\infty} \overline{C_A^{\text{RPM}}(t)} 
\propto 
\begin{cases}
     Z_A^{(t-1)\text{-Potts}}   \qquad & \text{No symmetries,}
     \\
    f(O_\mu) \,  Z_A^{\text{cluster}} \qquad & \text{TRS,}
\end{cases} 
\ea
\ee
where $ f(O_\mu) $ is defined in \eqref{eq:trs_2pcf_onekindop}, and gives rise to a sign if $O_\mu$ is antisymmetric.
For {\gqmbcs} without TRS~\cite{yoshimura2023operator}, the 2PAF \eqref{eq:2pcf_nosym_1d} displays a characteristic feature common to the SFF, namely a competition between exponential decays in $t$, and  polynomial increases in $t$, leading to a bump at an intermediate timescale without TRS at the order $q^{-2 |O|}$.
There are $t-1$ states in the emergent statistical mechanical model given by $t$ local time-parallel pairings of Feynman paths $\sigma_{\mathrm{ladder}}^{(m)}$ for $m=3,5,\dots, 2t-1$ defined in Lemma~\ref{lemma:leading_sff_diag}. One pairing $\sigma_{\mathrm{ladder}}^{m=1}$ is excluded due to the traceless operator condition.

\begin{figure}[ht]
    \centering
    \includegraphics[width=0.47 \textwidth]{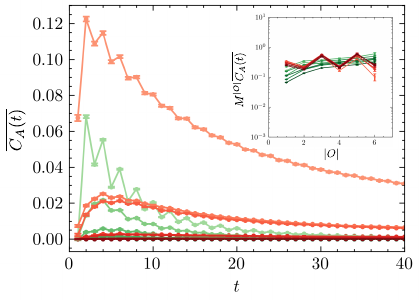}
    \caption{\textbf{2PAF numerical simulations} for the 1D Floquet TRS HRM (red) and 1D Floquet HRM (green) for $q=2$ and $L=10$ with obc. 
Main Panel: 2PAF $\overline{C_A(t)}$
against $t$  with $|O| \in [1,6]$ in increasingly dark shades. In (b) (and (c)), the operator supports are chosen to be contiguous regions positioned adjacent to the boundary. 
Inset: Normalized 2PAF $M^{|O|} \overline{C_A(t)}$ against $|O|$ with $t\in [3,9]$, where $M= q^{2 }-1 $ and $q^2$ for the TRS and no symmetries cases respectively. The normalized 2PAF does not grow (grows) in operator support size for the case with (without) TRS -- consistent with the theoretical predictions in Eq.~\eqref{eq:2pcf_nosym_1d} and Eq.~\eqref{eq:trs_2pcf}, which in turn requires the presence of negative contributions as in (a) and Eq.~\eqref{eq:trs_2pcf_onekindop}. 
}   \label{fig:2pcf_numerics}
\end{figure}

We provide numerical simulations of the 2PAF for two Floquet models (HRM and 3PM) and a chaotic Hamiltonian model with and without TRS. The models have system size $L=10$, local Hilbert space dimension $q=2$, and obc.
In Fig.~\ref{fig:2pcf_numerics_antisym}, we simulate 2PAF of symmetric operators (Pauli-\( Z \) and Pauli-\( X \)) and an antisymmetric operator (Pauli-\( Y \)) for the two Floquet models and the Hamiltonian model. 
In the presence of TRS, the symmetric operators yield positive 2PAF values, while the antisymmetric operator yields negative values. In contrast, all operators produce positive 2PAF values in the absence of TRS. 
Further, we observe that the absolute values of 2PAF for antisymmetric operators are smaller than the ones for symmetric operators (see discussion above).
The consistency of this behaviour in all three models supports Eq.~\eqref{eq:trs_2pcf_onekindop} and our assertions regarding the operator dependence of the 2PAF.
In Fig.~\ref{fig:2pcf_numerics}, we plot the operator-averaged $\overline{C_A^{\text{RPM}}(t)}$ against $t$  for the 1D Floquet HRM. The operator supports are chosen to be contiguous regions positioned adjacent to the boundary.  
For fixed operator support size \( |O| \), the averaged autocorrelation \( \overline{C_A(t)} \), computed over operators supported on region \( A \), is larger in Floquet TRS HRM than in Floquet HRM without symmetries. In both cases, $\overline{C_A(t)}$ exhibits a bump after averaging over operators in region $A$. 
This bump is not captured by \eqref{eq:trs_2pcf} in the TRS case, and requires subleading-in-$q$ analyses of 2PAF involving all $2t$ time-reversed and time-parallel pairings of Feynman paths (see discussions on finite-$q$ corrections below). 
In Fig.~\ref{fig:2pcf_numerics} inset, we plot $M^{|O|} \overline{C_A^{\text{RPM}}(t)}$ against $|O|$ for the 1D Floquet HRM, where $M=q^{2}-1$ and $q^2$ for the cases with and without TRS respectively.
%
$M^{|O|} \overline{C_A^{\text{RPM}}(t)}$ does not grow in operator size $|O|$ in the presence of TRS, but exhibits  an exponential increase  in operator size $|O|$ in the absence of TRS, supporting the large-$q$ results in Eq.~\eqref{eq:trs_2pcf} and Eq.~\eqref{eq:2pcf_nosym_1d}.

We expect the finite-$q$ corrections to lead to a suppressed 2PAF magnitude for antisymmetric operators compared to their symmetric counterparts, which is indeed observed in Fig.~\ref{fig:2pcf_numerics_antisym} for three Floquet and Hamiltonian models.  
While the unique dominant time-reversed pairing in 2PAF give negative or positive contributions depending on the symmetricity of the operators, the other subleading diagrams  from time-reversed pairings and time-parallel pairings [Fig.~\ref{fig:2pcf_diag} (a, c, d, e, f)] are always positive without operator dependence [Fig.~\ref{fig:op_dependence}].
As a result, the 2PAF for antisymmetric operators is expected to be weakened relative to symmetric operators, since the negative leading-order contribution is partially offset by these positive subleading terms.  
However, this argument is not rigorous, as the  subleading diagrams with time-reversed and time-parallel pairings do not span the full set of subleading contributions, which we will address in future work~\cite{upcoming}.

Lastly, we note that finite-$q$ corrections are necessary to accurately capture the bump observed in the operator-averaged 2PAF at $q = 2$ in the presence of TRS with discrete time translation symmetry, as shown in Fig.~\ref{fig:2pcf_numerics}.  
In the large-$q$ analysis, this bump is absent in Eq.~\eqref{eq:trs_2pcf} because the 2PAF is dominated by a single time-reversed Feynman pairing. This leading-order contribution alone is insufficient to account for the observed structure, indicating the importance of subleading-in-$q$ corrections.  
This necessity is further underscored by the fact that, at leading order, the 2PAF contributions are identical for both global and Floquet TRS RPMs, due to the same dominant pairing. However, their spectral form factors differ at leading order, driven by the summation over exponentially many subleading 2PAF terms in Eq.~\eqref{eq:sff_2paf_connection}.  
Moreover, the leading-order $2t$-ramp in the SFF of Floquet TRS {\gqmbcs} originates from the remaining $2t$ time-reversed and time-parallel pairings [Fig.~\ref{fig:2pcf_diag} (e, f, h, i)], which are absent at leading order in the 2PAF but emerge through subleading contributions.  
These observations collectively highlight the crucial role of subleading corrections in fully understanding the behaviour of the 2PAF -- a direction to be pursued in future work~\cite{upcoming}.

\subsubsection{SFF from 2PAF}

\begin{figure*}[ht]
    \centering
    \includegraphics[width=1 \textwidth]{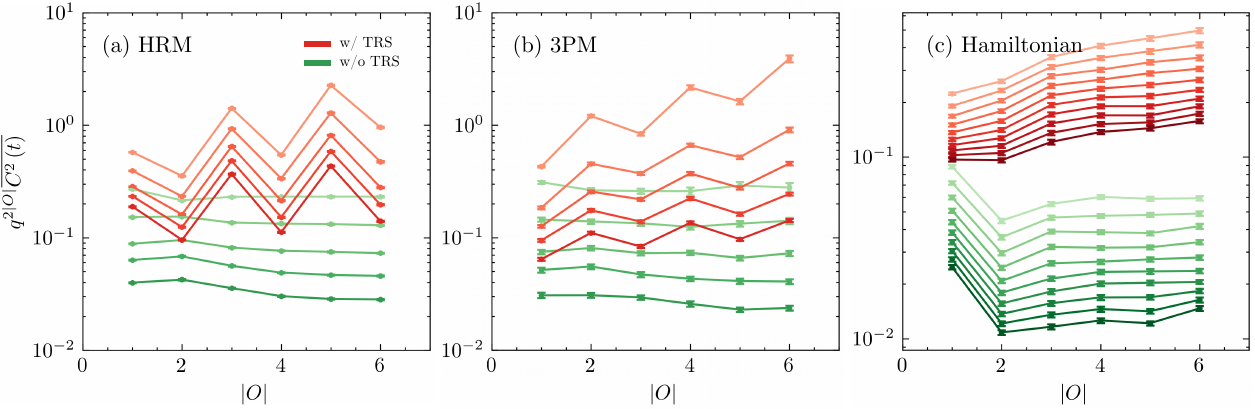}
    \caption{\textbf{Normalized 2PAF fluctuations against operator support size} for the 1D (a) Floquet HRM, (b) Floquet 3PM, and (c) Hamiltonian model in the presence (red) and absence (green) of TRS. The data are taken at $q=2$ and $L=10$ with $t\in [1,5]$ and $|O| \in [1,8]$ in increasingly dark shades. 
    The systems are configured with open boundary conditions, and the operator supports are chosen to be contiguous regions positioned adjacent to the boundary.
The normalized 2PAF exhibits an exponential scaling behaviour in $|O|$ only in the presence of TRS, as consistent with the emergent 3-state Potts model predicted in Eqs.~\eqref{eq:2paf_fluc_stat_mech}, \eqref{eq:2pcf_fluc_1d2}, and \eqref{eq:notrs_2pcf_fluc_1d}.
}\label{fig:2pcf_numerics_fluc_in_op_size}
\end{figure*}

Here we provide a third complementary approach to rederive the universal Ising scaling behaviour of SFF for global TRS {\gqmbcs} using \eqref{eq:2paf_full} and \eqref{eq:sff_2paf_connection} at sufficiently large $t$. 
%
%
For $t\gg 1$, we evaluate SFF as a sum of 2PAF for one-dimensional global TRS RPM in Eq.~\eqref{eq:2paf_full}
\be\label{eq:sff_2paf_eval}
\begin{aligned}
&\overline{K^{\text{RPM}}_{\text{g-TRS}}(t)} = \sum_{\mu \in \mathcal{P}} \overline{C_{\mu \mu, \, \text{g-TRS}}^{\text{RPM}} (t) }
\\
& =
\begin{cases}
 \Tr [T_{ \text{Ising}}^L ] =
 \left[1+e^{-\epsilon t}\right]^L + \left[1-e^{-\epsilon t}\right]^L
 \, ,  \quad & \text{pbc,}
\\
 \bra{\eta} T_{ \text{Ising}}^{L-1} \ket{\eta} = 2 \left[1 + e^{-\epsilon t}\right]^{L-1} \, ,  \quad & \text{obc,}
\end{cases}
\end{aligned}
\ee
where $\ket{\eta}= (1,1)$ accounts for the obc. Here the Ising transfer matrix $[T_{\text{Ising}}]_{ab} = \delta_{ab} + (1-\delta_{ab})e^{-\epsilon t}$ emerges after summing over all non-identity operator at each site in Eq.~\eqref{eq:2paf_full} and Eq.~\eqref{eq:sff_2paf_eval}. The two effective Ising degrees of freedom are then
\begin{enumerate}
\item time-reversed pairing for 2PAF on sites \textit{with} operator support, and
    \item time-parallel pairings for 2PAF on sites \textit{without} operator support.
\end{enumerate}
See Table~\ref{Tab:sff_3approaches} for a comparison of the effective Ising degrees of freedom with the other approaches.
There are $q^2-1$ operators giving the 2PAF diagram with time-reversed pairing of order $O(q^{-1})$. Among these contributions, $q(q-1)/2$ antisymmetric operators are associated with negative terms, and $q(q+1)/2$ symmetric operators (including the diagonal operators) give positive terms, leading to an over all factor of 1 associated to these diagrams. 
As before, take the Thouless scaling limit, where $t$ and $L$ are sent large, and $x= L / L_{\mathrm{Th}}$ with the  Thouless length $\Lth(t) = e^{\epsilon t}/t$ is held fixed. We directly obtain the SFF scaling form  from the 2PAF for global TRS {\gqmbcs} as 
\be\label{eq:ising_scaling_2paf}
\ba
&\kappa^{\text{2PAF}}_{\text{g-TRS}}(x)\equiv  
\lim_{\substack{L,t \to \infty\\ x = L/\Lth}} \sum_{\mu \in \mathcal{P}} \overline{C^{\text{RPM}}_{\mu \mu} (t)}  
=\begin{cases}
2 \cosh x \, ,  \quad & \text{pbc}\,, \\
2 \, e^{x} \, , \quad & \text{obc} \,,
\end{cases} 
\ea
\ee
and therefore we have $\kappa^{\text{SFF}}_{\text{g-TRS}}(x) = \kappa^{\text{2PAF}}_{\text{g-TRS}}(x)$. This is the third route to derive the emergent Ising scaling behaviour of SFF in {\gqmbcs} with global TRS.

\subsection{Fluctuations and three-state Potts model}\label{sec:2paf_fluc}

\begin{figure}[ht]
    \centering
    \includegraphics[width=0.4 \textwidth]{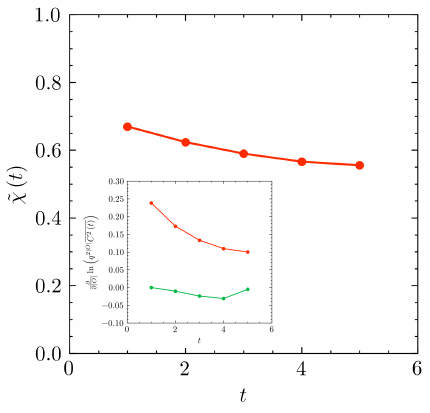}
    \caption{\textbf{Exponential scaling of 2PAF fluctuations} for the 1D Floquet TRS HRM (red) and 1D Floquet HRM (green) for $q=2$ and $L=10$ with open boundary condition as in Fig.~\ref{fig:2pcf_numerics_fluc_in_op_size}. 
 $\tilde{\chi}$ against $t$ (main panel) and the growth rate of $q^{2|O|} \overline{C^2_A(t)}$ as a function of $|O|$ (inset) for 1D Floquet TRS HRM, demonstrating that  the emergent three-state Potts model in Eq.~\eqref{eq:2pcf_fluc_1d2} captures the qualitative exponential scaling of the normalized 2PAF fluctuations as a function of operator support size $|O|$ in the presence of TRS. }\label{fig:2pcf_numerics_fluc_growth_rate}
\end{figure}

Here we show that in the presence of Floquet or global TRS, the ensemble-averaged of the fluctuations of 2PAF is mapped to the partition function of the three-state Potts model.
In the presence of TRS, at sites with non-trivial operator support, Lemma~\ref{lemma:2paf_fluc_v2} states that the three local leading diagrams of 2PAF fluctuations are  $\sigma_{(14|23)}$, $\sigma_{(12|34)}$, and $\sigma_{(13|24)}$  of order $O(q^{-2})$ [Fig.~\ref{fig:2pcf_fluc_diag} (d,e,f)]. 
In contrast, in the case without TRS, there is only a single leading diagram $\sigma_{(13|24)}$ of order $O(q^{-2})$ [Fig.~\ref{fig:2pcf_fluc_diag} (f)].
Regardless of the presence of TRS, at sites without operator support, there is a single leading local diagram given in Fig.~\ref{fig:2pcf_fluc_diag} (b). 
Note that even though the local diagrams Fig.~\ref{fig:2pcf_fluc_diag} (b, e) are disconnected, the many-body diagram containing this local diagram may not be disconnected.
Following the derivation described 
above Eq.~\eqref{eq:ising_mapping}, we derive the quantum-classical mapping for the 2PAF fluctuations of the Floquet TRS or Global TRS RPM.
While exact large-$q$ expression of 2PAF fluctuations can be obtained for general $t$, the result vastly simplifies for large $t \gg 1$, where the 2PAF fluctuations of the Floquet TRS RPM or Global TRS RPM is given by 
\be \label{eq:2paf_fluc_general_map}
\lim_{q\to \infty} \overline{ C_{\text{TRS}}^{2}(t) }= 
Z^{\text{3-Potts}}_A \,,
\ee
where we have left the operator index implicit. Defined in Eq.~\eqref{eq:Z_A}, $Z^{\text{3-Potts}}_A$  is the partition function of the three-state Potts model defined on region $A$. Let indices $a$ and $b$ label the three Potts states $\{(13|24), (12|34), (14|23)\}$. The Boltzmann weight for each bond in $A$ is given by $\mathcal{B}^{A }_{a b}= \delta_{ab} + (1-  \delta_{ab}) e^{-2\epsilon t} $ and the weight for each site is $\mathcal{B}^{A }_{a }= q^{-2}$. The complementary region $\overline{A}$ is associated with the trivial Boltzmann weight, $\mathcal{B}^{\overline{A} }_{a b}= 1$. The boundary between $A$ and $\overline{A}$ carries the Boltzmann weight $\mathcal{B}^{A \overline{A} }_{a b}= e^{-2\epsilon t}$. 
Specifically for one-dimensional Floquet TRS or global TRS RPM with pbc, we obtain using the transfer matrix approach
\be\label{eq:2pcf_fluc_1d2} 
\ba
\overline{C^{2}_{\text{TRS}}(t)}=&\,  q^{-2|O|}     \left[ 3 e^{-4\epsilon t} \right]^n   \left(1+ 2 e^{-2 \epsilon t}\right)^{|O|-n} \,.
\ea
\ee
where  $n$ is the number of connected regions of operator support, and $|O|$ is the total size of  the (possible disconnected) operator support. 
 The appearance of the three states is fundamental -- it stems from the combinatorics that there are three distinct ways of pairing four operators in the fluctuations of 2PAF, as schematically sketched in Eqs.~\eqref{eq:2paf_cont_ill_v1}, ~\eqref{eq:2paf_cont_ill_v2}, and Fig.~\ref{fig:2pcf_fluc_diag}. 
 For this reason, we expect the form of Eq.~\ref{eq:2paf_fluc_general_map} to remain valid (up to certain effective $\epsilon$) for TRS {\gqmbcs} beyond the TRS RPM.

 Consider the fluctuation of 2PAF of an operator $O$ with operator support size $|O|$, the corresponding Thouless time and Thouless operator support length are $t_{\mathrm{Th}, O} = \ln |O|/ 2 \epsilon$ and $|O|_{\mathrm{Th}} = e^{2 \epsilon t}$ respectively.
For $t \ll t_{\mathrm{Th}, O} $, on regions with operator support, the system can be heuristically viewed as a collection of expanding patches of random matrices, each adopting an identical pairing of Feynman paths, i.e. $\sigma_{(14|23)}$, $\sigma_{(12|34)}$, or $\sigma_{(13|24)}$. As time progresses, these patches grow and merge, eventually encompassing the regions with non-trivial operator support.
For $t \gg t_{\mathrm{Th}, O}$, the regions with non-trivial operator support behave like a a single large random matrix, with all regions adopting the same pairing.
For simplicity, consider 2PAF fluctuations of an operator which has support in a single connected cluster, i.e. $n=1$. 
We take the Thouless scaling limit, where $t$ and $|O|$ are sent to infinity, while $x= |O|/ |O|_{\mathrm{Th}}$ is fixed. 
For one-dimensional global or Floquet TRS RPM, we define and evaluate the scaling function of 2PAF fluctuation as 
\be
\ba
\gamma^{\text{2PAF-Fluc}}_{\text{TRS}}(x) \,  &:=  
  \lim_{\substack{t,|O| \to \infty \\ x= |O|/|O|_{\mathrm{Th}} }} 
  q^{2|O|} |O|^2 \,  \overline{C^{2}_{\text{TRS}}(t)} 
  \\
 & \qquad \qquad \qquad = \, 3 x^2 e^{2x}
\, .
\ea
\ee
Note that $L$ is bounded below by $|O|$ and so it is also sent to infinity.

\begin{figure}[h]
    \centering
    \includegraphics[width=0.46 \textwidth]{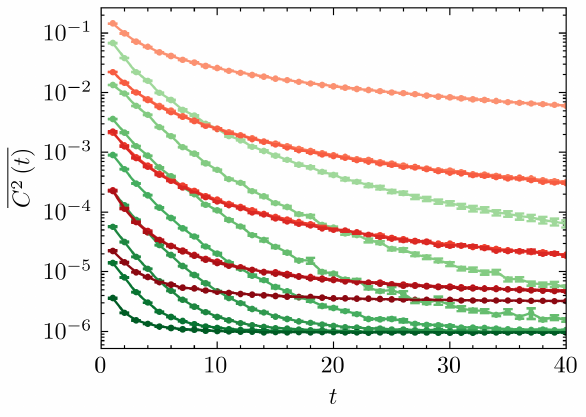}
    \caption{\textbf{2PAF fluctuations against time} for the 1D Floquet TRS HRM (red) and Floquet HRM (green) at $q=2$ and $L=10$ with $|O| \in [1,8]$ in increasingly dark shades. As in The systems are configured with open boundary conditions, and the operator supports are chosen to be contiguous regions positioned adjacent to the boundary.}    \label{fig:2pcf_numerics_fluc_vs_t}
\end{figure}

In contrast, in the absence of symmetries, RPM has only a single leading local diagram of 2PAF fluctuation [Lemma~\ref{lemma:2paf_fluc_v2}]. As a result, the fluctuation of 2PAF in the leading order is given by
\be\label{eq:notrs_2pcf_fluc}
\ba
&\lim_{q\to \infty} \overline{C^2_{\text{no-sym}} (t)}=
\Omega^{|O|} \Gamma^{|\partial O|}\, ,
\\
& \quad \Omega = q^{-2}\,  , \quad 
\Gamma = e^{-2 \epsilon t}\, ,
\ea
\ee
for RPM in general geometries. In particular, for one-dimensional RPM without symmetries with pbc, for operator $O$ with a single cluster, we have $|\partial O|=2$, and therefore
\be\label{eq:notrs_2pcf_fluc_1d}
\lim_{q\to \infty} \overline{C^2_{\text{no-sym}} (t)} =
q^{-2|O|} e^{-4\epsilon t}
\, .
\ee
Interpreting the 2PAF fluctuations without symmetries in \eqref{eq:notrs_2pcf_fluc}  as a trivial emergent statistical mechanical system with a single configuration (as in the case of 2PAF with TRS), we can collecting these results and write
\be \label{eq:2paf_fluc_stat_mech}
\ba
\lim_{q\to\infty} \overline{C^2} 
\propto 
\begin{cases}
     Z_A^{\text{cluster}}   \qquad & \text{No symmetries,}
     \\
     Z_A^{3\text{-Potts}} \qquad & \text{TRS,}
\end{cases} 
\ea
\ee
where the TRS case refers to global and Floquet TRS, but not local TRS.
A few remarks are in order. First, regardless of the presence of TRS, unlike the case of 2PAF, 2PAF fluctuations scale with $q^{-2|O|}$. 
Second, unlike the 2PAF itself, its fluctuations do not exhibit a peak at intermediate time. Within the framework of the emergent classical statistical mechanical model, this absence of a peak stems from the fact that the effective degrees of freedom at each site remain constant over time.
Third, and most notably, the presence of TRS and the associated time-reversed pairing of Feynman paths lead to a striking contrast in 2PAF fluctuation signatures.  In the absence of symmetries, $q^{2|O|} \overline{C^2_{\text{no-sym}} (t)}$ does not grow in $|O|$ in the large-$q$ limit. However, with TRS, $q^{2|O|} \overline{C^2_{\text{no-sym}} (t)}$ exhibits an \textit{exponential scaling} in $|O|$ with a growth rate derived from the three-state Potts model, $1+ 2 \chi^t$, where $\chi$ is a model-dependent parameter, identified as $e^{-\epsilon}$ in the RPM.

\begin{figure*}[ht]
    \centering
    \includegraphics[width=1 \textwidth]{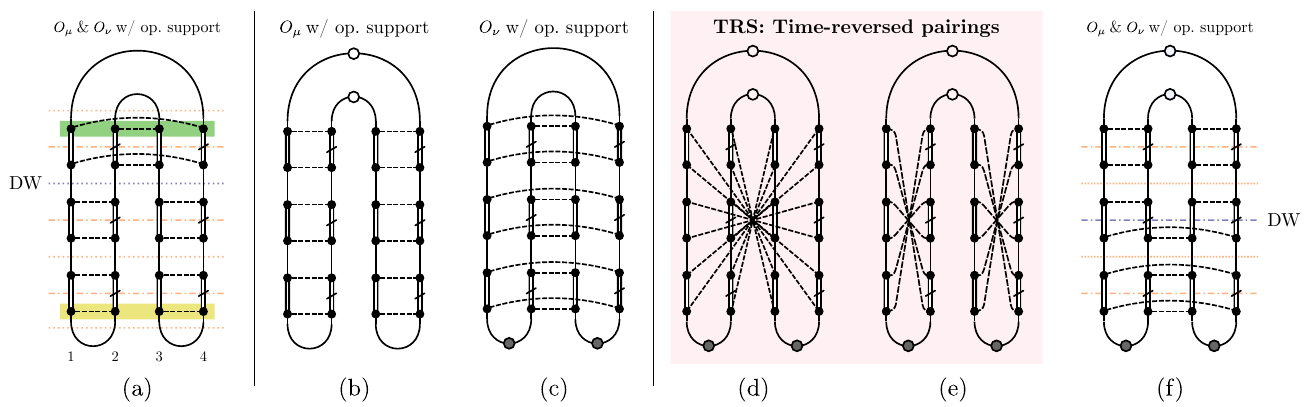}
    \caption{\textbf{OTOC diagrams in the presence of global or Floquet TRS.} Leading local diagrams of OTOC of $O_\mu$ and $O_{\nu}$ at sites where (a) $O_\mu$ and $O_{\nu}$ have no operator support; (b) only $O_\mu$ has operator support; (c) only $O_{\nu}$ has operator support; (d,e,f) both $O_\mu$ and $O_{\nu}$  have operator support.
   In panels (a) and (f), blue lines indicate the locations of domain walls that separate two types of contractions, highlighted in yellow and green in (a). Additional possible domain wall locations at leading order are shown as orange lines in (a) and (f). Diagrams with domain walls along dotted (dashed-dotted) lines are associated with positive (negative) Weingarten functions.
    (d) and (e) show leading order OTOC diagrams arising from time-reversed pairings of Feynman paths. These diagrams can contribute with negative signs when the local operator supports of \( O_\mu \) and \( O_\nu \) differ.
    }
    \label{fig:otoc_diag}
\end{figure*}

To test the large-$q$ prediction with finite-$q$ simulations, we numerically simulate the 2PAF fluctuations in 1D Floquet HRM, Floquet 3PM and a chaotic Hamiltonian model with or without TRS at $q=2$ and $L=10$ in Figs.~\ref{fig:2pcf_numerics_fluc_vs_t} to \ref{fig:2pcf_numerics_fluc_growth_rate}.  
For simplicity, we consider 2PAF fluctuations of operators supported in a single connected region adjacent to the boundary. 
In this setting, after removing a trivial normalization factor $q^{-2|O|}$, the 2PAF fluctuations in the presence of TRS can be evaluated as   $q^{2|O|}  \overline{C^{2}_{\text{TRS}}(t)}=   3 \chi_{C^2, \text{TRS}}^t   \big(1+ 2  \chi_{C^2, \text{TRS}}^t \big)^{|O|-1}$, while in the absence of TRS, we have $ q^{2|O|}  \overline{C^{2}_{\text{no-sym}}(t)}=      \chi_{C^2,\,  \text{no-sym}}^t$. Here, $\chi_{C^2,\, \text{TRS}}$ and  $\chi_{C^2,\, \text{no-sym}}$ are model-dependent parameters which for RPM is given by $e^{-4\epsilon}$. 
In Fig.~\ref{fig:2pcf_numerics_fluc_vs_t}, we plot $\overline{C^{2}_{\text{TRS}}(t)}$ against $t$. We observe that 2PAF fluctuations does not exhibit a bump as consistent with Eqs.~\eqref{eq:2pcf_fluc_1d2} and \eqref{eq:notrs_2pcf_fluc_1d}, unlike the 2PAF in Fig.~\ref{fig:2pcf_numerics}. 
In Fig.~\ref{fig:2pcf_numerics_fluc_in_op_size}, we plot $q^{2|O|}  \overline{C^{2}_{\text{TRS}}(t)}$ against $|O|$ for two Floquet and one Hamiltonian models, which demonstrate that there is a generic exponential scaling in the presence of TRS but not in the absence of TRS, supporting the analytical results in Eqs.~\eqref{eq:2pcf_fluc_1d2} and \eqref{eq:notrs_2pcf_fluc_1d}.  
In Fig.~\ref{fig:2pcf_numerics_fluc_growth_rate}, we extract a proxy of $\chi_{C^2, \text{TRS}}$ by defining and computing $\tilde{\chi} \left( t \right) := \left[ \frac{1}{2}\left( e^{\frac{\partial}{\partial \left| O \right|}\ln \left( q^{2\left| O \right|}\overline{C^2\left( t \right) } \right)}-1 \right) \right] ^{1/t}$. We also compute the growth rate  of $q^{2|O|} \overline{C^2_A(t)}$ as a function of $|O|$. We see that $\tilde{\chi}(t)$ tends towards a constant value, consistent with the claim that in the presence of TRS, the 2PAF fluctuations are governed by the emergence of the three-state Potts model in Eqs.~\eqref{eq:2paf_fluc_general_map} and \eqref{eq:2pcf_fluc_1d2}.

\section{Out-of-time-ordered correlators}\label{sec:otoc}

Under the quantum dynamics in generic quantum many-body systems, information of local perturbation spreads and scrambles among non-local degrees of freedom of the system, and can be diagnosed using the out-of-time-ordered correlator (OTOC) \cite{Shenker_2014, LarkinOvchinnikov}, defined at infinite temperature as 
\be
	F_{\mu \nu}(t):=\frac{1}{\mathcal{N}}\Tr[O_\mu (t) O_\nu(0) O_\mu(t) O_\nu(0)]   \,,
\ee
where $O_\mu(t) = U(t) O_\mu U^\dagger (t)$ is the time-evolved operator under the Heisenberg picture, and $\Nhil$ is the Hilbert space dimension. As before, $O_\mu$ is taken to be a Hermitian, unitary, traceless operator given by~\eqref{eq:gellmann}.
%
%
OTOC is a correlation function between multiple operators probed at times which are out of order, and has been instrumental in the characterisation of quantum chaos and thermalization, e.g. in the studies  of blackholes and holography~\cite{MSS, Shenker_2014}, quantum information theory~\cite{Hayden_2007, Hosur_2016, Roberts_2017}, operator spreading in many-body systems~\cite{nahum2018, keyserlingk2018, rakovszky2018, khemani2018}, quantum phase transitions~\cite{Shen_2017, Heyl_2018}, 
and experiments~\cite{G_rttner_2017, Li_2017, Landsman_2019, Nie_2020}.

OTOC has been computed in spatial- and temporal-random  quantum circuits with and without conserved quantities in~\cite{nahum2018, keyserlingk2018, rakovszky2018, khemani2018}, and Floquet random quantum circuits in the large-$q$ limit in~\cite{chan2018solution,yoshimura2023operator}. 
In the absence of TRS, the leading diagrams for Floquet random quantum circuits have been identified in~\cite{chan2018solution,yoshimura2023operator}.
In the presence of Floquet or global TRS, we provide the leading local diagrams for HRM, RPM, and RMT in Fig.~\ref{fig:otoc_diag}. 
These diagrams fall into three categories based on operator support at each site: (i) neither \( O_\mu \) nor \( O_\nu \) has support, (ii) only one of \( O_\mu \) or \( O_\nu \) has support, and (iii) both \( O_\mu \) and \( O_\nu \) have support.
For both cases (i) and (ii), even in the presence of TRS, the dominant contributions arise from local diagrams involving only time-parallel pairings. Specifically, for case (ii), the leading diagrams are unique and shown in Fig.~\ref{fig:otoc_diag} (b) and (c). In case (i), there are \( 2t + 1 \) leading diagrams, which, as in the Floquet case without TRS~\cite{chan2018solution,yoshimura2023operator}, do not involve time-reversed pairings of Feynman paths.
Among these diagrams, consider a group of four horizontally aligned dots, as illustrated by the yellow box in Fig.~\ref{fig:otoc_diag} (a). Two types of contractions are possible among these four dots:
 (1) Identity contraction, where dots in vertical slices 1 and 2 are paired, and dots in slices 3 and 4 are paired (e.g., yellow box in Fig.~\ref{fig:otoc_diag} (a)). (2) SWAP contraction, where dots in slices 1 and 4 are paired, and slices 2 and 3 are paired (e.g., green box in Fig.~\ref{fig:otoc_diag} (a)). These correspond to ``\(+\)'' and ``\(-\)'' domains in~\cite{nahum2018}, and are labeled \( A \) and \( B \) in~\cite{chan2018solution}.
Each of the \( 2t + 1 \) diagrams contains a domain of type (1) and a domain of type (2), separated by a domain wall that can appear at any of \( 2t + 1 \) positions. These configurations are depicted in Fig.~\ref{fig:otoc_diag} (a), where the blue line indicates the domain wall location, and the orange lines denote other possible positions. Diagrams with a domain wall along a dotted (dashed-dotted) line correspond to positive (negative) Weingarten functions [see~\cite{Brouwer1996, chan2018solution, yoshimura2023operator} for details].
For case (iii), the time-parallel pairings give rise to \( 2t - 1 \) diagrams shown in Fig.~\ref{fig:otoc_diag} (f), using notation consistent with Fig.~\ref{fig:otoc_diag} (a). In these diagrams, the two types of contractions appear similarly to those in case (i). However, the tracelessness condition of the operators enforces fixed contractions at the top and bottom of the diagrams, effectively requiring the presence of a domain wall (marked by the blue line in Fig.~\ref{fig:otoc_diag} (f)).
Figures~\ref{fig:otoc_diag} (d) and (e) depict leading order OTOC diagrams that arise from time-reversed pairings of Feynman paths. 
The resulting diagrams are analogous to those in Fig.~\ref{fig:2pcf_fluc_diag} (d) and (e), with the key distinction that all four operator insertions, two \( O_\mu \) and two \( O_\nu \), are connected through  time-reversed pairings.

Having established the leading local diagrams in Fig.~\ref{fig:otoc_diag} for the possible scenarios, we observe that for the OTOC between local operators \( O_{\mu} \) and \( O_{\nu} \) that are spatially separated, the leading order behaviour is insensitive to the presence of TRS. In such cases, one can directly apply previously established large-\( q \) results for the OTOC in Floquet RPM~\cite{yoshimura2023operator} and Floquet HRM~\cite{chan2018solution} without TRS.
In contrast, when \( O_{\mu} \) and \( O_{\nu} \) have overlapping support at the initial time, for example, when \( O_{\mu} = O_{\nu} \), the OTOC becomes sensitive to TRS. At a given site \( i \), the leading order contributions exhibit a sign dependence based on whether the local supports of \( O_{\mu} \) and \( O_{\nu} \) coincide. Notably, this sign dependence is different from that observed in the two-point autocorrelation function [Fig.~\ref{fig:op_dependence}]. 
The sign instead arises from the structure of time-reversed pairings that connect all four operator insertions in the OTOC at leading order in \( 1/q \).
In the presence of TRS, the OTOC for general operators \( O_\mu \) and \( O_\nu \) can be expressed in terms of a transfer matrix using the large-\( q \) techniques introduced earlier. However, the exact evaluation using this transfer matrix is technically challenging, and we leave a detailed analysis to future work.

\section{Discussion and outlook}\label{sec:outlook}
In this paper, we studied the universal signatures of quantum chaos in the presence of TRS in generic and  local quantum many-body systems. 
We presented both analytical and numerical results on SFF, PSFF, 2PAF and OTOC in random quantum circuits, serving as minimal models of {\gqmbcs} in the presence of TRS. 
For SFF, we analytically derived and numerically verified the universal Ising and generalized Potts scaling functions of SFF in {\gqmbcs} with global and Floquet TRS respectively.  We showed that the universal scaling behaviours originate from the structure of the many-body interactions between time-reversed and time-parallel pairings of Feynman paths. 
We substantiated the universality of these results by deriving the Ising scaling behaviour through three complementary approaches:  evaluating SFF (i) as double Feynman integrals, (ii) as a summation of correlation functions, and (iii) using space-time duality combined with non-Hermitian Ginibre random matrix ensembles.
Furthermore, we showed that, unlike the case of SFF, 2PAF favours time-reversed over time-parallel pairings, leading to characteristic negative and suppressed values of 2PAF for antisymmetric operators in the presence of TRS.
We demonstrated that the fluctuations of 2PAF exhibit a distinctive exponential scaling in the operator support size, with a rate governed by an emergent three-state Potts model, marking a sharp contrast to the case without TRS. 
Non-Hermitian Ginibre ensembles serve as important tools for capturing the emergent SFF scaling behaviour of {\gqmbcs} without relying on the large-$q$ limit. 
A natural question arises: How can Ginibre ensembles effectively model generic quantum many-body chaotic systems -- particularly the bump in the SFF caused by many-body interactions -- despite containing no explicit information about the quantum system's underlying many-body structure? 
The answer is analogous to the role played by Gaussian or circular ensembles in modeling temporal dynamics. In that context, the SFF captures long-time behaviour -- the ramp and the plateau -- through repeated applications of random matrices in time. Similarly, the SFF of Ginibre ensembles captures the effects of many-body interactions via \textit{repeated applications in space}, effectively encoding many-body dynamics despite the absence of microscopic detail.

 There are several promising directions for future exploration. 
First, understanding finite-$q$ corrections and non-perturbative contributions to spectral and dynamical quantities remains an intriguing challenge. Although the SFF and 2PAF are related through Eq.~\eqref{eq:sff_2paf_connection}, the time-parallel and time-reversed pairings contribute with different weights in these quantities. Consequently, examining spectral observables such as the SFF and PSFF may offer insights into finite-$q$ effects in dynamical observables like the 2PAF, and vice versa.
Additionally, in the presence of TRS, it is well known from periodic orbit theory that Sieber-Richter pairs~\cite{Sieber_2001} contribute to subleading corrections in the SFF. Investigating whether the many-body analogues [Lemma~\ref{app_lemma:sublead}] give rise to universal signatures through their interactions would be an interesting avenue of study.
Second, our large-$q$ analytical results for the SFF and 2PAF are applicable to general geometries in higher dimensions. Results in this direction will be presented in an upcoming work~\cite{upcoming}.
Third, a natural direction is to explore the interplay between global symmetries like TRS, and internal symmetries, including Abelian $U(1)$ and non-Abelian $SU(2)$ symmetries, with the aim of identifying universal signatures of quantum many-body chaos that emerge uniquely in their presence.

\section*{Acknowledgement}
We would like to thank Andrea De Luca, Tara Kalsi, Chun Y. Leung, Henning Schomerus, and especially David A. Huse and Ilyoun Na for helpful discussions. AC acknowledges support from the Royal Society grant RGS{$\backslash$}R1{$\backslash$}231444, and the Open Fellowship from EPSRC EP/X042812/1. 

\bibliographystyle{apsrev4-1}
\bibliography{biblio.bib}

\onecolumngrid

\appendix

\setcounter{equation}{0}
\setcounter{figure}{0}
\renewcommand{\thetable}{S\arabic{table}}
\renewcommand{\theequation}{S\thesection.\arabic{equation}}
\renewcommand{\thefigure}{S\arabic{figure}}
\setcounter{secnumdepth}{2}

\section{Models}\label{app:models}
In this section, we provide the algebraic definitions of random phase model (RPM), Haar-random model (HRM), the corresponding random matrix model acting in the temporal direction (RMT), and the corresponding dual Ginibre model acting in the spatial direction.

\subsection{Dynamics in time}\label{app:time_dyn}
In the context of TRS, we are interested in three classes of dynamics in time: (i) temporal-random models, and (ii) global symmetric models, and (iii) discrete time translational symmetric  (Floquet) models. Specifically, they are defined by the following unitary time evolution operators at time $t$,
\begin{IEEEeqnarray}{rrll} \label{app_eq:time_dyn}
\text{Temporal-random:}  \qquad  &  \uu_{ 1}(t;u) =& \,\prod_{t'=1}^t u(t') \, , 
\qquad  &  
\\
\text{Global TRS:} \qquad   & \uu_{ 2}(t; w ) =& \, \,  w^{T} (t)  \, w(t) \,   & 
\\
\text{Discrete time translational invariant (Floquet):} \qquad   & \uu_{3} (t; u) =& \, \,  u^t \,,   & 
 \end{IEEEeqnarray}
where $u$ and $w$ denote certain choice or distribution of unitary operators.

\subsection{Random quantum circuits}

\subsubsection{Random phase models (RPM)}
\begin{figure*}[htbp]
    \centering
    \includegraphics[width=0.95\textwidth]{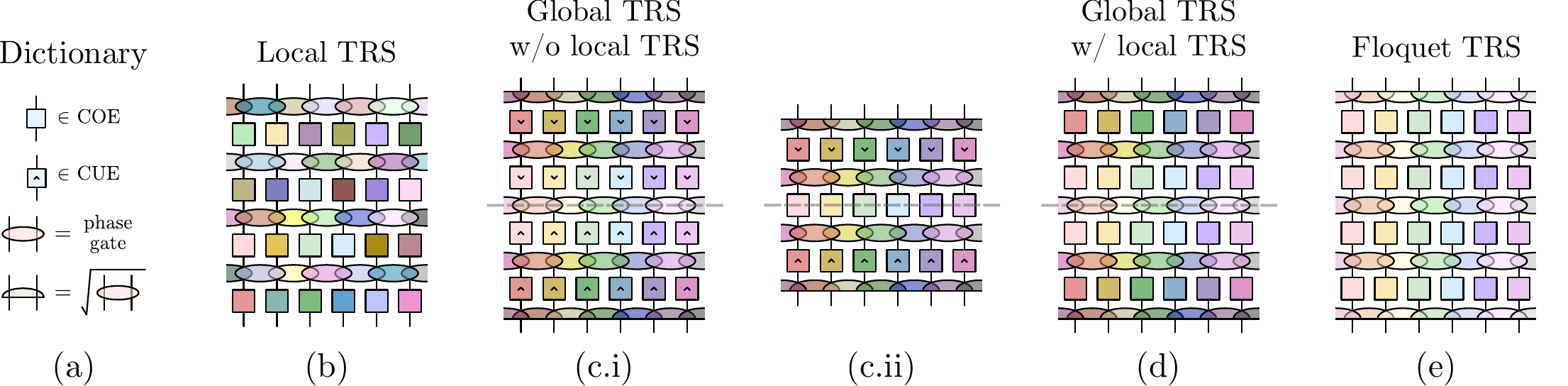}
    \caption{\textbf{TRS random phase models (RPM).} Gates with the same colour (and shade) are identical. 
    (a) Dictionary of the diagrammatical representation. (b-e) Illustration of the TRS RPM at $t=4$ with (b) local TRS; (c.i) global TRS w/o local TRS; (d) global TRS w/ local TRS; and (e) Floquet TRS. 
    Note that gates reflected across the time reversal axes (dashed lines) are identified.
    RPM at $t=3$ for global TRS w/o local TRS is given in (c.ii). The odd time for the other Global and Floquet TRS RPM share the same geometry. 
    }
    \label{sm_fig:model_RPM}
\end{figure*}

 The random phase model (RPM)~\cite{chan2018spectral} is a quantum circuit that acts on the Hilbert space $\mathbb{C}^{q^L}$ of $L$ qudits. While the RPM can be defined and the corresponding results in the main text can be derived in general geometries for arbitrary dimensions, we will define RPM in one dimension for simplicity. RPM is composed of the one-site gates Haar-random unitaries acting on the $j$-th qudit, 
   \begin{align}
 	u^{(j)}(s)  \quad  \in \text{CUE}\left(q \right) \text{ or } \text{COE}\left(q \right)\, ,
\end{align}
and two-site gates acting on the $j$-th and $(j+1)$-th qudits, 
\begin{align}
    [\Theta^{(j,j+1)}(s, \epsilon)]_{a_j a_{j+1}, b_j b_{j+1}}=\delta_{a_j, b_j} \delta_{a_{j+1}, b_{j+1}}\exp[ \imath \varphi^{(j)}_{a_j,a_{j+1}}(s,\epsilon )] \,,
\end{align}
coupling neighbouring sites via a diagonal random phase ($a_j = 1,2\ldots, q$). Each coefficient $\varphi_{a_j,a_{j+1}}^{(j)}(s, \epsilon)$ is an independent Gaussian random real variable with mean zero and variance $\epsilon$, which controls the coupling strength between neighboring spins. The one-site Haar-random gates and two-site coupling gates generally can depend on some location in time labeled by $s$. Then, we can define monolayers of one-site gates, of two-site gates and of two-site half gates (see discussion in the main text) as 
 \begin{IEEEeqnarray}{rrl} \label{app_eq:rqc_cue}
\text{1-site gate monolayer:} \qquad   & w_{\text{h}}(s;u) \, &  =  \bigotimes_{j=1}^L u^{(j)}(s) \, , \, 
\\
\text{2-site gate monolayer:}  \qquad  &   w_{\text{ph}}(s;\epsilon) \, &= \prod_{j=1}^L \Theta^{(j, j+1)}(s,\epsilon)  \,,  
\\
\text{RPM bilayer:}  \qquad  &   w_{\text{bl}}(s, \epsilon;u) \, &= 
 w_{\text{ph}}(s;\epsilon ) w_{\text{h}}(s;u)
 \\
\text{Shifted RPM bilayer:}  \qquad  &   w_{\text{sbl}}(s, \epsilon;u) \, &= 
 w_{\text{ph}}(s;\epsilon/4 ) w_{\text{h}}(s;u)  w_{\text{ph}}(s;\epsilon/4 ) 
 \end{IEEEeqnarray}
 where the variable $s$ allows the freedom for different bilayer labelled by $s$ to be independent of each other. Note that in the shifted TRS, the first and the third layers are drawn from the same ensemble. The variances in these layers are chosen such that they are half gates (as defined in the main text), such that $w_{\text{ph}}(s;\epsilon/4 )^2 = w_{\text{ph}}(s;\epsilon)$. Without TRS and other symmetries, we define the temporal-random and Floquet random phase models without symmetries as
\begin{IEEEeqnarray}{rrl} \label{app_eq:rqc_cue}
\text{Temporal-random RPM w/o symmetries:}  \qquad  &  U_{\text{l}}^{\text{RPM}}(t,L,\epsilon) \; & \equiv \uu_1 [t; 
w_{\text{bl}}(s, \epsilon;u_{\mathrm{CUE}})
] \, , 
\\
\text{Floquet RPM w/o symmetries:} \qquad   &U_{\text{f}}^{\text{RPM}}(t,L, \epsilon) \; & \equiv \uu_{ 3}[t; w_{\text{bl}} ( 1, \epsilon ; u_{\mathrm{CUE}})] \, . 
 \end{IEEEeqnarray}
Now we define random phase models with various types of TRS. The local TRS RPM can be defined by stacking bilayers that are independently drawn. To define global TRS RPM with or without local TRS, we define the two circuits 
\begin{IEEEeqnarray}{rl}
w_{\text{g-TRS-1}}(t,L) \; &=
\begin{cases} 
w_{\text{h}}(s;u_{\mathrm{CUE}})
\left[ \prod_{s=1}^{(t-1)/2} 
w_{\mathrm{bl}}(s, \epsilon; u_{\text{CUE}}) \right]  
w_{\text{ph}}(s;\epsilon/4 ) \,, \quad \qquad \qquad & t \text{ odd,} 
\\
 w_{\text{ph}}(s;\epsilon/4 )
  w_{\text{h}}(s;u_{\mathrm{CUE}})
\left[ \prod_{s=1}^{t/2-1} 
w_{\mathrm{bl}}(s, \epsilon; u_{\text{CUE}})  \right]  
w_{\text{ph}}(s;\epsilon/4 )   \,,  \qquad \qquad  & t \text{ even,}
\end{cases}
\\
w_{\text{g-TRS-2}}(t,L) \; &= 
\begin{cases} 
w_{\text{h}}(s;u_{\mathrm{CUE}})
\left[ \prod_{s=1}^{(t-1)/2} 
w_{\mathrm{bl}}(s, \epsilon; u_{\text{COE}}) \right]  
w_{\text{ph}}(s;\epsilon/4 ) \,, \quad \qquad \qquad & t \text{ odd,} 
\\
 w_{\text{ph}}(s;\epsilon/4 )
  w_{\text{h}}(s;u_{\mathrm{COE}})
\left[ \prod_{s=1}^{t/2-1} 
w_{\mathrm{bl}}(s, \epsilon; u_{\text{COE}})  \right]  
w_{\text{ph}}(s;\epsilon/4)   \,,  \qquad \qquad  & t \text{ even,}
\end{cases}
 \end{IEEEeqnarray}
 which governs the first half of the global TRS circuits. The entire global TRS circuits are then generated by extending the above circuits with their transpositions, as defined below.
 %
 Lastly, the Floquet TRS RPM can be defined by repeated action of a shifted RPM bilayer.
 Together, we define 
\begin{IEEEeqnarray}{rrl} \label{app_eq:rqc_coe}
\text{Local TRS RPM:}  \qquad  &  U_{\text{l-TRS}}^{\text{RPM}}(t,L) \; & \equiv \uu_1 [t; w_{\text{bl}} ( \cdot  ,\epsilon ; u_{\mathrm{COE}})] \, . 
\\
\text{Global TRS RPM w/o local TRS:}  \qquad  &  U_{\text{g-TRS-1}}^{\text{RPM}}(t,L) \; & \equiv \uu_2 [w_{\text{g-TRS-1}} (t , L)] \, ,
\\
\text{Global TRS RPM w/ local TRS:}  \qquad  &  U_{\text{g-TRS-2}}^{\text{RPM}}(t,L) \; & \equiv \uu_2 [ w_{\text{g-TRS-2}}  ( t, L)] \, ,
\\
\text{Floquet TRS RPM:}  \qquad  &  U_{\text{f-TRS}}^{\text{RPM}}(t,L) \; & \equiv \uu_3 [t; w_{\text{sbl}} ( 1 ,\epsilon ; u_{\mathrm{COE}})]\, . \label{app_eq:floq_trs_RPM} 
 \end{IEEEeqnarray}
See Fig.~\ref{sm_fig:model_RPM} for illustrations of the TRS RPMs.

\begin{figure*}[htbp]
    \centering
    \includegraphics[width=0.95\textwidth]{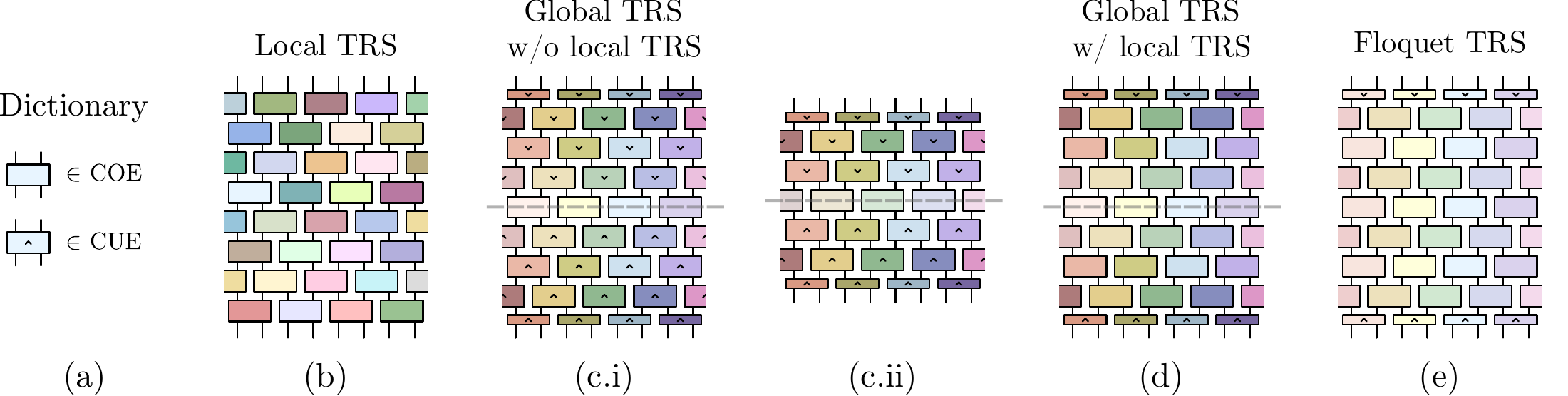}
    \caption{\textbf{TRS Haar-random models (HRM).} Gates with the same colour (and shade) are identical. 
    (a) Dictionary of the diagrammatical representation. (b-e) Illustration of the TRS HRM at $t=4$ with (b) local TRS; (c.i) global TRS w/o local TRS; (d) global TRS w/ local TRS; and (e) Floquet TRS. 
    Note that gates reflected across the time reversal axes (dashed lines) are identified.
    HRM at $t=3$ for global TRS w/o local TRS is given in (c.ii). The odd time for the other Global and Floquet TRS HRM share the same geometry. 
    }
    \label{sm_fig:model_bwm}
\end{figure*}

\subsubsection{Haar-random models (HRM)}

 The Haar-random model (HRM)~\cite{nahum2017} is a one-dimensional quantum circuit that acts on the Hilbert space $\mathbb{C}^{q^L}$ of $L$ qudits. The HRM has a brick-wall geometry composed of two-site gates acting on the $j$-th and $(j+1)$-th qudits, 
  \begin{align}
 	u^{(j,j+1)}(s) \quad  \in \text{CUE}\left(q^2 \right) \text{ or } \text{COE}\left(q^2 \right)\, ,
\end{align}
where $\text{CUE}(N)$ and $\text{COE}(N)$ are the CUE and COE of the unitary group of $N$-by-$N$ unitary matrices. The Haar-random gates generally can depend on some location in time labeled by $s$.
 In the brick-wall geometry, we only consider circuits with even system sizes. 
 %
 Specifically, we define a monolayer and a bilayer of two-site gates in the brick-wall geometry respectively as
 \begin{IEEEeqnarray}{rrl} \label{app_eq:rqc_cue}
\text{Brick-wall monolayer:} \qquad   & w_{m}(s;u) \, &  =  \bigotimes_{\substack{j \in 2 \mathbb{Z} + m\, \text{mod}\,2}}u^{(j,j+1)}(s) \, , \, 
\\ \label{app_eq:rqc_cue2}
\text{Brick-wall bilayer:}  \qquad  &   w_{\text{bl}}(s;u) \, &= w_2(s;u) \, w_1(s;u) \,,  
\\ \label{app_eq:rqc_cue3}
\text{Shifted TRS brick-wall bilayer:}  \qquad  &   w_{\text{sbl}}(s) \, &= w_1(s, u_{\mathrm{CUE}})^T \, w_2(s;u_\mathrm{COE}) \, w_1(s;u_\mathrm{CUE}) \,,  
 \end{IEEEeqnarray}
 where the variable $s$ allows the freedom for different bilayer labelled by $s$ to be independent of each other. Note that in the shifted TRS, the first and the third layers are drawn from the same ensemble. 
Without TRS or other symmetries,  we define the temporal-random and Floquet Haar-random models without symmetries as
\begin{IEEEeqnarray}{rrl} \label{app_eq:rqc_cue_def}
\text{Temporal-random HRM w/o symmetries:}  \qquad  &  U_{\text{l}}^{\text{HRM}}(t,L) \; & \equiv \uu_1 [t; w_{\text{bl}} ( \cdot  ; u_{\mathrm{CUE}})] \, , 
\\
\text{Floquet HRM w/o symmetries:} \qquad   &U_{\text{f}}^{\text{HRM}}(t,L) \; & \equiv \uu_{ 3}[t; w_{\text{bl}} ( 1 ; u_{\mathrm{CUE}})] \, . 
 \end{IEEEeqnarray}
Now we define the Haar-random models with various types of TRS. The local TRS HRM can be defined by stacking brick-wall bilayers that are independently drawn. To define global TRS HRM with or without local TRS, we define the two circuits 
\begin{IEEEeqnarray}{rl}
w_{\text{g-TRS-1}}(t,L) \; &=
 \prod_{s=1}^{t} 
w_{s}(s; u_{\text{CUE}}) \,, 
\\
w_{\text{g-TRS-2}}(t,L) \; &= 
w_{t}(t; u_{\text{CUE}})
\left[ \prod_{s=2}^{t-1} 
w_{s}(s; u_{\text{COE}}) 
\right]
w_{1}(1; u_{\text{CUE}})  \,,
 \end{IEEEeqnarray}
 which governs the first half of the global TRS circuits. The entire global TRS circuits are then generated by extending the above circuits with their transpositions, as defined below. Note that the first and last layers of $w_{\text{g-TRS}}$ are drawn from CUE, such that, upon the extension with transposition, we can construct a COE gate across the time reversal axis (dashed lines in Fig.~\ref{fig:model}), using the construction, $u_{\mathrm{COE}}= u^T_{\mathrm{CUE}} u_{\mathrm{CUE}}$.
 Lastly, the Floquet TRS HRM can be defined by repeated action of a brick-wall bilayer.
 Together, we define for both odd and even $t$,
\begin{IEEEeqnarray}{rrl} \label{app_eq:rqc_coe}
\text{Local TRS HRM:}  \qquad  &  U_{\text{l-TRS}}^{\text{HRM}}(t,L) \; & \equiv \uu_1 [t; w_{\text{bl}} ( \cdot  ; u_{\mathrm{COE}})] \, . 
\\
\text{Global TRS HRM w/o local TRS:}  \qquad  &  U_{\text{g-TRS-1}}^{\text{HRM}}(t,L) \; & \equiv \uu_2 [w_{\text{g-TRS-1}} (t+1 , L)] \, ,
\\
\text{Global TRS HRM w/ local TRS:}  \qquad  &  U_{\text{g-TRS-2}}^{\text{HRM}}(t,L) \; & \equiv \uu_2 [ w_{\text{g-TRS-2}} ( t+1, L)] \, ,
\\
\text{Floquet TRS HRM:}  \qquad  &  U_{\text{f-TRS}}^{\text{HRM}}(t,L) \; & \equiv \uu_3 [t; w_{\text{sbl}} ( 1 )] \, . 
 \end{IEEEeqnarray}
 Notice that the Floquet TRS HRM has a shift of a half step in the time direction, such that the first and the last layers in time contain half gates of COE, which belong to CUE. See Fig.~\ref{sm_fig:model_bwm} for illustrations of TRS HRMs.

\subsubsection{3-parameter models (3PM)}
The 3-parameter model (3PM)~\cite{znidaric2022, huang2023outoftimeorder} is a one-dimensional quantum circuit that acts on the Hilbert space $\mathbb{C}^{q^L}$ of $L$ qudits. In this manuscript we focus on Floquet TRS 3PM and Floquet 3PM (no time-reversal symmetry), although one may likewise introduce 3PM with local TRS or global TRS in direct analogy with the HRM. The  3PM shares the same brick-wall geometry as the HRM, but with two-site gates acting on the $j$-th and $(j+1)$-th qudits given by
  \begin{align}
 	& u^{(j,j+1)} = \left[u^{(j)}_1 \otimes u^{(j)}_2\right]
    \exp\left(i \sum_{\mu = x, y, z} a_\mu \sigma_j^\mu  \sigma_{j+1}^\mu \right)
    \left[u^{(j)}_1 \otimes u^{(j)}_2\right]\, ,
\end{align}
where for each $j$ and $i$, the unitary matrix $u^{(j)}_i$ is independently drawn from the COE of $2$-by-$2$ unitary matrices if TRS is present, and from the CUE if TRS is not present. We take $(a_x, a_y, a_z) = (0.3, 0.4, 0.5)$ as in \cite{huang2023outoftimeorder}. In the presence of TRS, we define half gates for 3PM as 
\be
 u_{\half}^{(j,j+1)} = 
    \exp\left( \frac{i}{2} \sum_{\mu = x, y, z} a_\mu \sigma^\mu_{j}  \sigma^\mu_{j+1} \right)
    \left[u^{(j)}_1 \otimes u^{(j)}_2\right]\, ,
\ee
with $u^{(j)}_i \in$ COE, such that $u^{(j,j+1)} = \left[ u_{\half}^{(j,j+1)}\right]^T u_{\half}^{(j,j+1)}$. 
Similar to HRM, we define the monolayer and bilayers in brick-wall geometry for 3PM as
\begin{IEEEeqnarray}{rrl} \label{app_eq:rqc_3pm}
\text{Brick-wall monolayer:} \qquad   & w_{m}(s;u) \, &  =  \bigotimes_{\substack{j \in 2 \mathbb{Z} + m\, \text{mod}\,2}}u^{(j,j+1)}(s) \, , \, 
\\ \label{app_eq:rqc_3pm2}
\text{Brick-wall bilayer:}  \qquad  &   w_{\text{bl}}(s;u) \, &= w_2(s;u) \, w_1(s;u) \,,  
\\ \label{app_eq:rqc_3pm3}
\text{Shifted TRS brick-wall bilayer:}  \qquad  &   w_{\text{sbl}}(s) \, &= w_1(s, u_{\half})^T \, w_2(s;u) \, w_1(s;u_\half) \,. 
 \end{IEEEeqnarray}
With these ingredients, we can define the Floquet TRS 3PM and Floquet 3PM as 
\begin{IEEEeqnarray}{rrll} \label{app_eq:rqc_3pm}
\text{Floquet TRS 3PM:}  \qquad  &  U_{\text{f-TRS}}^{\text{3PM}}(t,L) \; & \equiv \uu_3 [t; w_{\text{sbl}} ( 1 )] \, , \qquad  &\text{ with } u^{(j)}_{i} \in \text{COE,} 
\\ 
\text{Floquet 3PM:}  \qquad  &  U_{\text{Floquet}}^{\text{3PM}}(t,L) \; & \equiv \uu_3 [t; w_{\text{bl}} ( 1 )] \, , \qquad  &\text{ with } u^{(j)}_{i} \in \text{CUE.} 
 \end{IEEEeqnarray}
3PMs exhibit signatures of generic quantum many-body chaos at $q=2$ (in contrast to $q=3$ for RPMs) and its non-TRS variations have previously been studied in~\cite{znidaric2022, huang2023outoftimeorder}.

\subsection{Hamiltonian models}
For completeness, we repeat in this appendix the definitions of time-independent Hamiltonian models  used to test the applicability of the results derived from random quantum circuits. 
Specifically, we define 
\be
\ba
H= & \, H_{\text{XYZ}} + H_{\text{NNN}}
\\
H_{\text{XYZ}}=& \, \sum_{\langle i,j \rangle} \left( 
a_x \sigma^x_i \sigma^x_j + a_y \sigma^y_i \sigma^y_j + a_z \sigma^z_i \sigma^z_j   
\right)
\\
H_{\text{NNN}} =& \, 
\begin{cases}
\sum_{\langle\!\langle i,j\rangle\!\rangle} h^{\mathrm{GOE}}_{i,j} \qquad \qquad & \text{TRS} 
\\
\sum_{\langle\!\langle i,j\rangle\!\rangle} h^{\mathrm{GUE}}_{i,j}
\qquad \qquad & \text{No symmetries} 
\end{cases}
\ea
\ee
where we take $(a_x, a_y, a_z) = \frac{1}{4}(1.2,1,0.8)$ at which the model displays chaotic behaviour. $h^{\mathrm{GOE/GUE}}_{i,j}$ are random matrices drawn from the Gaussian orthogonal ensemble (GOE) and Gaussian  unitary ensemble (GUE)  respectively. $\langle i,j \rangle$ and  ${\langle\!\langle i,j\rangle\!\rangle}$ denote all pairs of nearest-neighbour sites and next-to-nearest-neighbour (NNN) sites respectively. By adding the NNN terms drawn from the GOE, we break all symmetries of the XXZ model except the continuous time translational symmetry and TRS. When the NNN terms are drawn from the GUE instead, we additionally break the TRS.

\subsection{Random matrix models (RMT)}\label{app:matrixmodel}

\begin{figure*}[htbp]
    \centering
    \includegraphics[width=0.7\textwidth]{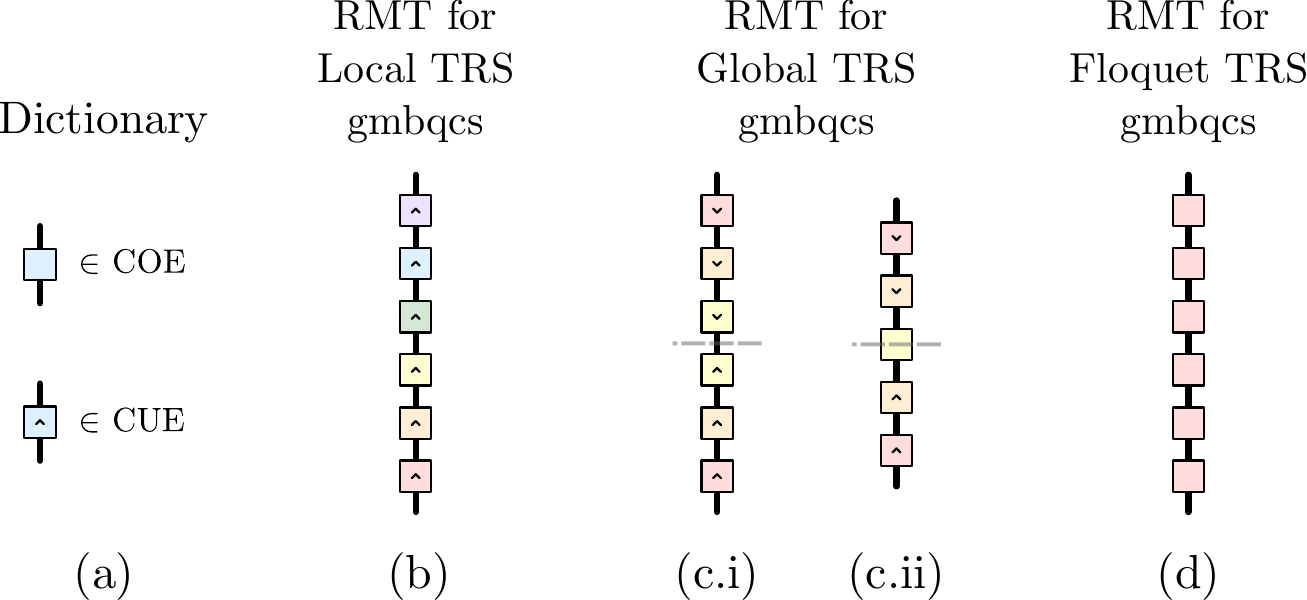}
    \caption{\textbf{TRS random matrix models (RMT).} (a) Dictionary of the diagrammatical representations of COE and CUE. (b-d) Illustration of the RMT quantum circuits at $t=6$ for RMTs that model {\gqmbcs} with (b)  local TRS; (c) global TRS with or without local TRS; and (d) Floquet TRS. 
    For completeness, RMT for global TRS {\gqmbcs} at $t=5$ is given in (c.ii), although it trivially coincides with (c.i) for gate choices from CUE and COE. 
    The odd time for the other TRS RMT models are defined similarly. 
    Note that gates reflected across the time reversal axes (dashed lines) are identified.
  Note also that both global TRS quantum circuits with local TRS, and global TRS quantum circuits without local TRS are modelled (in sufficiently late time) by global TRS RMT without local TRS. 
    }
    \label{sm_fig:model_mm}
\end{figure*}

Random matrix theory (RMT) can be used to model the dynamics of generic many-body quantum chaotic systems ({\gqmbcs}) in sufficiently late time scales or sufficiently small energy scales. These models are defined via the time evolution operators in ~\ref{app:time_dyn} with unitary $u$ drawn from the circular unitary ensemble (CUE) or the circular orthogonal ensemble (COE) of $N$-by-$N$ unitaries. Specifically, in the absence of symmetries, we define random matrix models (RMT) that model temporal-random and Floquet  {\gqmbcs} respectively with

\begin{IEEEeqnarray}{rrl} \label{app_eq:time_dyn}
\text{RMT for temporal-random {\gqmbcs}  w/o symmetries:}  \qquad  &  U_{\text{l}}^{\text{RMT}}(t) \; & \equiv \uu_{1}(t;u_{\mathrm{CUE}}) \, , 
\\
\text{RMT for Floquet {\gqmbcs}  w/o symmetries:} \qquad   &U_{\text{f}}^{\text{RMT}}(t) \; & \equiv \uu_{ 3}(t;u_{\mathrm{CUE}}) \, , 
 \end{IEEEeqnarray}
 where $u_{\mathrm{CUE}}$ is  drawn from the CUE.  
 To define RMT for global TRS {\gqmbcs}, we define 
\begin{IEEEeqnarray}{rl}
w_{\text{g-TRS}}(t) \; &=
\begin{cases} 
u_{\mathrm{CUE}} \prod_{s=1}^{(t-1)/2} 
u_{\text{CUE}}  \,, \quad \qquad \qquad & t \text{ odd,} 
\\
\prod_{s=1}^{t/2} 
u_{\text{CUE}}  \,,  \qquad \qquad  & t \text{ even,}
\end{cases}
 \end{IEEEeqnarray}
 which governs the first half of the global TRS circuits. The RMT for global TRS {\gqmbcs} are then generated by extending the above circuit with its transposition, as defined below.
Together with the other dynamics in the presence of TRS, we define
\begin{IEEEeqnarray}{rrll} \label{app_eq:time_dyn}
\text{RMT for {\gqmbcs} w/ local TRS:}  \qquad  &  U_{\text{l-TRS}}^{\text{RMT}}(t) \;&\equiv \uu_{ 1}(t;u_{\mathrm{CUE}}) \, ,&  
\\
\text{RMT for {\gqmbcs} w/ global TRS:} \qquad   &  U_{\text{g-TRS}}^{\text{RMT}}(t) \; &\equiv \uu_{ 2}(t; w_{\text{g-TRS}} ) \, ,  \qquad \qquad  &
\\
\text{RMT for {\gqmbcs} w/ Floquet TRS:} \qquad   &U_{\text{f-TRS}}^{\text{RMT}}(t) \;&\equiv \uu_{ 3}(t;u_{\mathrm{COE}}) \, ,   &
 \end{IEEEeqnarray}
 where $u_{\mathrm{COE}} $ and $u_{\mathrm{CUE}}$ are  drawn from the COE and CUE respectively. 
Notice that in many-body quantum circuits with TRS, even though 2-site unitary gates may be locally TRS, the corresponding bi-layer of such gates can break TRS since the bi-layer is not transpose invariant due to its geometry (see e.g. Fig.~\ref{fig:model} (c,d)). Consequently, the RMT for \textit{both} {\gqmbcs} w/ local TRS and RMT for {\gqmbcs} w/ global TRS are composed of unitaries from the CUE (as opposed to COE). 
In contrast, the Floquet TRS model  satisfies the local TRS and global TRS condition, and the corresponding RMT is made of unitaries from the COE. 
See Fig.~\ref{sm_fig:model_mm} for illustrations of the TRS RMTs. Note also that under this convention, global TRS RMT  for odd and even times coincide.

\subsection{Dual Ginibre models}
\begin{figure*}[htbp]
    \centering
    \includegraphics[width=0.25\textwidth]{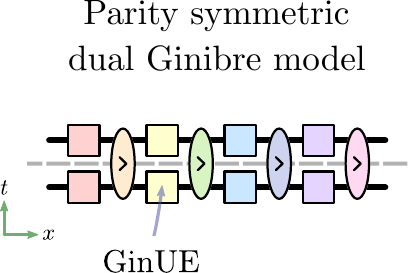}
    \caption{\textbf{TRS dual Ginibre model is parity symmetric.} In sufficiently large system size $L$, {\gqmbcs} with TRS as modelled by quantum circuit can be modelled by the parity-symmetric dual Ginibre model, that is, a RMT model of non-Hermitian matrices drawn from the Ginibre ensemble that is parity symmetric in the dual Hilbert space.  The dashed lines refer to the TRS inversion axes.
    }
    \label{sm_fig:model_dual}
\end{figure*}

It has been observed that universal behaviour of Ginibre ensemble emerges in dual (spatial) evolution of spatially-extended generic many-body quantum chaotic systems~\cite{Shivam_2023}. 
Here we construct a coarse-grained dual Ginibre model for global TRS  many-body quantum chaotic systems, as illustrated in Fig.~\ref{fig:ginue}. 
Consider a system with two (dual) coarse-grained sites, each with Hilbert space size $\mathbb{C}^{\NN}$ such that  the total Hilbert space has the size $\mathbb{C}^{\NN^2}$. 
The non-unitary evolution in the spatial direction is defined by
\be\label{app_eq:gin1}
\begin{aligned}
V_{ \text{g-TRS}}^{\mathrm{Gin}}(L) =&  \prod_{x}^L v(x) \;,
\end{aligned}
\ee
where for spatially-random systems, $v(x)$ does not correlate with $v(x')$ for $x\neq x'$. For translational invariant systems, $v(x) = v$ for all $x$. 
To incorporate global TRS in the dual Ginibre model, we choose 
\begin{IEEEeqnarray}{rl} \label{app_eq:gin2}
v(x) =& \; v_2 (x) v_1 (x) \, ,  
\\
v_2 (x) =&  \; G_2(x;\NN^2) + S G_2(x;\NN^2)  S \equiv   \tilde{G}_2 (x;\NN^2 ) \,  , 
\\
v_1 (x) =&  \; G_1(x;\NN) \otimes G_1(x;\NN)  \,,
 \end{IEEEeqnarray}
where $S$ is the two-site swap operator, defined by its action $S\ket{a_1, a_2}= \ket{a_2, a_1}$ where $\ket{a_1,a_2}$ with $a_i = 1,2,\dots, N$ is the computational basis in the Hilbert space.  
$G_i(x;{N})$ are independently drawn from the complex Ginibre ensembles of $\NN$-by-$\NN$ non-Hermitian random matrices, such that $\overline{[G_i]_{a a'}(x;\NN ) [G^*_i]_{b b'}(x;\NN)}=\delta_{ab} \delta_{a' b'} \sigma_i^{-2}(\NN)$ with $\NN$-dependent variance $\sigma_1^2(N)=\NN$ and $\sigma_2^2(N^2)=2\NN^2$.  We will leave the second argument on the size of the Ginibre matrices implicit in the rest of the manuscript. See Fig.~\ref{sm_fig:model_dual} for illustrations of the TRS dual Ginibre model.


\section{Weingarten functions for COE}\label{app:wg_coe}
Here we review some properties of the Weingarten function of circular orthogonal ensemble (COE), following \cite{Brouwer1996}. To make this section self-contained, we restate the integral equation over $N$-by-$N$ unitary matrix $\ucoe$ and its complex conjugate $\ucoe^*$ drawn from the COE~\cite{samuel1980wg1, creutz1978wg2, Brouwer1996, MATSUMOTO_2012}, which is given by
\be \label{app_eq:coe_formula}
\begin{aligned}
\overline{
[\ucoe]_{a_1 a_2} \dots 
[\ucoe]_{a_{2n-1} a_{2n}} 
[\ucoe^*]_{b_1 b_2} \dots 
[\ucoe^*]_{b_{2m-1} b_{2m}}}
= \delta_{n,m} \sum_{\sigma \in S_{2n}} \wg_{\mathrm{COE}}[\sigma; N] \prod_{i=1}^{2n} \delta_{a_i, b_{\sigma(i)}} \;,
\end{aligned}
\ee
where $\wg_{\mathrm{COE}}[\sigma; N]$ is the Weingarten function of the COE, taking arguments from $\sigma \in S_{2n}$ in the symmetry group of $2n$ objects. Define $\sigma_\mathrm{o}\in S_n$ by $\sigma_\mathrm{o}(n) = \left\lceil \sigma (2n-1)/2 \right\rceil$, and similarly $\sigma_\mathrm{e}\in S_n$ by $\sigma_\mathrm{e}(n) = \left\lceil \sigma (2n)/2 \right\rceil$, where $\left\lceil  \cdot \right\rceil$ is the ceiling function. The Weingarten function is found to depend only on the cycle structure, $\{c_1, \dots, c_k \}$,  with $c_i$ the length of the $i$-th cycle of $\sigma_\mathrm{o}^{-1} \sigma_\mathrm{e}$. With an abuse of notation, we write  $\wg_{\mathrm{COE}}[\sigma; N]= \wg_{\mathrm{COE}}[ \{c_1, \dots, c_k\}; N]$. The Weingarten function can be generated under a recursive relation~\cite{Brouwer1996, MATSUMOTO_2012}, 
\be
\ba
(N + c_1) \wg_\COE[\{c_1, \dots, c_k\}] 
+ &\sum_{p+q= c_1} \wg_\COE[\{ p,q,c_2, \dots, c_k\}] 
\\
&+ 2 \sum_{j=2}^k c_j \wg_\COE[\{ c_1 + c_j,c_2, \dots,c_{j-1} , c_{j+1}, \dots , c_k\}] 
= \delta_{c_1, 1}  \wg_\COE[\{ c_2, \dots, c_k\}] 
\ea
\ee
where we set $\wg[\emptyset] \equiv 1$, and we have left the argument $N$ implicit above. 
As examples, we have $\wg_\COE[\{ 1\}, N ]= 1/(N+1)$, 
$\wg_\COE[\{1, 1\}, N ]= (N+2)/N(N+1)(N+3)$, and
$\wg_\COE[\{2\}, N ]= -1/N(N+1)(N+3)$.

\section{SFF of global TRS random matrix theory}\label{app:global_trs_mm_sff}
Here we compute the SFF of the global TRS RMT, $U^{\mathrm{RMT}}_{\text{g-TRS}}$, as defined in Appendix ~\ref{app:matrixmodel} and Fig.~\ref{fig:model}. To this end, we will use the CUE formula for the second moment of a pair of $N$-by-$N$ unitary matrix and its conjugate is given by $\overline{[u_{\mathrm{CUE}}]_{a_1 b_1 } [u^*_{\mathrm{CUE}}]_{a_1'b_1'} [u_{\mathrm{CUE}}]_{a_2 b_2} [u^*_{\mathrm{CUE}}]_{a_2'b_2'} } = (N^2-1)^{-1} \left( \delta_{a_1 a_1'} \delta_{b_1 b_1'}
\delta_{a_2 a_2'} \delta_{b_2 b_2'}
+
\delta_{a_1 a_2'} \delta_{b_1 b_2'}
\delta_{a_2 a_1'} \delta_{b_2 b_1'}
\right) -[N(N^2-1)]^{-1} \left( 
\delta_{a_1 a_1'} \delta_{b_1 b_2'}
\delta_{a_2 a_2'} \delta_{b_2 b_1'}
+
\delta_{a_1 a_2'} \delta_{b_1 b_1'}
\delta_{a_2 a_1'} \delta_{b_2 b_2'}
\right)$, where the prefactors follow from the Weingarten functions for the CUE.
 By performing a folding procedure as in Fig.~\ref{fig:fold}, all identical unitaries drawn from the CUE and their conjugates are stacked locally. By treating the delta functions in the equation above as Wick contractions, and using these contractions as basis states, we can construct a transfer matrix $T$ in the \textit{time} direction in contrast to transfer matrix in the \textit{spatial} direction for the many-body RPM models, and for in the dual Ginibre model. The computation leads to
\be
K_{\text{g-TRS}}^{\mathrm{RMT}}(t,N) =
\left\langle l \right|
T_1 \left(T_2 T_1\right)^{\lceil t \rceil /2 -1}
\left| r \right\rangle
\,,  \quad 
T_1 =\frac{1}{N^2 -1} 
\begin{pmatrix}
1 & - \frac{1}{N}\\
- \frac{1}{N} & 1 
\end{pmatrix} \,,
\quad 
T_2 = 
\begin{pmatrix}
N^2 & N \\
N & N^2 
\end{pmatrix} \,,
\ee 
where $T_1$ accounts for the Weingarten functions, and $T_2$ accounts for the contraction due to the SFF structure.  $\left\langle l \right| = (N, N)$ and $\left|r \right\rangle = (N, N)^T$. $\lceil \cdot \rceil$ is the ceiling function.
Direct computation leads to 
\be
K_{\text{g-TRS}}^{\mathrm{RMT}}(t,N)= 2 \frac{N}{N+1}   \stackrel{N\rightarrow \infty}{\longrightarrow }  2 \,,
\ee
which is independent of $t$, and tends to 2 in large $N$ as expected.

\section{SFF diagrams}\label{app:sff_diag_proof}
In this section, we provide the proofs of Lemma~ \ref{lemma:leading_sff_diag} which is restated below. In doing so, we discuss technical details required to adapt the proof for three models, RPM, HRM and RMT, in four kinds of TRS symmetries. These variations will also be useful for Appendix \ref{app:2paf_proof}.
\begin{figure*}[htbp]
    \centering
    \includegraphics[width=1\textwidth]{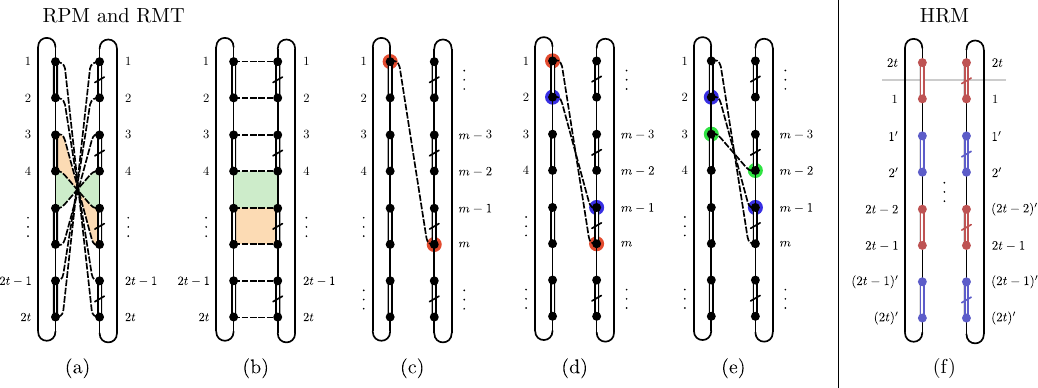}
    \caption{\textbf{Proof for leading local SFF diagrams for Floquet TRS models.} (a, b) Diagrammatical representation of SFF $K(t)$ with loops of consecutive solid and dashed lines (green), and loops of consecutive double and dashed lines (orange). 
    The SFF diagrams are the (a) twisted diagram \eqref{app_eq:twisted_sff} with $m=2t$, and (b) ladder diagram \eqref{app_eq:ladder_sff} with $m=1$. 
    (c) We can construct leading SFF twisted diagrams by choosing $\sigma(1)=m$ with even $m$. In order to form 1-loops (for reasons described in the main text), we are forced to choose $\sigma(2)=(m-1) \Mod{2t}$, as in (d), and $\sigma(3)=(m-2) \Mod{2t}$, as in (e). For HRM, the diagrammatical representation of the SFF on site $i$ is given in (f), where the red (blue) nodes denote unitaries that act on sites $i-1$ and $i$ (sites $i$ and $i+1$). Due to cyclicity, the representation of SFF allows us to shift the half gates around the trace, and the grey line denotes the boundary between the top and bottom half gates in Fig.~\ref{fig:model} (e,f,g). 
    }  
    \label{app_fig:sff_proof}
\end{figure*}

\begin{lemma}
\label{app_lemma:leadingsffcaseD}
\textbf{Leading local spectral form factor (SFF) diagrams for generic quantum many-body quantum chaotic systems with TRS and discrete time translational symmetry.} Consider the SFF \eqref{eq:sff_def} for (i) Floquet TRS random phase model (RPM), (ii) Floquet TRS Haar-random model (HRM), and (iii) the random matrix model (RMT) for Floquet TRS {\gqmbcs}.
%
For (i) and (ii) at each physical site in the order of the local Hilbert space dimension $q$, and for (iii)  in the order of matrix dimension $N$, the $2t$ leading SFF diagrams are of order $O(1)$ and are given by $t$ twisted diagrams,
\be\label{app_eq:twisted_sff}
\sigma_{\mathrm{twisted}}^{(m)}(i) = (m - i +1) \Mod{2t}
 \,, 
\ee
for $m=2,4, \dots, 2t$, and $t$ ladder diagrams,
\be\label{app_eq:ladder_sff}
\sigma_{\mathrm{ladder}}^{(m)}(i) = (m + i -1) \Mod{2t}
 \,, 
\ee
where $m=1,3,\dots, 2t-1$, and we take the convention with $2t \Mod{2t} = 2t$. The conventions of permutation labels are given in Fig.~\ref{app_fig:sff_proof} (a) for models (i) and (iii), and Fig.~\ref{app_fig:sff_proof} (f) for model (ii).
For (ii), the leading SFF diagrams have contractions on unitaries acting on sites $i-1$ and $i$ (in red), and contractions on unitaries acting on sites $i$ and $i+1$ (in blue) both taking identical permutation values given by \eqref{app_eq:twisted_sff} or \eqref{app_eq:ladder_sff}. 
\end{lemma}

\noindent \textit{Examples.} For the twisted SFF diagrams of Floquet TRS RPM on each site at $t=3$, see Fig.~\ref{fig:sff_diag} (d), (c), and (b)  for  $\sigma_{\mathrm{twisted}}^{(m)}$ with $m=2,4,$ and $6$ respectively. For the ladder SFF diagrams of the same model, Fig.~\ref{fig:sff_diag} (e), (f), and (g)  for  $\sigma_{\mathrm{ladder}}^{(m)}$ with $m=1,3,$ and $5$ respectively.

\vspace{0.5cm}
\begin{proof}
Consider the SFF $K(t)$ of models (i), (ii) and (iii), which can be diagrammatically represented as Fig.~\ref{app_fig:sff_proof} (a,b). Note that for (ii), the half gates discussed in the main text [e.g. Fig.~\ref{fig:model_hrm_maintext} (c) top-most and bottom-most layers] can be combined to a COE gate using the cyclic property of trace in the SFF [Fig.~\ref{app_fig:sff_proof} (f)].  
We use $\mathcal{N}$ to denote the dimension of a single tensor network worldline, i.e. $\mathcal{N}=q$ for model (i) and (ii), and $\mathcal{N}=N$ for model (iii). 
We use $\tilde{\mathcal{N}}$ to denote the dimension of the COE gates in each model, i.e. $\tilde{\mathcal{N}}=q, q^2, N$ for models (i), (ii), and  (iii) respectively. 
Under Eq.~\eqref{eq:coe_formula} and the diagrammatical approach, there are two types of loops in the evaluation of $K(t)$: Loops composed of consecutive solid and dashed lines [Fig.~\ref{app_fig:sff_proof} (a,b) green], and loops composed of consecutive double and dashed lines [Fig.~\ref{app_fig:sff_proof} (a,b) orange]. Loops of the first type arise from the evaluation of the delta functions in Eq.~\eqref{eq:coe_formula}, and correspond to a sum over the Hilbert space of a single degrees of freedom, giving rise to a factor of $\mathcal{N}$ per loop. Loops of the second type in each diagram give rise to factors of $\mathcal{N}$, due  to the Weingarten $\wg[(c_1, c_2, \dots , c_k); \tilde{\mathcal{N}}] = O(\tilde{\mathcal{N}}^{k - 2\sum_{i=1}^k c_i})$. 
Together, these facts imply that the leading SFF diagrams correspond to diagrams with the most number of loops. 

Now we identify the order of a leading SFF diagram in $\mathcal{N}$. For all models, there always exists the SFF diagram, $\sigma_{\mathrm{ladder}}^{(m=1)}$, on each site (for (iii), the RMT models, there is effectively only one site). This diagram is associated with $\mathcal{N}^t$ from $t$ 1-loops  of the first type, and with $\wg[\{ 1,1,\dots, 1\};\tilde{\mathcal{N}}] = O(\tilde{\mathcal{N}}^{-t})$ with $t$ 1-loops of the second type [see e.g. Fig.~\ref{fig:sff_diag} (e) for $K(t=3)$]. Note that for (ii), this local diagram is assigned $\sqrt{\wg[\{ 1,1,\dots, 1\};\tilde{\mathcal{N}}]} = O(\mathcal{N}^{-t})$ \cite{chan2018solution}. Together, the SFF diagram $\sigma_{\mathrm{ladder}}^{(m=1)}$ is of order $O(1)$ under the normalization in \eqref{eq:sff_def}.
Thus, any diagrams of order $O(\mathcal{N}^{-1})$ is subleading, and in fact, any diagrams with loops longer than length 1 are subleading diagrams in $\mathcal{N}$. This is because for a given diagram at time $t$, the total length of all loops are fixed, and therefore having higher number of longer loops implies fewer number of loops.


For fixed odd $m$, consider the diagram labelled by the permutation $\sigma_{\mathrm{ladder}}^{(m)}$ in \eqref{app_eq:ladder_sff}. \cite{chan2018solution} has shown that this diagram is the only leading diagrams for a given odd $m$. 

For fixed even $m$, consider the permutation $\sigma_{\mathrm{twisted}}^{(m)}$ with $\sigma_{\mathrm{twisted}}^{(m)}(1)=m$, as in Fig.~\ref{app_fig:sff_proof} (c). In order to have 1-loop, the condition of $\wg[\{ 1,1,\dots, 1\}, \tilde{\mathcal{N}}]$ forces $\sigma(2) = (m-1) \Mod{2t}$, such that a 1-loop of loop type 2 is formed [Fig.~\ref{app_fig:sff_proof} (d)]. Any other choices $\sigma(2) \neq (m-1) \Mod{2t}$ would lead to a loop with loop length larger than 1. Next, $\sigma(3) = (m-2) \Mod{2t}$ so that a new loop of type 1 with length 1 is formed [Fig.~\ref{app_fig:sff_proof} (e)]. Again, any other choices of $\sigma(3) \neq (m-2) \Mod{2t}$  leads to a loop with loop length larger than 1. This procedure can be iterated until $\sigma(2t) = (m-2t+1) \Mod{2t}$. 

For  TRS HRM, the diagrammatical representation of the SFF on site $i$ is given in Fig.~\ref{app_fig:sff_proof} (f), where the red (blue) nodes denote unitaries that act on sites $i-1$ and $i$ (sites $i$ and $i+1$). Contractions are allowed among dots of the same colour from the unitaries to their conjugation due to the bond constraint (see main text and \cite{chan2018solution}). In contrast, for TRS RPM, we use two permutations $\sigma_{\mathrm{r}}, \sigma_{\mathrm{b}} \in S_{2t}$ to label the possible contractions. The joint permutation of these two permutations lie within $S_{4t}$. By repeating the same reasoning described in the last paragraphs onto this joint permutation, one can identify \eqref{app_eq:twisted_sff} and \eqref{app_eq:ladder_sff} as the $2t$ leading diagrams that respect the bond constraint.
\end{proof}

\begin{figure*}[htbp]
    \centering
    \includegraphics[width=1\textwidth]{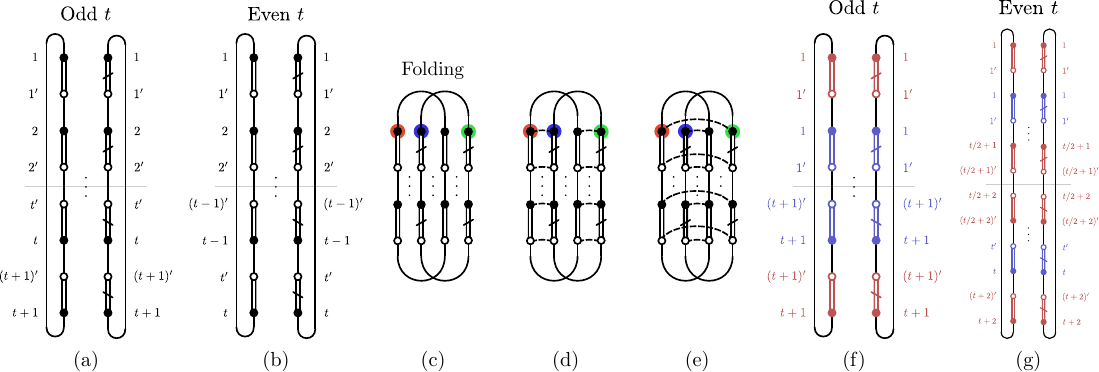}
    \caption{\textbf{Proof for leading local SFF diagrams for global TRS models w/o local TRS.} 
    (a, b) Diagrammatical representation of SFF $K(t)$ of global TRS RPM w/o local TRS and its corresponding RMT for odd and even time $t$. 
    (c) The diagram can be folded [see Fig.~\ref{fig:fold}] such that identically-drawn unitary matrices are placed together locally. The red dot can be contracted to either the blue or the green dot. The former choice imposes on all other contractions, giving the leading diagram in (d), and the latter giving (e). 
    For HRM, the diagrammatical representation of the SFF on site $i$ is given in (f) and (g) for odd and even time $t$, where the red (blue) nodes denote unitaries that act on sites $i-1$ and $i$ (sites $i$ and $i+1$).
    For (a,b,f,g), the grey lines denote the time reversal axis.
    }  
    \label{app_fig:sff_proof2}
\end{figure*}

\begin{figure}[ht]
    \centering
    \includegraphics[width=0.4 \textwidth]{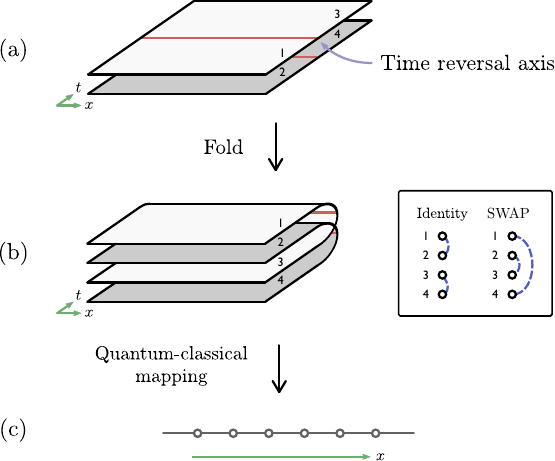}
    \caption{\textbf{Global TRS and folding.} The folding procedure and quantum-classical mapping. (a) Illustration of $U(t,L)$ and $U^\dagger(t,L)$ in light and dark grey respectively. The TRS inversion axis is illustrated in red. 
    (b) Folding can be performed such that, for global TRS models, the identically-sampled unitary gates are placed as stacks of unitaries in proximity to each other.    (c) A quantum-classical mapping can be performed such that the large-$q$ limit of SFF is mapped to the partition function of a classical ferromagnetic Ising model. 
    This folding procedure naturally arises in global TRS circuits and is also used in~\cite{khanna2025randomquantumcircuitstimereversal}.    }
    \label{fig:fold}
\end{figure}

\begin{corollary} \label{app_cor:globaltrs_w_localtrs}
\textbf{Leading local SFF diagrams for generic quantum many-body quantum chaotic systems with global TRS and local TRS.} Consider the SFF \eqref{eq:sff_def} for (i) global TRS RPM w/ local TRS, (ii) global TRS HRM  w/ local TRS. 
At each physical site in the order of local Hilbert space dimension $q$, the $2$ leading SFF diagrams are of order $O(1)$ and are given by the twisted diagram 
$\sigma_{\mathrm{twisted}}^{(m=2t)}(i)$
and the ladder diagram  
$\sigma_{\mathrm{ladder}}^{(m=1)}(i)$ defined in \eqref{app_eq:twisted_sff} and \eqref{app_eq:ladder_sff} in \ref{app_lemma:leadingsffcaseD}. 
\end{corollary}

\noindent Note again that we use the permutation convention as in Fig.~\ref{app_fig:sff_proof} (a), and the  convention $2t \Mod{2t} =  2t$. The SFF for RMT for global TRS {\gqmbcs} can be evaluated exactly in Appendix~\ref{app:global_trs_mm_sff}.

\vspace{0.2cm}

\begin{corollary}\label{app_cor:globaltrs_wo_localtrs}
\textbf{Leading local SFF diagrams for generic quantum many-body quantum chaotic systems with global TRS  without local TRS.} Consider the SFF \eqref{eq:sff_def} for (i) global TRS RPM w/o local TRS, (ii) global TRS HRM  w/o local TRS, and (iii) the RMT for global TRS {\gqmbcs} (w/ or w/o local TRS).
For (i) and (ii) at each physical site in the order of the local Hilbert space dimension $q$, and for (iii)  in the order of matrix dimension $N$, the $2$ leading SFF diagrams are of order $O(1)$ and are given by the twisted diagram $\sigma_{\mathrm{twisted}}^{(m=\tmax)}(i)= \tmax-i+1$, and the ladder diagram  $\sigma_{\mathrm{ladder}}(i) =i $. 
For (i) and (iii), $\tmax=t+\delta_{0,(t+1)\Mod{2}}$. 
For (ii) and odd $t$, we have $\tmax = t+1$ for both contractions of the red gates, $\sigma_{\mathrm{r}}$, and of the blue gates, $\sigma_{\mathrm{b}}$.
For (ii) and even $t$, we have $\tmax = t+2$ for $\sigma_{\mathrm{r}}$ and $\tmax = t$ for $\sigma_{\mathrm{b}}$.
\end{corollary}

\begin{corollary}\label{app_cor:localtrs}
\textbf{Leading local SFF diagrams for generic quantum many-body quantum chaotic systems with local TRS.} Consider the SFF \eqref{eq:sff_def} for (i) local TRS RPM, (ii) local TRS HRM, and (iii) the RMT for local TRS {\gqmbcs}.
For each (effective) physical site, in the order of local Hilbert space size $q$ for (i) and (ii), and in the order of matrix dimension $N$ for (iii), the leading SFF diagrams  is of order $O(1)$ and is given by  the ladder diagram  
$\sigma_{\mathrm{ladder}}^{(m=1)}(i)$ defined by \eqref{app_eq:ladder_sff}  in \ref{app_lemma:leadingsffcaseD}.
\end{corollary}

\begin{proof}
For \ref{app_cor:globaltrs_w_localtrs} and \ref{app_cor:globaltrs_wo_localtrs},
the proofs follow directly from the proof of \ref{app_lemma:leadingsffcaseD}, except that due to the lack of translational invariance in time, only the 2 leading SFF diagrams, one twisted diagram and one ladder diagram, are allowed. It is instructive to illustrate the proof by folding the diagram across the time reversal axis, grey lines in Fig.~\ref{app_fig:sff_proof2} (a,b,f,g), such that identically-drawn unitary matrices are placed together locally [Fig.~\ref{fig:fold}]. The first contraction can be assigned between red and blue, or red and green dots in Fig.~\ref{app_fig:sff_proof2} (c), which imposes on all the other contractions, giving rise to the leading order SFF ladder diagram in (d), and the SFF twisted diagram in (e). The above argument applies to the RPM and RMT in Fig.~\ref{app_fig:sff_proof2} (a-e), and also the HRM in Fig.~\ref{app_fig:sff_proof2} (f,g).
For \ref{app_cor:localtrs}, the additional lack of global TRS implies that only a single leading SFF diagram, a ladder diagram, specified in is allowed. 
\end{proof}

\section{2PAF and 2PAF fluctuation diagrams}\label{app:2paf_proof}
In this section, we provide the proofs of Lemma~\ref{lemma:leading2pcf} and Lemma~\ref{lemma:2paf_fluc_v2}, which we restate below.

\begin{lemma}
\label{app_lemma:leading2pcf}
\textbf{Leading local two-point autocorrelation function (2PAF) for generic quantum many-body quantum chaotic systems with or without TRS.} Consider the 2PAF \eqref{eq:2paf_def} with local diagram  labeled as in Fig.~\ref{fig:2pcf_diag} (a). 
For  (i) random phase model (RPM), (ii) Haar-random model (HRM), and (iii) the random matrix model (RMT) with Floquet TRS / with global TRS / without symmetries.
For (i) and (ii) at each physical site in the order of the local Hilbert space dimension $q$, and for (iii)  in the order of matrix dimension $N$, 
 the leading local 2PAF diagram is of order $O(1)$ given by the ladder contraction in $S_{2t}$   [Fig.~\ref{fig:2pcf_diag} (b)] 
\be \label{app_eq:2paf_diag_noop}
\sigma_{\mathrm{ladder}}^{(m=1)}(i) =  i \, ,   \quad \text{TRS / No symmetries, }\ee
for sites without local operator support with or without TRS. 
For sites with non-trivial operator support, the leading local 2PAF diagram  in the presence of TRS is of order $O(q^{-1})$ or $O(N^{-1})$, and is given by the twisted contractions  in $S_{2t}$   [Fig.~\ref{fig:2pcf_diag} (d)] 
\be \label{app_eq:2paf_diag}
\sigma_{\mathrm{twisted}}^{(m=2t)}(i) = 2t - i +1 \, , \quad \text{TRS.}
\ee
For sites with non-trivial operator support, for RPM, HRM, and RMT in the absence of TRS, the leading local 2PAF diagram are of order $O(q^{-2})$ or $O(N^{-2})$, and are given by ladder contractions in $S_{2t}$ [Fig.~\ref{fig:2pcf_diag} (h,i)] 
%
\be \label{app_eq:2paf_diag_notrs}
\sigma_{\mathrm{ladder}}^{(m)}(i) = (m + i -1) \Mod{2t}\, ,  \quad \text{No sym.,} 
\ee
for $m=3,5,\dots, 2t-1$.
\end{lemma}

\begin{lemma}\label{app_lemma:2paf_fluc_v2}
\textbf{Leading local two-point autocorrelation function (2PAF) fluctuations diagrams for generic quantum many-body quantum chaotic systems with or without TRS.} Consider the 2PAF fluctuations with local diagram labeled as in Fig.~\ref{fig:2pcf_fluc_diag} (a).
For  (i) random phase model (RPM), (ii) Haar-random model (HRM), and (iii) the random matrix model (RMT) with Floquet TRS / with global TRS / without symmetries.
For (i) and (ii) at each physical site in the order of the local Hilbert space dimension $q$, and for (iii)  in the order of matrix dimension $N$,  the leading local 2PAF fluctuations diagram is of order $O(1)$ given by the ladder contraction in $S_{4t}$   [Fig.~\ref{fig:2pcf_fluc_diag} (b)] 
\be
\sigma_{\mathrm{ladder}}^{(m=1)}(i) =  i \, ,   \quad \text{TRS / No symmetries, }\ee
for sites without local operator support with or without TRS. 
 In the presence of Floquet TRS or global TRS, for sites with non-trivial operator support, the leading local 2PAF fluctuations diagrams are of order $O(q^{-2})$ or $O(N^{-2})$, and are given by contractions in $S_{4t}$ [Fig.~\ref{fig:2pcf_fluc_diag} (d,e,f) respectively] 
\begin{IEEEeqnarray}{rll} 
\sigma_{(14|23)}(i) = & \,(4t +1 - i) \Mod{4t}  \, , 
\quad &    \nonumber
\\
\sigma_{(12|34)}(i) = & \,  (2t +1 - i) \Mod{4t} \, ,  
\quad &  \text{TRS.}   \label{app_eq:2paf_fluc_diag}
\\
\sigma_{(13|24)}(i) = &  \, (2t+ i) \Mod{4t}  \, , 
\quad &   \nonumber
 \end{IEEEeqnarray}
In the absence of TRS, for sites with non-trivial operator support, the leading local 2PAF fluctuations diagram is of order $O(q^{-2})$ or $O(N^{-2})$, and is given by contractions in $S_{4t}$ [Fig.~\ref{fig:2pcf_fluc_diag} (f)]
 \begin{IEEEeqnarray}{rll} 
\sigma_{(13|24)}(i) = &  \, (2t+ i) \Mod{4t}  \, 
,  \quad & \text{No sym.}  \label{app_eq:2paf_fluc_notrs_diag}
 \end{IEEEeqnarray}
\end{lemma}

Note that the parametrization of the 2PAF diagram [Fig.~\ref{fig:2pcf_diag} (a)]  and 2PAF fluctuations diagram [Fig.~\ref{fig:2pcf_fluc_diag} (a)] need to be varied depending on model and the specific TRS symmetries, as in the case in Appendix~\ref{app:sff_diag_proof}.

\begin{proof}
For brevity, we will provide the proof in the setting of (i) Floquet TRS RPM and (iii) RMT for Floquet TRS {\gqmbcs}, and the proof can be extended to other symmetries like global TRS, and also to HRM as in \ref{app:sff_diag_proof}. We will first prove the leading diagrams for 2PAF fluctuation and then the leading diagram for 2PAF.

Consider the 2PAF fluctuation $\overline{C^2_{\mu\mu}(t)}$ of operator $O_\mu$, which can be diagrammatically represented as Fig.~\ref{fig:2pcf_fluc_diag} (a), where the indices of the matrix elements of local unitary gates in the circuit  $U$ are labelled from $1$ to $4t$, and similarly  for the complex conjugation of local unitary gates.
We use $\mathcal{N}$ to denote the dimension of a single tensor network worldline, i.e. $\mathcal{N}=q$ for model (i) and (ii), and $\mathcal{N}=N$ for model (iii). 
We use $\tilde{\mathcal{N}}$ to denote the dimension of the COE gates in each model, i.e. $\tilde{\mathcal{N}}=q, q^2, N$ for models (i), (ii), and  (iii) respectively. 
Under Eq.~\eqref{eq:coe_formula} and the diagrammatical approach, there are two types of loops in the evaluation of 2PAF: Loops composed of consecutive solid and dashed lines, and loops composed of consecutive double and dashed lines. Loops of the first type arise from the evaluation of the delta functions in Eq.~\eqref{eq:coe_formula}, and correspond to a sum over the Hilbert space of a single degrees of freedom, giving rise to a factor of $\mathcal{N}$ per loop. Loops of the second type in each diagram give rise to factors of $\mathcal{N}$, due  to the Weingarten $\wg[(c_1, c_2, \dots , c_k); \tilde{\mathcal{N}}] = O(\tilde{\mathcal{N}}^{k - 2\sum_{i=1}^k c_i})$. 
Together, these facts imply that the leading diagrams correspond to diagrams with the most number of loops.
Since the total lengths of all loop is fixed for a given diagrammatical representation of an observable, we hope to search for diagrams with the highest possible number of short loops are of high order in $\mathcal{N}$. 

Start from choosing $\sigma(1)$. To form short loop, we choose $\sigma(1) = 1$, but this diagram must vanish due to traceless properties of operator $O_\mu$. To form the shortest loop while fulfilling $\Tr[O_\mu^2]=O(\mathcal{N})$ (forming diagrams with higher power than 2 leads to longer loops and lower order in $\mathcal{N}$), $\sigma(1)$ must take the value $2t$, $2t+1$, or $4t$. These are the only three possibilities since there are only four operators in $\overline{C^2_{\mu\mu}(t)}$.
Suppose $\sigma(1) = 2t$. In order to form a short loop, choose $\sigma(2)=2t-1$. Iterate this procedure until we have $\sigma(j)= 2t+1-j$ for $j=1,2,\dots,2t$. Next, choose $\sigma(2t+1) = 4t$ so that we satisfy $\Tr[O_\mu^2]=O(\mathcal{N})$ for the remaining two operators (operators 3 and 4). Repeat the short loop argument and we arrive $\sigma_{(12|34)}$ in Eq.~\eqref{app_eq:2paf_fluc_diag}. $\sigma_{(12|34)}$ must be one of the leading 2PAF fluctuation diagrams, since the loops involving $O_\mu$-s are the shortest possible loops satisfying  $\Tr[O_\mu^2]=\mathcal{N}$, and the loops not involving $O_\mu$-s are also the shortest possible loops.  

Repeat the argument by choosing $\sigma(1)= 2t+1$ and we arrive $\sigma_{(13|24)}$. Repeat the argument for $\sigma(1)= 4t$ and we arrive $\sigma_{(14|23)}$. These are the only three leading diagrams because $\sigma(1) = 2t, 2t+1, 4t$ are the only three possible ways to form shortest loops involving $O_\mu$-s while satisfying traceless operator condition. This concludes the proof for the leading local 2PAF fluctuation diagrams.

The 2PAF $\overline{C_{\mu\mu}(t)}$ of operator $O_\mu$ is  diagrammatically represented as Fig.~\ref{fig:2pcf_diag} (a), where the indices of the matrix elements of local unitary gates in the circuit  $U$ are labelled from $1$ to $2t$, and similarly  for the complex conjugation of local unitary gates. Apply the same arguments used for 2PAF fluctuation, the leading 2PAF diagram must have $\sigma(1) = 2t$ (since the values $2t+1$ and $4t$ are not available), which leads to $\sigma_{\mathrm{twisted}}^{(m=2t)}$ in Eq.~\eqref{app_eq:2paf_diag}. 

Lastly, for the case without TRS symmetries,
the proof for the leading diagram for 2PAF fluctuations Eq.~\eqref{app_eq:2paf_fluc_notrs_diag} follow immediately from Eq.~\eqref{app_eq:2paf_fluc_diag} since the $\sigma_{(13|24)}$ consists of ladder contractions, and 
$\sigma_{(14|23)}$ and $\sigma_{(12|34)}$ consists of twisted contractions. The result for leading diagram for 2PAF  Eq.~\eqref{app_eq:2paf_diag_notrs} is stated in \cite{yoshimura2023operator}, and the proof can be constructed by extending the tools provided above.

\end{proof}

\section{Additional data and numerical methods}\label{app:numerics}

In this section, we provide additional numerics and describe our numerical methods to obtain the data. 
Generally, we gather data for 1D HRM with $q=2$ and $L= 4,6,8,10,12$, and for 1D RPM for $q=3$ and $L= 4,5,6,7,8$. The number of realizations is of order $10^5$.

 In Fig.~\ref{fig:sff_gtrs2}, we obtain SFF and the scaling collapse of SFF for global TRS RPM and HRM with local TRS (in addition to the scaling collapse of global TRS models without local TRS), which again shows excellent agreement between infinite-$q$ Ising scaling function and finite-$q$ data. 

 In Fig.~\ref{fig:sff_ftrs_additional}, we provide numerical data for Floquet SFF RPM (in addition to Floquet SFF HRM in the main text), which shows that SFF exhibits an initial bump which increases in system size. The bump disappears at the Thouless time which increases in system size.

In Fig.~\ref{app_fig:psff_3pm}, we provide numerical data of PSFF for Floquet 3PM which display behaviour qualitatively similar to PSFF of Floquet HRM provided in the main text. In  Fig.~\ref{app_fig:psff_3pm} (a), we plot PSFF against time $t$, which shows a characteristics bump in PSFF qualitatively described by the large-$q$ solution of the Floquet TRS RPM in the main text. In Fig.~\ref{app_fig:psff_3pm} (b), we plot PSFF against $L_A$ and find that PSFF grows at least exponentially with $L_A$ with a higher growth rate in the presence of TRS compared to the case without TRS. 

 Now we describe how we obtain $\Lth(t)$ and the scaling collapse  in Fig.~\ref{fig:sff_gtrs_num} and \ref{fig:sff_ftrs_num}, following the approach in \cite{chan2021trans}. For concreteness, we will use the global TRS case as an example, and the methodology is also applied to the Floquet TRS case. For fixed $t$, we compute the SFF by contracting a tensor network with increaseing system size $L$ in the space direction, allowing us to access small $t$ but large $L$ data.  
 We then plot the SFF against system size in Fig.~\ref{fig:sff_vs_L}. For each SFF data set at fixed $t$, and  we find a length $\tilde{L}(t)$ such that the finite-$q$ SFF scaling function coincides with the infinite-$q$ scaling function  $K(t,\tilde{L}) = \kappa_{\text{g-TRS}}^{\text{RPM}} (x_0)$
 at a choice of $x_0$, which we choose to be $x_0=3$. 
Since $x_0=\tilde{L}(t)/L_{\mathrm{Th}}(t)$ by construction, we obtain the corresponding Thouless length $L_{\mathrm{Th}}(t)=\tilde{L}(t)/x_0$, which are plotted in Fig.~\ref{fig:sff_lth_vs_t}. 
Lastly, we rescale the horizontal axis of Fig.~\ref{fig:sff_vs_L} to be $x=L/L_{\mathrm{Th}}(t)$, so the finite-$q$ numerics and the infinite-$q$ scaling function can be compared directly as in Figs.~\ref{fig:sff_gtrs_num} and \ref{fig:sff_ftrs_num}.

\begin{figure*}[t]
\begin{minipage}[t]{0.47\textwidth}
\includegraphics[width=\linewidth,keepaspectratio=true]{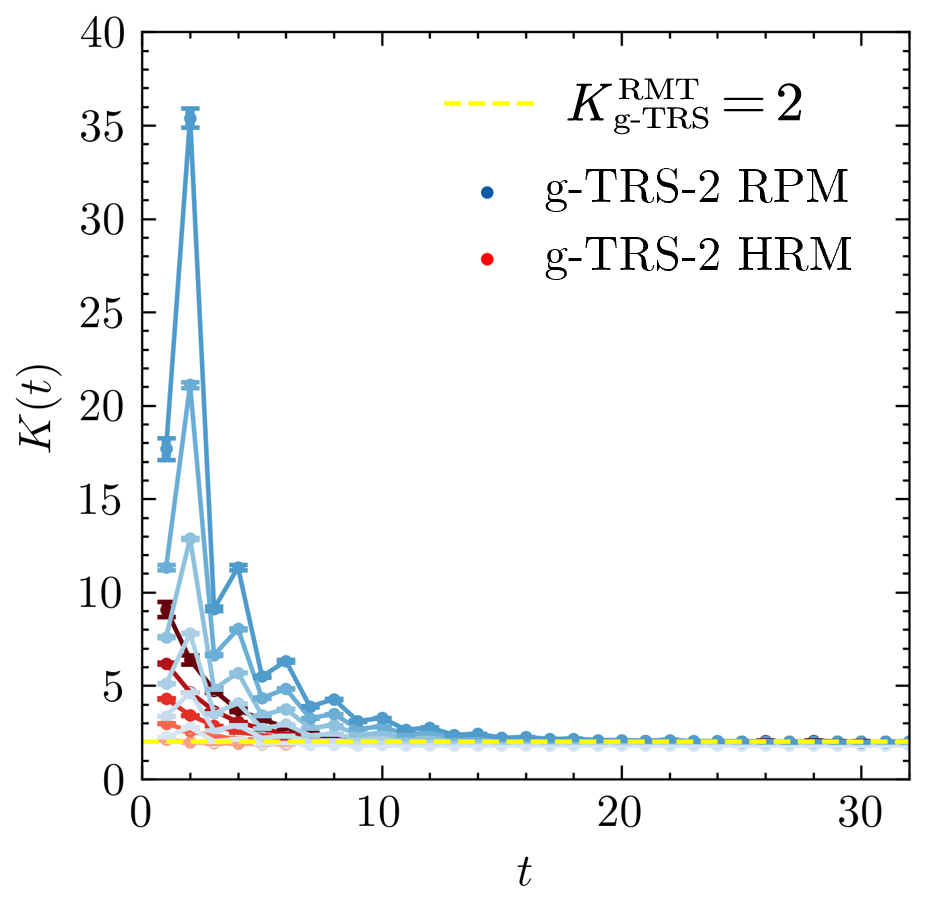}
\end{minipage}
\hspace*{\fill} 
\begin{minipage}[t]{0.48\textwidth}
\includegraphics[width=\linewidth,keepaspectratio=true]{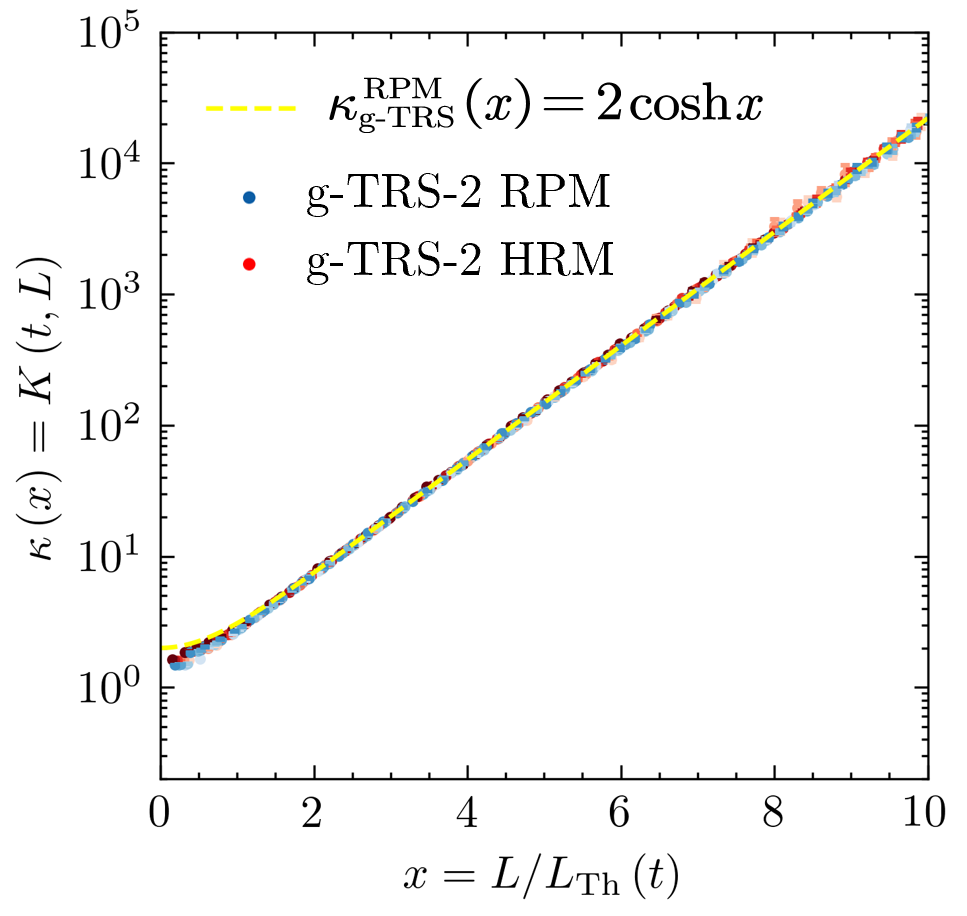}
\end{minipage}
    \caption{\textbf{SFF numerics and scaling collapse for global TRS with local TRS models.} $K(t)$ versus $t$ (left) and $\kappa_{\mathrm{g-TRS}}(x) = K_{\mathrm{g-TRS}}(t,L)$ versus $L/\Lth(t)$ (right)  for the 1D global TRS HRM with local TRS (red) for $q=2$ and $L= 4,6,8,10,12$ in increasingly dark shades, and for the 1D global TRS RPM with local TRS (blue) for $q=3$ and $L= 4,5,6,7,8$ in increasingly dark shades. Both models have pbc.
    For both models, the SFF converges to the RMT behaviour after the Thouless times, which increase with system size, and the scaling collapse shows excellent agreement with the Ising scaling function.
    }
    \label{fig:sff_gtrs2}
\end{figure*}

\begin{figure}[h]
    \centering
    \includegraphics[width=0.5 \textwidth]{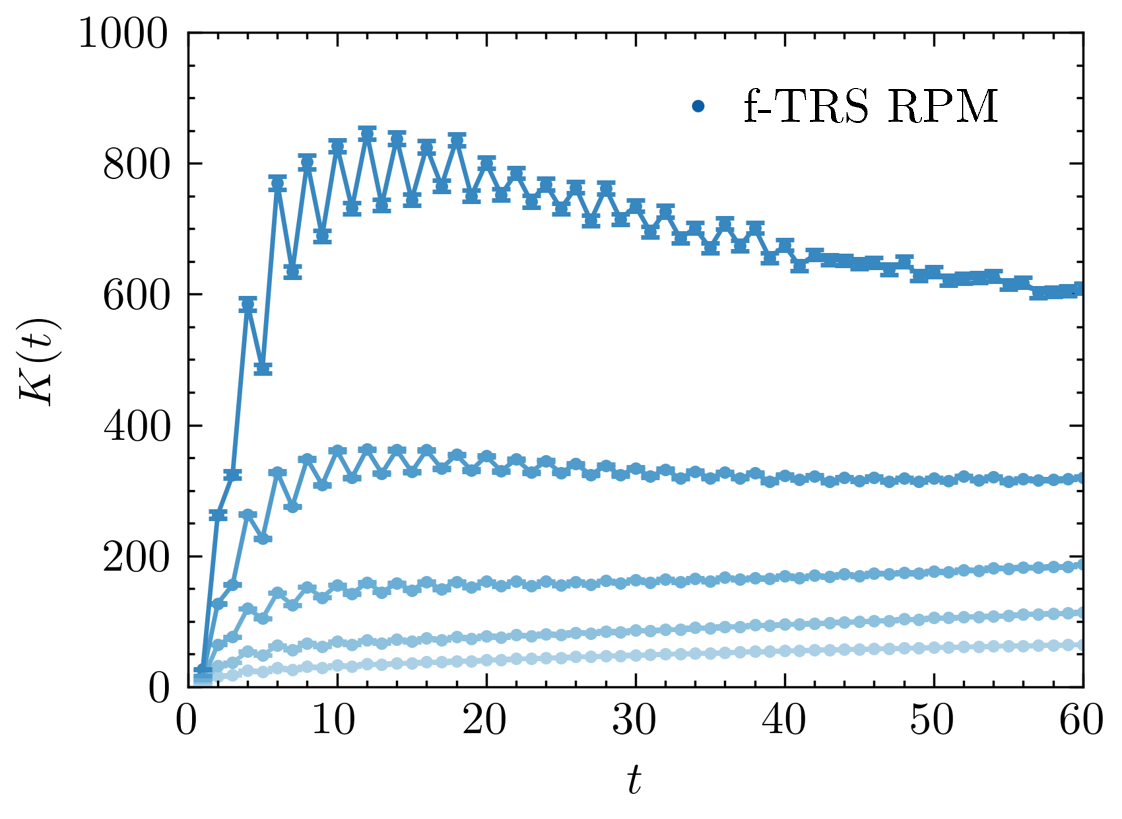}
       \caption{\textbf{Floquet SFF numerics for RPM.} SFF $K(t)$ versus $t$ for the 1D Floquet TRS RPM for $L= 4,5,6,7,8$ in increasingly dark shades. SFF exhibits an initial bump which increases in system size. The bump disappears as the SFF approaches the RMT behaviour at the Thouless times which increases in system size.}
    \label{fig:sff_ftrs_additional}
\end{figure}

\begin{figure}[h]
    \centering
    \includegraphics[width=0.8 \textwidth]{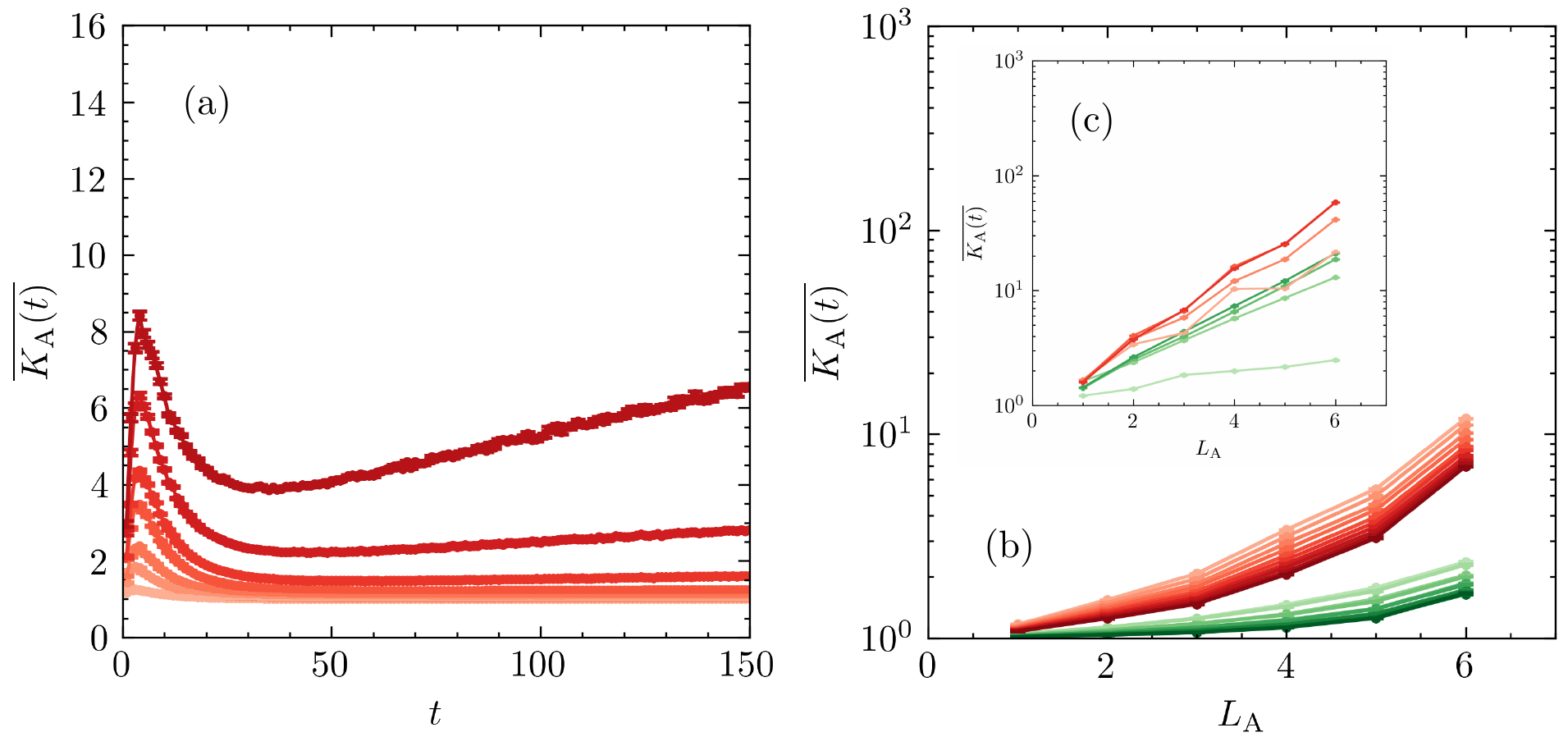}
       \caption{\textbf{PSFF numerics for 3PM.} (a) Numerical simulations of PSFF $\overline{K_A(t)}$ against $t$ for the 1D Floquet TRS 3PM for $q=2$ and $L=10$ with $L_A=1,2,\dots, 7$ in increasingly dark shades. Numerical simulations of PSFF $\overline{K_A(t)}$ against $L_A$ for the 1D Floquet 3PM with TRS (green) and without TRS (red) for $q=2$ and $L=10$ with (b) $t\in [15,24]$ (main panel), and (c) $t\in [1,4]$ (inset) in increasingly dark shades. 
          The 3PM is chosen to have obc with the region $A$ connected to the boundary. 
       }
    \label{app_fig:psff_3pm}
\end{figure}

\begin{figure}[h]
    \centering
    \includegraphics[width=0.5 \textwidth]{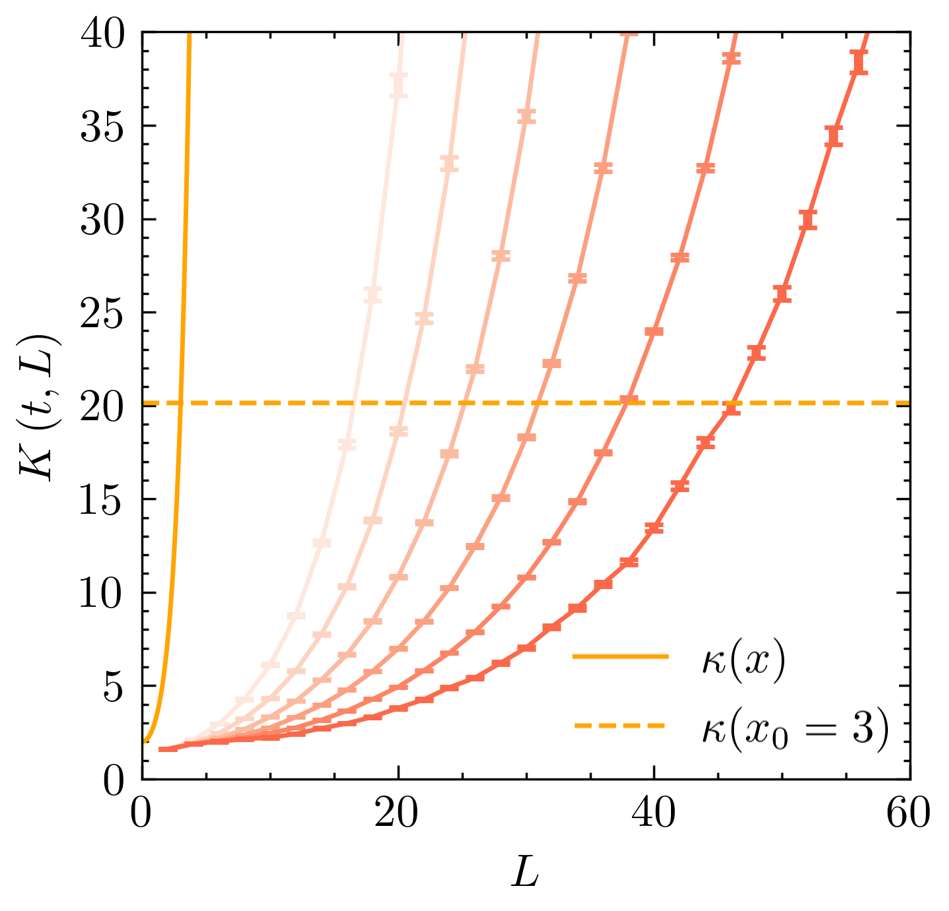}
    \caption{\textbf{Methodology for identification of $\Lth(t)$.}  $K(t, L)$ versus $L$  for global TRS HRM without local TRS and with pbc for different $t=1,2,3,4,5,6$ using space direction simulations. 
    The dashed line is $\kappa(x_0=3)$, and for reference, the scaling function $\kappa(x)$ is plotted in yellow with $\Lth$ taken to be one. The intersections $\tilde{L}(t)$ are used to compute the Thouless length $\Lth(t) = \tilde{L}(t)/x_0$. }
    \label{fig:sff_vs_L}
\end{figure}

\begin{figure*}[ht]
\begin{minipage}[t]{0.3\textwidth}
\includegraphics[width=\linewidth,keepaspectratio=true]{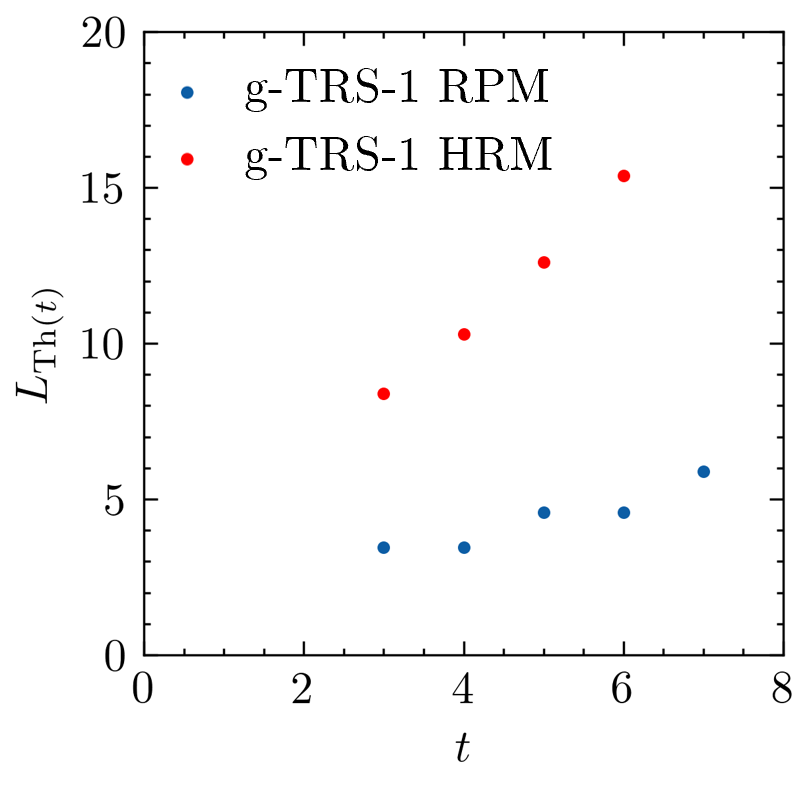}
\end{minipage}
\hspace*{\fill} 
\begin{minipage}[t]{0.3\textwidth}
\includegraphics[width=\linewidth,keepaspectratio=true]{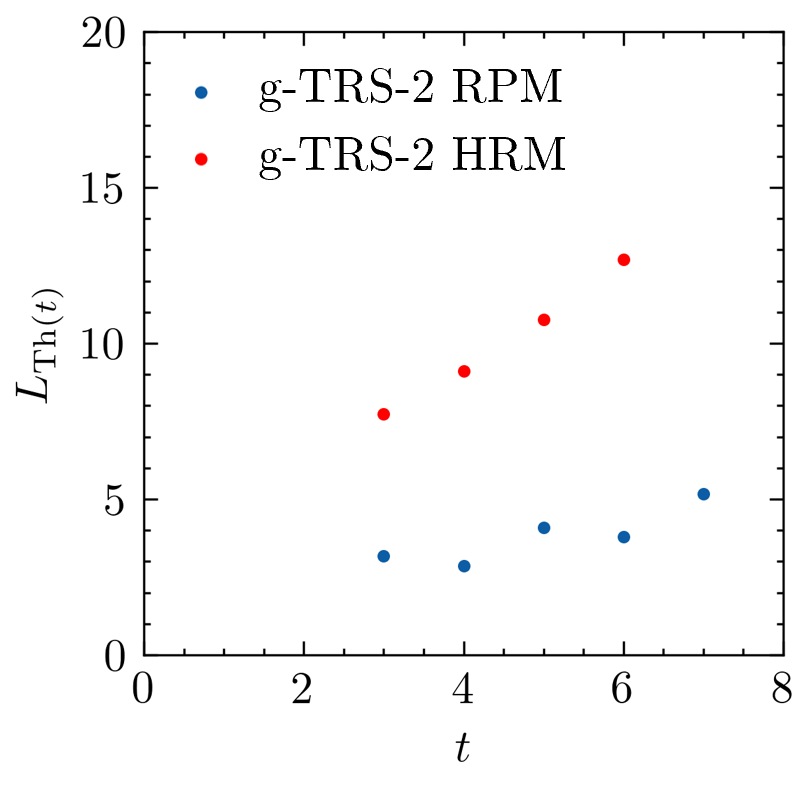}
\end{minipage}
\hspace*{\fill} 
\begin{minipage}[t]{0.3\textwidth}
\includegraphics[width=\linewidth,keepaspectratio=true]{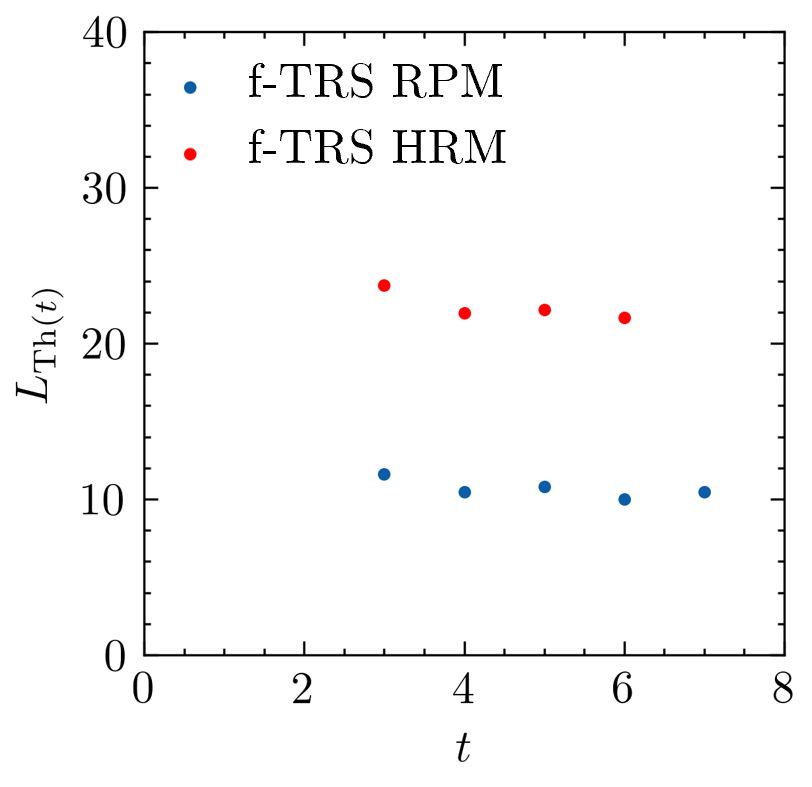}
\end{minipage}
    \caption{\textbf{Thouless length $\Lth(t)$ versus time $t$} for 1D RPM (red) and HRM (blue) in the cases of  global TRS without local TRS (left),  global TRS with local TRS (middle), and Floquet TRS (right). The data is obtained by simulating the quantum circuit in the spatial direction.  \label{fig:sff_lth_vs_t}}
\end{figure*}

\end{document}